\documentclass[aps,pra,10pt,superscriptaddress,floatfix,twocolumn,nofootinbib]{revtex4-2}
\makeatletter \def\switch@array{}\makeatother
\usepackage{graphicx}
\usepackage{amsmath,amsthm,amssymb,dsfont,enumitem}
\usepackage{orcidlink}
\usepackage{appendix}
\newcommand{\ket}[1]{\left| #1 \right\rangle}
\newcommand{\bra}[1]{\left\langle #1 \right|}
\newcommand{\proj}[1]{\ket{#1}\hskip-1mm\bra{#1}}
\newcommand{\av}[1]{\langle #1\rangle}
\newcommand{\braket}[2]{\langle #1|#2 \rangle}

\usepackage{array}
\usepackage{hyperref}
\theoremstyle{definition}
\newtheorem{theorem}{Theorem}
\newtheorem{definition}{Definition}
\newtheorem{proposition}{Proposition}
\newtheorem{lemma}{Lemma}
\newtheorem{corollary}{Corollary}
\newcommand{\id}{{\mathds{1}}}
\usepackage[capitalise]{cleveref}
\crefname{equation}{Eq.}{Eqs.}
\crefname{figure}{Fig.}{Figs.}
\crefname{section}{Sec.}{Secs.}
\crefname{appendix}{Appendix}{Appendices}
\Crefname{appendix}{Appendix}{Appendices}
\crefname{table}{Table}{Tables}
\crefname{theorem}{Theorem}{Theorems}
\crefname{lemma}{Lemma}{Lemmas}
\crefname{proposition}{Proposition}{Propositions}
\crefname{corollary}{Corollary}{Corollaries}
\Crefname{corollary}{Corollary}{Corollaries}
\newcommand{\tr}{\mathrm{Tr}}
\newcommand{\pcrit}{p_{\mathrm{crit}}}
\newcommand{\pW}{p_{\mathrm{W}}}
\newcommand{\pstab}{p_{\mathrm{stab}}}
\newcommand{\pC}{p_{\mathrm{C}}}
\newcommand{\pCL}{p_{\mathrm{CL}}}
\newcommand{\pPM}{p_{\mathrm{PM}}}
\newcommand{\Sd}{\mathcal{S}_d}
\newcommand{\eigmin}{\lambda_{\min}}

\begin{document}

\title{Dissipation of Spekkens' Generalised Contextuality under Decoherence: The Classical Limit}
\author{Enrico Bozzetto}
\affiliation{DET, Politecnico di Torino, Corso Duca degli Abruzzi, 24, 10129 Torino, Italy}
\affiliation{Quantum Group, School of Computing, Newcastle University, 1 Science Square, Newcastle upon Tyne, NE4 5TG, UK}
\author{Jonte R. Hance\,\orcidlink{0000-0001-8587-7618}}
\email{jonte.hance@newcastle.ac.uk}
\affiliation{Quantum Group, School of Computing, Newcastle University, 1 Science Square, Newcastle upon Tyne, NE4 5TG, UK}

\begin{abstract}
When does a decohering quantum system become classical? Generalised noncontextuality gives one criterion, for scenarios rich enough in preparations: a system is classical once its statistics admit an underlying model in which operationally indistinguishable procedures are physically identical. Where it crosses depends not on the state and channel alone but on which preparations, transformations and measurements the experiment offers, and if every measurement is allowed there is no crossing. For any pairwise unbiased bases in any dimension, given preparations holding their eigenstates and surrounding the maximally mixed state, the question reduces to a finite linear program; for a qudit of odd prime dimension $d$ restricted to the stabiliser subtheory, with one depolarised magic state, we solve it in closed form. With preparations and measurements alone the qutrit strange state is classical at zero noise; with the Clifford gates added, each channel represented once, it crosses at $3/5$; and with a random gate represented, as Spekkens' treatment asks, by the matching mixture of stochastic maps, at $3/4$. A Clifford orbit of assignments supports such a representation only if it has $d^2$ members, leaving, in every odd prime dimension, only Gross's discrete Wigner function: the crossing sits where it stops being negative.
\end{abstract}

\maketitle

\section{Introduction}

Contextuality is one of the clearest departures of quantum theory from classical physics. It was first formalised by Bell~\cite{bell1966} and by Kochen and Specker~\cite{KS1967}, and then (at least nominally) generalised by Spekkens~\cite{Spekkens2005GenContext}. Recently it has proven one of the most promising witnesses of quantum advantage~\cite{Raussendorf2013, Howard2014, shahandeh2021quantumcomputationaladvantage, Giordani2023Experimental, Flatt2026QAdvantage}.

The system we follow is an odd-dimensional stabiliser subtheory $\Sd$ with a magic state $\rho_m$ adjoined as an extra preparation of the same single qudit~\cite{bravyi2005}, so that the operational scenario is $\Sd\cup\{\rho_m\}$. The subtheory plus a magic state is Spekkens contextual and universal for quantum computation, while the maximally mixed state, and every other state of the odd-dimensional stabiliser subtheory, is Spekkens noncontextual. We let decoherence act on the magic state alone and ask where along the trajectory the transition happens.

Bell's, Kochen and Specker's, and Spekkens' notions of contextuality are not interchangeable, and in Ref.~\cite{bozzetto2026warringcontextualitiesprovably} we arranged them into a hierarchy: classicality implies Spekkens noncontextuality, which implies Kochen-Specker noncontextuality, which implies Bell locality. None of the converses holds in general; the clearest case is the bottom one, since a single system can be Kochen-Specker contextual while Bell nonlocality needs a multipartite one. The nearest thing to a converse is that any Kochen-Specker contextuality, state-dependent or state-independent, can be converted into a violation of a bipartite Bell inequality by a maximally entangled state of two copies of the system~\cite{cabello2021}, a statement about that new bipartite scenario, with its enlarged measurement set, rather than about the system one started with. Kochen-Specker contextuality is therefore sufficient for nonclassicality, and Spekkens noncontextuality is the level at which a system counts as classical. A system losing coherence should descend this hierarchy, and asking where it crosses each level turns ``when does a quantum system become classical'' into a question with numerical answers.

Which procedures the scenario contains decides the answer, and not by a small amount. Dropping a measurement from a scenario removes Born-rule constraints, and at the same time creates new equivalences between preparations, which are constraints of their own; neither effect dominates in general. The noise level at which contextuality is lost is therefore a property of the scenario as much as of the state and the channel.

One elementary object supplies every threshold below. Give each of the $d{+}1$ stabiliser bases one outcome and add up the probabilities the state assigns to that tuple. \Cref{sec:fragments} turns noncontextuality of a fragment into a condition on these scores, whose forced solutions recover Gross's model at one end and the Margenau-Hill part of the Kirkwood-Dirac distribution at the other, and \cref{sec:assignments} sorts them into a trichotomy over Clifford orbits of assignments.

We compute the threshold for three readings of ``the scenario of $\Sd$ plus the noisy magic state'', set out in \cref{sec:readings}. The first is \emph{prepare-and-measure}: the preparations are the stabiliser states and the noisy magic state, the measurements are the $d{+}1$ stabiliser bases, and nothing else; its threshold is $\pPM$. The second is \emph{the Clifford scenario}: the same, with the Clifford unitaries available as transformations; its threshold is $\pC$ when each transformation equivalence class is represented once, and $\pCL$ when every probabilistic implementation of a channel is represented by the corresponding mixture of stochastic maps, as Spekkens' own treatment asks. By \cref{thm:trichotomy}, $\pC$ is also the ordinary noncontextuality threshold of a scenario with no transformation procedures in it at all, the prepare-and-measure fragment whose preparations are closed under the Clifford group; that is the form in which we take it. The third and final is the \emph{exact-frame} reading, which asks for a single non-negative representation of the whole subtheory at once, its states, effects, channels and their composites together, in the sense that Ref.~\cite{Schmid2024structuretheorem}'s structure theorem attaches to a noncontextual model, with Ref.~\cite{wagner2025structuretheoremcomplexvaluedquasiprobability} supplying the same canonical form for complex-valued frames; its threshold is $\pcrit$. The four satisfy $\pPM\le\pC\le\pCL\le\pcrit$, and for the qutrit strange state under depolarising noise they are $0$, $3/5$, $3/4$ and $3/4$. Throughout, the measurements are those of the subtheory, which \cref{sec:allmeas} shows keeps the question non-trivial: if we allow every POVM, then the scenario is contextual at every noise level, with no magic state needed.

The last two coincide. \Cref{thm:convexlinear} shows that a Clifford orbit of value assignments admits a convex-linear representation of the transformations only if the orbit has $d^2$ members, in every odd prime dimension; the orbit of phase-space points has that size, and no other orbit does, which leaves Gross's representation and so $\pCL=\pcrit=\pW$ for every magic state, where $\pW$ is the Wigner-positivity point, the noise level at which the state's Wigner function stops taking negative values. The widest separation is the one between $\pPM$ and the coinciding $\pCL$ and $\pcrit$, $0$ against $3/4$.

The $0$ at the other end is the same mechanism with its sign reversed. Dropping the Clifford transformations removes the operational equivalences that made the strange state contextual, and its bare four-basis prepare-and-measure fragment then has an explicit noncontextual model at zero noise. Removing procedures does not always help: \cref{sec:twoqutrit} gives a two-qutrit scenario with two measurements, stabiliser preparations and no magic state at all which is contextual. So $\pPM$ is not monotone in the measurement set, and a threshold of $0$ is a statement about a fragment rather than about a state.

Two hypotheses sit behind reading Spekkens noncontextuality as classicality. The first is needed for the implication to Kochen-Specker noncontextuality. Spekkens noncontextuality implies Kochen-Specker noncontextuality once outcome determinism for projective measurements has been derived, and within quantum theory that follows from preparation noncontextuality when the scenario contains the eigenstates of each measurement as preparations, with $\id/d$ their uniform mixture~\cite{Spekkens2005GenContext}; Ref.~\cite{kunjwalspekkens2015} derives noncontextuality inequalities from uncolourability proofs of the Kochen-Specker theorem on the same basis, without assuming determinism. For a general generalised probabilistic theory (GPT), or a fragment without those preparations, a Spekkens-noncontextual model can reproduce Kochen-Specker-contextual statistics: a fragment with a single preparation is Spekkens noncontextual whatever its statistics. Every scenario in the main text contains the eigenstates of each of its measurements among its preparations, so the hypothesis holds, and the proof of \cref{thm:convexlinear} derives outcome determinism from it. The second concerns the other end of the hierarchy, which is asymmetric: a classical system must be Spekkens noncontextual, but the converse fails for the narrow notion of classicality defined in Ref.~\cite{bozzetto2026warringcontextualitiesprovably}, with Gaussian quantum mechanics as the counterexample, so we read Spekkens noncontextuality as certifying classicality in a broader sense. Under those two hypotheses, a Spekkens noncontextual scenario is provably classical and a Kochen-Specker contextual one provably nonclassical, with a gap between them in which neither verdict is available.

Most analyses of Spekkens contextuality have been carried out in the abstract. Mazurek \textit{et al.}~\cite{mazurek2016} show how to test noncontextuality without assuming that any laboratory procedure is ideal. Baldij\~ao \textit{et al.}~\cite{baldijao2021} show that, in a quantum-Darwinism setting, Spekkens noncontextuality emerges for the observers' scenarios once the environment encodes the proliferated information well enough, a sufficient condition on the encoding rather than a threshold in a noise parameter of the system's own trajectory. Dutra \textit{et al.}~\cite{dutra2026} put the same question to Kochen-Specker contextuality, sending depolarising noise through sequential KCBS and Peres-Mermin scenarios both before and between the measurements and finding the noise level at which the inequalities stop being violated; their crossing point is that of a fixed inequality in a sequential scenario, where ours is that of the existence of a generalised-noncontextual model for a whole subtheory. Closer to the present question, Aloy \textit{et al.}~\cite{aloy2026} perform theory-independent process tomography on a superconducting qubit, reconstruct its realised GPT state and effect spaces as a function of wait time, and test simplex-embeddability at each. They find the system contextual at $0$, $5$ and $10\,\mu$s and noncontextual from $15\,\mu$s onwards, so the transition is bracketed by the times they sampled. Their scenario is whatever the device supplies, which is the right choice for a device-independent claim but means the scenario cannot be varied. We do the opposite: for a structured subtheory the thresholds can be computed in closed form, and the dependence of the answer on the scenario becomes something to calculate rather than something fixed by the apparatus.

\Cref{sec:reconcile} sets the three readings against the uniqueness theorem of Schmid \textit{et al.}~\cite{Schmid2022StabSub}, whose third clause says that a state promoting the subtheory to universal quantum computation leaves it with no generalised-noncontextual model, and says which hypothesis each reading drops. The remaining sections cover qubits, evaluate the closed forms, examine other noise models, and finish with two qutrits.
\section{Generalised contextuality}
\label{sec:gencontext}

Spekkens' framework~\cite{Spekkens2005GenContext} generalises the Kochen-Specker~\cite{KS1967} definition of contextuality by shifting from a logic of projectors to an operational representation of experiments. In this model, the primitive elements are experimental procedures: preparations ($P$), transformations ($T$), and measurements ($M$). 

An ontological model of an operational theory maps the procedures of a scenario to an underlying reality whose attributes hold regardless of what anyone knows about them. It has three elements: a space of ontic states $\lambda\in\Lambda$; a map from each quantum state $\rho$ to a probability measure $\mu_\rho$ over $\Lambda$; and a map from each quantum observable to a Markov kernel from $\Lambda$ to its outcomes.\footnote{A Markov kernel from $(X,\Sigma_X)$ to $(Y,\Sigma_Y)$ is a map that assigns to every $x \in X$ a probability measure $K(x,\cdot)$ on $(Y,\Sigma_Y)$ such that for every measurable set $A \in \Sigma_Y$, the function $x \mapsto K(x,A)$ is $\Sigma_X$-measurable.} The model must reproduce the statistical predictions of the operational theory we are considering \cite{Spekkens2005GenContext}.
The probability of an outcome is therefore given by the law of total probability:
\begin{equation}
 p(k|P,T,M) = \int_\Lambda d\lambda d\lambda'\, \mu_P(\lambda) \Gamma_T(\lambda',\lambda) \xi_{M,k}(\lambda')
 \label{eq:ontological}
\end{equation}
where $\mu_P(\lambda):\Lambda\rightarrow[0,1]$ is the probability distribution associated with the preparation $P$, $\Gamma_T(\lambda',\lambda):\Lambda\times \Lambda\rightarrow[0,1]$ is the transformation matrix which represents the probability of moving from the ontic state $\lambda$ to the state $\lambda'$ and $\xi_{M,k}(\lambda'):\Lambda\rightarrow[0,1]$ is the probability of obtaining the outcome $k$ performing the measurement $M$, given the ontic state $\lambda'$.

Central to the definition of Spekkens contextuality is the concept of operational equivalence. Two preparation procedures $P$ and $P'$ are operationally equivalent if they yield identical statistics for all possible measurements. In the ontological model this can be expressed as
\begin{equation}
\label{eq:prepequiv}
\begin{aligned}
 P\sim P' \leftrightarrow & \int_\Lambda \xi_{M,k}(\lambda)\,\mu_{P}(\lambda)d\lambda
 =\\&
 \int_\Lambda \xi_{M,k}(\lambda)\,\mu_{P'}(\lambda)d\lambda
 \quad \forall\,M,k
\end{aligned}
\end{equation}
Similarly, two measurements $M$ and $M'$ are equivalent if they yield the same statistics for all possible preparations:
\begin{equation}
\label{eq:measequiv}
\begin{aligned}
 M\sim M' \leftrightarrow &\int_\Lambda \xi_{M,k}(\lambda)\,\mu_P(\lambda)d\lambda
 =\\&
 \int_\Lambda \xi_{M',k}(\lambda)\,\mu_P(\lambda)d\lambda
 \quad \forall\,P, k
 \end{aligned}
\end{equation}
Finally, two transformations are equivalent if they yield the same statistics for all possible preparations and measurements:
\begin{equation}
\label{eq:transfequiv}
\begin{aligned}
 &T\sim T' \leftrightarrow\\
 &\int_\Lambda d\lambda d\lambda'\, \mu_P(\lambda) \Gamma_{T}(\lambda',\lambda) \xi_{M,k}(\lambda')=\\&\int_\Lambda d\lambda d\lambda'\, \mu_P(\lambda) \Gamma_{T'}(\lambda',\lambda)\xi_{M,k}(\lambda')\quad \forall P,M,k
\end{aligned}
\end{equation}

Read as an assumption about the world rather than a definition, Spekkens noncontextuality asks for Leibnizianity (operations with the same operational statistics have the same underlying ontology; the term is due to Schmid, Selby and Spekkens~\cite{schmidselbyspekkens2020}, who take the underlying principle from Spekkens~\cite{spekkens2019leibniz}), linearity (the expectation value of a sum of variables is the sum of their expectation values), and diagram-preservation (an operational representation of a scenario preserves the structure of the underlying physical processes)~\cite{schmidselbyspekkens2023, zhang2025quantifierswitnessesnonclassicalitymeasurements, zhang2026reassessingboundaryclassicalnonclassical, Schmid2024structuretheorem}.

An equivalence class of procedures is a set of equivalent procedures; a context is the set of features not specified by the equivalence class. An ontological model is noncontextual if every procedure depends only on its equivalence class, not on its context: operationally equivalent procedures have identical ontological representations,
\begin{equation}
\label{eq:spekkensNC}
\begin{split}
 &P\sim P'\Rightarrow\mu_P=\mu_{P'},\\& T\sim T'\Rightarrow \Gamma_T=\Gamma_{T'},\\& M\sim M'\Rightarrow\xi_{M,k}=\xi_{M',k}
 \end{split}
\end{equation}
In particular a system is (Spekkens) preparation, transformation or measurement contextual when the first, second or third condition does not hold. A model in which all three conditions hold is what Spekkens calls \emph{universally noncontextual}~\cite{Spekkens2005GenContext}; a system is (Spekkens) contextual when no such model exists, i.e., when at least one condition fails in every model.

The quantifiers in \crefrange{eq:prepequiv}{eq:transfequiv} run over the procedures the scenario contains, so operational equivalence is a scenario-relative notion, which \cref{sec:readings} makes the organising point of the paper. Ref.~\cite{gittonwoods2024} makes the same point in general, defining noncontextuality relative to an explicit reference of indistinguishability and concluding that whether an experiment is noncontextual ``is not as absolute as often perceived''; the three readings of \cref{sec:readings} are three such references for one subtheory, and each gives a different threshold.

Two results connect this framework to quasiprobability. First, a scenario is generalised noncontextual if and only if the GPT it defines, obtained by quotienting its procedures by the operational equivalences of \crefrange{eq:prepequiv}{eq:transfequiv}, admits an ontological model: Ref.~\cite{Schmid2021}'s Proposition~1 for prepare-and-measure experiments, extended to general compositional scenarios, and so to the transformation equivalences of \cref{eq:transfequiv}, by Ref.~\cite{Schmid2024structuretheorem}'s Proposition~3.2, whose Corollary~3.5 also makes it equivalent to a non-negative, diagram-preserving quasiprobability representation. Which GPT this is depends on the scenario and not only on the quantum system: the stabiliser states and the stabiliser effects each span the Hermitian operators on $\mathcal{H}_d$, so the full subtheory defines a GPT on the $d^2$-dimensional space of the qudit, while a fragment with $k<d+1$ bases defines one on the $(k(d-1)+1)$-dimensional span $V$ of its effects (\cref{sec:fragments}). The extension needed to cover complex-valued representations such as Kirkwood-Dirac is due to Wagner \textit{et al.}~\cite{wagner2025structuretheoremcomplexvaluedquasiprobability}: their Theorem III.1 puts any linearity-preserving, empirically adequate functor out of a tomographically local, finite-dimensional GPT into the canonical form $M(\tilde T)=\chi_B\circ \tilde T\circ \chi_A^{-1}$, with each $\chi_A$ an invertible complex-linear map fixed by the action on states, and their Corollary III.2 adds that $\dim_\mathbb{C}M(A)=\dim_\mathbb{C}\mathbb{C}(A)$, so that choosing such a representation is choosing an exact (non-over-complete) complex frame $\{F_\lambda\}$ with dual $\{D_\lambda\}$; their Theorem III.3 extends this to semifunctors, where $\chi_A$ is only injective and the target space may be infinite-dimensional. It follows that a theory is contextual when no exact complex frame represents it without negativity or imaginarity, which is the form we use in \cref{sec:readings}. Second, negativity of one chosen representation does not by itself establish contextuality~\cite{Spekkens2008Negativity, Ferrie2011Quasiprobabilities}: contextuality is equivalent to \emph{every} admissible representation being negative. \Cref{app:qprep} defines frames and dual frames, and collects the Wigner and Kirkwood-Dirac representations we use.
\section{The setting}
\label{sec:setting}

Throughout, $d$ is an odd prime unless stated otherwise, and $\omega=e^{2\pi i/d}$.

\subsection{The stabiliser subtheory}
\label{sec:stabsub}

The stabiliser subtheory is built from Clifford unitaries, themselves built from Weyl operators. If we consider a $d$-dimensional Hilbert space, single-system Weyl operators are
\begin{equation}
 W_{a,b}=Z^aX^b
\end{equation}
where for computational basis states $\ket{x}$, $X\ket{x}=\ket{x+1}$, and $Z\ket{x}=\omega^x\ket{x}$ with $a,b\in \mathbb{Z}_d$.
Clifford unitaries map Weyl operators to Weyl operators under conjugation, up to a phase; the Clifford group $\mathcal{C}_d$ is the unitary normaliser of the Pauli group, generated by the Fourier gate, the quadratic phase gate, and the Weyls. Following Ref.~\cite{Schmid2022StabSub}, the single-qudit stabiliser subtheory $\Sd$ is generated by sequential composition of the $d(d{+}1)$ pure stabiliser states (simultaneous eigenstates of a set of Weyl operators), the $d{+}1$ stabiliser measurements (the eigenbases of the Weyl classes, which form a complete set of mutually unbiased bases (MUBs)~\cite{gibbons2004}, the unbiasedness following from the partition of the traceless Weyl operators into $d{+}1$ commuting classes~\cite{bandyopadhyay2002}) and the Clifford unitary super-operators, together with convex mixtures of these; it is odd-dimensional when $d$ is odd. We write $\Pi^m_i$ for the projector onto outcome $i$ of basis $m\in\{0,\dots,d\}$, and $s_{m,i}$ for the corresponding pure state; distinct stabiliser bases are mutually unbiased, $\tr(\Pi^m_i\Pi^{m'}_j)=1/d$ for $m\ne m'$.

In this context, Ref.~\cite{Schmid2022StabSub} proves the following, which fixes the exact-frame level of the hierarchy for us. Clause 3 below does not forbid the single-qudit models of \cref{sec:disagree}; \cref{sec:reconcile} sets out the difference.
\begin{theorem}[Ref.~\cite{Schmid2022StabSub}, Theorems 1(a), 1(b), 2 and 3]
\label{thm:grossuniq}
\leavevmode
\begin{enumerate}
 \item The unique non-negative and diagram-preserving quasiprobability representation for any stabiliser subtheory (single- or multi-particle) in odd dimensions is Gross's representation.
\item For any stabiliser
subtheory (single- or multi-particle) in even dimensions, there is no non-negative and diagram-preserving quasiprobability representation.
\item Consider any state $\rho$ which promotes the
stabiliser subtheory to universal quantum computation.
There is no generalised noncontextual model for the
stabiliser subtheory together with $\rho$.
\item A (necessary and) sufficient condition for any unitary or pure state to promote the stabiliser subtheory to universal quantum computation is that it be negatively represented in Gross's representation.
\end{enumerate}
\end{theorem}

\subsection{Gross's Wigner function and the noisy magic state}

Gross's discrete Wigner function \cite{gross2006,wootters1987,gibbons2004} is built from the phase-point operators $A_u=\mathcal{D}_u A_0 \mathcal{D}_u^\dagger$, $A_0=\frac1d\sum_u \mathcal{D}_u$, where $\mathcal{D}_{(a_1,a_2)}=\omega^{-a_1a_2/2}Z^{a_1}X^{a_2}$ (with $1/2$ the inverse of $2$ mod $d$) and $u$ ranges over the $d^2$ points of the phase space $\mathbb{Z}_d^2$:
\begin{equation}
 W_\rho(u)=\tfrac1d\tr(A_u\rho).
\end{equation}
Each $A_u$ satisfies $A_u^2=\id$, $\tr A_u=1$, so its spectrum is $+1$ (multiplicity $(d{+}1)/2$) and $-1$ (multiplicity $(d{-}1)/2$). Pure stabiliser states have $W$ equal to $1/d$ on a phase-space line and zero elsewhere. The \emph{mana} of a state is $\mathcal{M}(\rho)=\ln\sum_u|W_\rho(u)|$~\cite{veitch2014}.

Let a magic state $\rho_m$ (any non-stabiliser pure state) be adjoined to $\Sd$ and subjected to a depolarising channel $\mathcal{E}_p(\rho)=(1-p)\rho+p\,\id/d$ parametrised by a noise parameter $p\in[0,1]$, so
\begin{equation}
 \rho_M(p) = \mathcal{E}_p(\rho_m) = (1-p)\rho_m + p \frac{\id}{d},
\end{equation}
and $W_{\rho_M(p)}(u)=(1-p)W_{\rho_m}(u)+p/d^2$. For every $p>0$ the state is mixed, so clause 4 of \cref{thm:grossuniq}, stated for pure states, no longer applies. Wigner positivity begins at
\begin{equation}\label{eq:pW}
 \pW=\frac{d^2|W_{\min}|}{1+d^2|W_{\min}|},\qquad W_{\min}=\min_u W_{\rho_m}(u),
\end{equation}
and, since every Wigner-positive state is classically simulable with stabiliser operations \cite{Veitch2012Resource}, $p\ge\pW$ marks the loss of the state's known route to quantum advantage. \Cref{eq:pW} has the algebraic form of van Dam and Howard's depolarising threshold~\cite{vandam2011}, but with a different negativity. Their $|N(\rho)|$ is the magnitude of the minimum of $\tr(\rho A(u))$ over all $d^{\,d+1}$ \emph{generalised} phase-point operators, so their threshold is the stabiliser-polytope entry point $\pstab$, which equals $\pW$ only when the collinear clause is the maximal one (\cref{thm:trichotomy}).

Two geometric quantities track the trajectory: \cref{eq:pW}'s $\pW$, and the \emph{stabiliser-polytope entry point} $\pstab$, the smallest $p$ with $\rho_M(p)\in\mathcal{P}_{\mathrm{stab}}:=\mathrm{conv}\{s_{m,i}\}$. Since stabiliser states are Wigner-non-negative, $\pW\le\pstab$, with the inequality strict for some states (\cref{sec:assignments}). The set of states represented positively by the Gross-Wigner frame is the Wigner polytope~\cite{Veitch2012Resource}; it contains the stabiliser polytope, and the inclusion is known to be strict~\cite{gross2006,Veitch2012Resource}, the gap being \cref{sec:boundmagic}'s bound magic. Each is a convex body in the space of density matrices whose centre is the maximally mixed state, so each contains a ball of some radius $\epsilon>0$ about it, and a transition at finite noise is guaranteed to exist. That every quantum system becomes preparation noncontextual under sufficiently strong noise was shown by Marvian~\cite{marvian2020}, who also fixes the crossing point for the full qubit with all states and all measurements, under the same depolarising channel, at $p=1/2$. Thresholds of this kind have since been computed for other fragments, as robustness of contextuality to depolarising and to dephasing noise in parity-oblivious multiplexing scenarios~\cite{fonseca2025}, and as the channel parameters at which a preparation-noncontextuality inequality stops being violated~\cite{kumari2026}. We add the crossing point for a structured subtheory in odd prime dimension, and its dependence on the scenario.

\subsection{Scores of value assignments}

All the thresholds below are controlled by one family of functionals. A \emph{value assignment} $\lambda$ is a choice of one outcome for each of the $d+1$ stabiliser bases, and its score on a state is
\begin{equation}\label{eq:score}
 s_\lambda(\rho) \;=\; \sum_{m=0}^{d} p(\lambda_m|m,\rho),
\end{equation}
the total Born probability the state assigns to that tuple of outcomes. The condition $s_\lambda(\rho)\ge 1$, imposed for various subsets of assignments, identifies Wigner positivity (assignments consistent with a phase-space point), membership of the stabiliser polytope (all assignments), and the existence of noncontextual ontological models (Clifford orbits of assignments). Its restriction to $k$ bases, $s_\lambda(\rho)\ge(k-1)/d$, identifies for $k=2$ the positivity of the real part of the Kirkwood-Dirac distribution of the $V$-projected state; the raw distribution has a different positivity point, which can lie on either side (\cref{sec:fragments}). \Cref{sec:fragments,sec:assignments} derive these statements.

\section{Three readings of the scenario}\label{sec:readings}

An operational scenario is a collection of preparations, transformations and measurements together with their Born-rule statistics, and an ontological model of it is as in \cref{sec:gencontext}. We consider three readings of ``the scenario of $\Sd$ plus $\rho_M(p)$'':

\begin{enumerate}[label=(\roman*)]
 \item \emph{Exact frames:} following Ref.~\cite{Schmid2024structuretheorem}'s structure theorem, which ties noncontextuality to the existence of a non-negative, diagram-preserving quasiprobability representation, and Ref.~\cite{wagner2025structuretheoremcomplexvaluedquasiprobability}'s extension of its canonical form to complex-valued frames, one reading asks for such a representation of $\Sd\cup\{\rho_M(p)\}$: a frame $\{F_\lambda\}$ and dual $\{D_\lambda\}$ for $\mathcal{L}(\mathcal{H}_d)$, that is, a spanning family of operators together with the family biorthogonal to it, $\tr(F_\lambda D_{\lambda'})=\delta_{\lambda\lambda'}$, so that $M=\sum_\lambda\tr(F_\lambda M)\,D_\lambda$ for every operator $M$ (\cref{app:qprep} gives the general definition), with $|\Lambda|=d^2$, representing all states, effects and channels of the subtheory, plus the noisy magic state, by non-negative reals. We write
 \begin{equation} \label{eq:p_crit_def}
 \pcrit = \inf\{p\in[0,1] \mid \text{such a frame exists}\}
 \end{equation}
 for the smallest such $p$, which \cref{thm:pcrit} shows is attained. Explicitly, a frame $\{F_\lambda\}$~\cite{Christensen2003Frames} represents the noisy magic state by $(1-p)\tr(F_\lambda\rho_m)+\tfrac{p}{d}\tr(F_\lambda)$, and writing $\mathcal{W}(M)=\sum_{i,j}|\mathrm{Im}(M_{ij})|+\sum_{i,j}|\min(0,\mathrm{Re}(M_{ij}))|$ for the nonclassicality penalty of a represented process, and $\Omega(p,f)$ for the maximum of $\mathcal{W}$ over the states, transformations and effects of $\Sd\cup\{\rho_M(p)\}$ at frame $f$, \cref{eq:p_crit_def} becomes $\pcrit=\inf\{p\in[0,1]\mid\Omega(p,f)=0$ for some $f\in\mathcal{F}\}$, with $\mathcal{F}$ the exact frame/dual-frame pairs for $\mathcal{L}(\mathcal{H}_d)$. The maximum is attained, since $\mathcal{W}$ is convex and the preparations and Clifford mixtures are polytopes, so it suffices to test the stabiliser states, $\rho_M(p)$, the Cliffords and the stabiliser effects.\label{ScenarioReadingiExactFrame}

 \item \emph{The Clifford scenario:} the operational scenario whose preparations are the stabiliser states (and mixtures) together with $\rho_M(p)$, whose transformations are the Clifford unitaries and their mixtures, and whose measurements are the $d{+}1$ stabiliser bases. Two thresholds go with it, according to how a transformation procedure is represented: $\pCL$ under Spekkens' convex-linearity assumption, and $\pC$ in the weaker class-representative reading of \cref{def:classrep} below.\label{ScenarioReadingiiClifford}

 \item \emph{The prepare-and-measure fragment:} \ref{ScenarioReadingiiClifford} with the transformations omitted, a bare prepare-and-measure fragment \cite{Schmid2021,selby2024}. We write this scenario's noncontextuality threshold $\pPM$.\label{ScenarioReadingiiiPrepMeas}
\end{enumerate}

\begin{definition}[Class-representative reading]\label{def:classrep}
A \emph{class-representative model} of a prepare-transform-measure scenario satisfies \cref{eq:spekkensNC} for preparations and measurements, and assigns one stochastic map $\Gamma$ to each transformation \emph{channel} of the scenario, required to reproduce the operational statistics on the scenario's preparations. It does not require that a procedure implementing a channel as a probabilistic mixture of other transformations be represented by the corresponding mixture of those transformations' stochastic maps.
\end{definition}

By \cref{thm:trichotomy} and its proof in \cref{app:ray}, $\pC$ has a second description containing no transformation procedures: it is the ordinary generalised-noncontextuality threshold of the prepare-and-measure fragment whose preparations are the stabiliser polytope together with the Clifford orbit of $\rho_M(p)$, a scenario in which \cref{def:classrep} and \cref{eq:spekkensNC} coincide because there are no transformation procedures to represent. So $\pC$ can be read off a scenario that needs no weakening of Spekkens' notion; the weakening is what the Clifford scenario must adopt in order to agree with that fragment.

\Cref{def:classrep} is a deliberate weakening of Spekkens' own treatment of transformation noncontextuality, which additionally assumes that a convex sum of transformation procedures is represented by the convex sum of the associated transition matrices~\cite{Spekkens2005GenContext}, and which Ref.~\cite{lillystone2019} describes as the standard, often implicit, assumption. The non-collinear models of \cref{sec:disagree} use this weakening. The completely depolarising channel is both the uniform mixture of the $d^2$ Weyl unitaries and the uniform mixture of the whole Clifford group: one channel, two laboratory procedures, operationally equivalent by \cref{eq:transfequiv} and both present in this scenario. In the non-collinear model of \cref{thm:trichotomy} at $d=3$ the nine Weyl translations act freely on the $72$ non-collinear assignments, in eight orbits of nine, while $\mathrm{ASL}(2,3)$ is transitive on all $72$, so representing a mixture by the corresponding mixture of stochastic maps would send a point mass to the uniform distribution on nine assignments in the first case and on seventy-two in the second. 

Refs.~\cite{schmid2024ptm,yianni2026} give general methods for prepare-transform-measure scenarios of this shape, Ref.~\cite{schmid2024ptm} a linear program returning necessary conditions for a noncontextual model and Ref.~\cite{yianni2026} a full decision procedure, and both impose the convex-linear reading throughout: Ref.~\cite{schmid2024ptm} as the constraint $\sum_T\alpha_T\Gamma_T(\lambda'|\lambda)=0$ imposed by every linear identity $\sum_T\alpha_TT=0$ among transformation procedures, and Ref.~\cite{yianni2026} as the equality of the rank of the stacked stochastic maps with the rank of the transformation matricisation of the statistics, which is the same demand in another form. Neither implements \cref{def:classrep}, so by \cref{thm:convexlinear} both return $3/4$ for the qutrit Clifford scenario rather than $3/5$. Ref.~\cite{schmid2024ptm}'s constraint makes the difference: it confines the support to the collinear orbit, which the strange state admits from $p=3/4$. We implemented their Formulation F1 for the qutrit Clifford scenario with no appeal to \cref{thm:convexlinear}, and it returns $3/4$, with an exact Farkas certificate below the crossing (\cref{app:numerics}); since F1 returns necessary conditions, this independently establishes that no convex-linear model exists below $3/4$ at $d=3$, the matching model above being Gross's. Dropping their transformation constraint does not return $\pC$: it removes convex-linearity and with it any requirement that the support be closed under the Clifford orbit, which \cref{def:classrep} keeps, and it is the loss of that closure that leaves the program feasible at every $p$. Formulation F1 with the constraint dropped is therefore not a program for \cref{def:classrep}. Ref.~\cite{yianni2026} recovers the transformation contextuality of the single-qubit stabiliser subtheory~\cite{lillystone2019} by the same route, which is \cref{cor:orbitcriterion}'s qubit instance. Because \cref{def:classrep} keeps that orbit closure, \cref{thm:trichotomy}'s finite orbit disjunction is available for $\pC$.

Reading \ref{ScenarioReadingiExactFrame} gives a convex-linear model of \ref{ScenarioReadingiiClifford}, a convex-linear model is in particular a class-representative one, and \ref{ScenarioReadingiiClifford} restricts to \ref{ScenarioReadingiiiPrepMeas}, so, with \cref{thm:pcrit,thm:convexlinear},
\begin{equation}\label{eq:order}
 \pPM \;\le\; \pC \;\le\; \pCL \;\le\; \pcrit \;=\; \pW.
\end{equation}

\begin{table*}[tbp]
\caption{\label{tab:symbols}The six noise thresholds used in this paper, with the values they take for the qutrit strange state $\ket a$ of \cref{tab:six} under depolarising noise. The first four are noncontextuality thresholds of the three readings of \cref{sec:readings}, $\pC$ and $\pCL$ both belonging to the second reading, under the class-representative and the convex-linear treatment of transformation procedures respectively; $\pW$ and $\pstab$ are defined geometrically rather than as noncontextuality thresholds, though by \cref{thm:pcrit} $\pW$ coincides with $\pcrit$, and by \cref{thm:convexlinear,prop:orbituniq} with $\pCL$ as well in every odd prime dimension.}
\begin{ruledtabular}
\begin{tabular}{p{0.2\linewidth}|p{0.5\linewidth}|p{0.1\linewidth}|p{0.2\linewidth}}
symbol & what it is the threshold for & defined & strange state, $d{=}3$\\
\hline
$\pPM$ & prepare-and-measure fragment, no transformations & \ref{ScenarioReadingiiiPrepMeas} & $0$\\
$\pC$ & Clifford scenario, class-representative (\cref{def:classrep}) & \ref{ScenarioReadingiiClifford} & $3/5$\\
$\pCL$ & Clifford scenario, convex-linear & \ref{ScenarioReadingiiClifford} & $3/4$\\
$\pcrit$ & exact frame for the whole subtheory & \ref{ScenarioReadingiExactFrame} & $3/4$\\
$\pW$ & Wigner function stops being negative & \cref{eq:pW} & $3/4$\\
$\pstab$ & state enters the stabiliser polytope & \cref{sec:setting} & $3/4$\\
\end{tabular}
\end{ruledtabular}
\end{table*}

\Cref{tab:symbols} collects the six symbols. The last inequality is an equality, so that the top two thresholds collapse to one, in every odd prime dimension (\cref{thm:convexlinear,prop:orbituniq}). The first inequality of \cref{eq:order} can be strict, and so can the second: for the qutrit strange state, $(\pPM,\pC,\pCL)=(0,\,3/5,\,3/4)$.

\subsection{Why the measurements must be restricted}\label{sec:allmeas}

All three readings take the measurements to be those of the subtheory. Without that restriction
there is nothing to compute.

\begin{proposition}
Let the preparations include the stabiliser polytope and let every POVM on $\mathcal{H}_d$ ($d\ge2$) be an allowed measurement. Then the scenario admits no generalised-noncontextual ontological model, for any depolarising strength, including $p=1$ and with no magic state present.
\end{proposition}

\begin{proof}
This specialises Spekkens' proof of preparation contextuality for the maximally mixed state~\cite{Spekkens2005GenContext}. Measurement noncontextuality and the availability of all coarse-grainings make the response function affine in the effect, and in finite dimension an affine function on the (full-dimensional) effect set extends uniquely to a linear functional, so $\xi(E|\lambda)=\tr(E D_\lambda)$ for a Hermitian $D_\lambda$; allowing every effect forces $D_\lambda\succeq0$, and $\xi(\id|\lambda)=1$ forces $\tr D_\lambda=1$: ontic states are density matrices. Reproducing the Born rule for all effects gives $\int\mu_\rho(d\lambda)D_\lambda=\rho$; for pure $\rho$, extremality forces $D_\lambda=\rho$ for $\mu_\rho$-almost every $\lambda$. Decomposing $\id/d$ over stabiliser basis $0$ and preparation noncontextuality give $\mu_{\id/d}(\{\lambda: D_\lambda\in\text{basis }0\})=1$, and likewise for basis $1$; the two sets are disjoint because mutually unbiased bases share no rank-one projector, so $1=\mu_{\id/d}(\Lambda)\ge2$, a contradiction.
\end{proof}

Numerically, adjoining even a single Haar-random basis to the $d{+}1$ stabiliser bases already makes the bare stabiliser polytope contextual at $d=3$ and $d=5$, while re-using a stabiliser basis leaves the Gross model intact. The Gross representation evades the argument only because its duals $A_u$, with eigenvalues $\pm1$, are not density matrices, admissible only because non-stabiliser measurements are excluded. This is why the measurements must be restricted to the subtheory for any of the three questions above to be non-trivial.

\section{The exact-frame threshold}\label{sec:frames}

\begin{proposition}\label{thm:pcrit}
For every odd prime $d$ and every magic state $\rho_m$,
\[ \pcrit=\pW=\frac{d^2|W_{\min}|}{1+d^2|W_{\min}|}. \]
\end{proposition}

\begin{proof}
If $p\ge\pW$, then Gross's representation itself is a non-negative diagram-preserving representation of $\Sd\cup\{\rho_M(p)\}$: it represents the subtheory non-negatively~\cite{gross2006}, stabiliser states and effects being (normalised) indicator functions of phase-space lines and each Clifford acting as the permutation $u\mapsto\pi_Uu$ of phase space, since $UA_uU^\dagger=A_{\pi_Uu}$; and $W_{\rho_M(p)}\ge0$ by \cref{eq:pW}. Conversely, suppose some exact frame gives a non-negative representation at noise $p$. Its restriction to $\Sd$ is a non-negative diagram-preserving representation of the stabiliser subtheory, which by \cref{thm:grossuniq}'s uniqueness result \cite{Schmid2022StabSub} must be Gross's; clause 1 there covers the single-system case with $|\Lambda|=d^2$, which we use here. Therefore, the frame is Gross's up to a relabelling of $\Lambda$, and non-negativity of the represented magic state is the condition $W_{\rho_M(p)}\ge0$, i.e., $p\ge\pW$.
\end{proof}

The same argument applies verbatim to $n$ qudits, with $d^2$ replaced throughout by $D^2=d^{\,2n}$ where $D=d^{\,n}$, since clause 1 of \cref{thm:grossuniq} covers the multi-particle case: a state adjoined to the $n$-qudit stabiliser subtheory admits an exact non-negative frame if and only if its $n$-qudit Gross Wigner function is non-negative. \Cref{sec:twoqutrit} uses it in this form, for the locally depolarised pair $\rho_M(p)^{\otimes2}$; along the global trajectory $(1-p)\rho_m+p\,\id/D$ it gives $\pcrit=\pW=D^2|W_{\min}|/(1+D^2|W_{\min}|)$.

\Cref{thm:pcrit} closes the optimisation of \cref{eq:p_crit_def}: no other representation witnesses the transition at a different point. It is a corollary of clause 1 of \cref{thm:grossuniq} rather than an independent argument, and we state it because the reduction makes the exact-frame rung computable, and because the same reasoning covers mixed adjoined states, for which clause 4 of \cref{thm:grossuniq} is silent. This bears on the question Ref.~\cite{Schmid2022StabSub} finishes on:

\begin{quote}
 \textit{``Perhaps the most important open question that remains is whether an analogous sufficiency result} [clause 4 of our \cref{thm:grossuniq}, their Theorem 3] \textit{holds for mixed states and generic quantum channels.''}
\end{quote}

 \Cref{thm:pcrit} does not settle that, and the states of Ref.~\cite{prakashgupta2020}, which for every code length $N$ include Wigner-negative ones that are bound for all distillation routines based on stabiliser codes shorter than $N$, are evidence against a mixed-state version of clause 4. \Cref{thm:pcrit} instead settles the exact-frame analogue, which holds for any mixed state $\rho$ adjoined to $\Sd$: the scenario admits no exact frame iff $W_\rho$ is negative. Generic channels are open in both readings. For class-representative models of the Clifford scenario even the contextuality analogue fails, and \cref{sec:disagree} exhibits a state and an interval of noise on which the scenario has such a model while the state stays Wigner-negative.

\emph{No Kirkwood-Dirac frame witnesses the transition.} A Kirkwood-Dirac frame is built from rank-one elements $F^{\mathrm{KD}}_{i,i'}=\proj{a'_{i'}}\proj{a_i}$, whereas clause 1 of \cref{thm:grossuniq} forces any non-negative diagram-preserving representation of $\Sd$ to be Gross's, whose frame elements are the rank-$d$ phase-point operators $A_u/d$. No relabelling reconciles the two, so no Kirkwood-Dirac frame represents $\Sd$ non-negatively at any noise level, and there is no subtheory-wide Kirkwood-Dirac threshold to set against $\pW$. Ref.~\cite{Schmid2024KDSpekkensC} establishes this by the same uniqueness argument, separating Gross's frame from a Kirkwood-Dirac frame by Hermiticity rather than by rank; \cref{thm:pcrit} adds that it holds for $\Sd\cup\{\rho_M(p)\}$ at every $p$. Restricting attention to the state does not rescue the comparison: for any $\rho=\sum_i r_i\proj{a_i}$, a Kirkwood-Dirac frame whose first basis is $\{\ket{a_i}\}$ gives $\varrho_{i,j}(\rho)=r_i|\braket{a_i}{b_j}|^2\ge0$ for any second basis with no vanishing overlaps, so a state-level Kirkwood-Dirac threshold is $0$ for every state, the pure magic state included, and certifies nothing. Kirkwood-Dirac objects return in \cref{sec:fragments}, where the frame is not required to be exact for the whole subtheory and the two-basis case has a forced answer.

\section{Ontological models of fragments}\label{sec:fragments}

\subsection{A noncontextuality criterion for fragments}

Fix a finite set of projective measurements $M_1,\dots,M_k$ on $\mathcal{H}_d$, each a basis, and pairwise sharing no rank-one projector, and a convex preparation set $\mathcal{P}$ of states containing a ball about $\id/d$. Let $V=\mathrm{span}\{\Pi^m_i\}$ be the real span of the allowed effects, and $\mathcal{P}_V$ the orthogonal projection onto $V$ (Hilbert-Schmidt inner product).

\begin{theorem}[Fragment criterion]\label{thm:criterion}
The scenario $(\mathcal{P};M_1,\dots,M_k)$ admits a generalised-noncontextual ontological model iff there exist Hermitian operators $G_\lambda$, indexed by outcome tuples $\lambda\in[d]^k$, such that
\begin{align}
 &\text{(0)}\quad G_\lambda\in V, \nonumber\\
 &\text{(a)}\quad \textstyle\sum_{\lambda:\lambda_m=i}G_\lambda=\Pi^m_i \text{ for all } m,i, \nonumber\\
 &\text{(b)}\quad \tr(G_\lambda\rho)\ge0 \text{ for every } \rho\in\mathcal{P}. \nonumber
\end{align}
The model is then $\mu_\rho(\lambda)=\tr(G_\lambda\rho)$ with deterministic response functions $\xi_m(i|\lambda)=\delta_{i,\lambda_m}$.
\end{theorem}

We prove this in \cref{app:criterion}; condition (0) encodes preparation noncontextuality (the epistemic state may depend on a preparation only through its accessible statistics, i.e., through $\mathcal{P}_V\rho$), and it is necessary because $\mathcal{P}$ contains a ball about $\id/d$.

The criterion is a linear program, closely related to the simplex-embedding programme for accessible GPT fragments~\cite{Schmid2021,selby2024}, to Ref.~\cite{gittonwoods2022}'s unit-separability criterion (whose reduced space, for a measurement set that is not tomographically complete, plays the role the $V$-projection plays here), to Ref.~\cite{schmid2024ptm}'s quantifier-elimination and linear-programming methods, and to the rank criteria of Ref.~\cite{shahandeh2025unified} for prepare-and-measure and Ref.~\cite{yianni2026} for prepare-transform-measure scenarios, which test the noncontextuality of an abstract table of statistics. Those methods apply to any single fragment mechanically, Ref.~\cite{yianni2026}'s deciding it and Ref.~\cite{schmid2024ptm}'s certifying contextuality whenever its linear program is infeasible; \cref{thm:criterion} adds that it is stated inside quantum theory, with the ontic space fixed in advance to outcome tuples, which lets the next subsection solve the whole $d$-parameterised family in closed form rather than one instance at a time. Fixing the ontic space to outcome tuples loses no generality here, by the argument in \cref{app:criterion}, and the resulting $d^{\,k}$ ontic states may exceed the $d^2$ that diagram preservation would force (and do for $k>2$). That is allowed: for prepare-and-measure scenarios the ontic space may have to be larger than the dimension of the state space~\cite{Schmid2021}, with Ref.~\cite{gittonwoods2022}'s Theorem~3 showing that any Riemann-integrable noncontextual model can be replaced by one whose ontic space has cardinality between $\dim\mathcal{R}$ and $(\dim\mathcal{R})^2$, $\mathcal{R}$ being their reduced space, and only composability pins it to $d^2$ (\cref{sec:reconcile}).

Neither this criterion nor \cref{lem:ray} below uses the odd-prime convention of \cref{sec:setting}, or any stabiliser structure. The two conditions on $\mathcal{P}$ are independent, and both are needed: \cref{thm:criterion} uses only the ball about $\id/d$, while \cref{lem:ray} applies condition (b) at the basis eigenstates, and for $k<d+1$ those eigenstates have affine hull of dimension $k(d-1)<d^2-1$ while a ball about $\id/d$ contains no pure state at all. Both hold whenever $\mathcal{P}$ contains the stabiliser polytope and the bases are stabiliser bases, as they do throughout. They hold well beyond it: in any dimension, take any $k$ pairwise-unbiased bases, which need not be stabiliser bases and need not extend to a complete set, and let $\mathcal{P}$ be the convex hull of their eigenstates together with a ball about $\id/d$; \cref{thm:criterion,lem:ray} then apply verbatim, so the fragment threshold of any state adjoined to $\mathcal{P}$ is decided by a finite linear program over $d^{\,k}$ outcome tuples. The orbit-clause half of \cref{thm:trichotomy} needs in addition only that the relevant group act transitively on (basis, outcome) pairs and on ordered pairs of non-parallel such pairs, so that $c_2=c_1/d$, which holds at $d=2$ for $\{Z,X,Y\}$ ($c_1=4$, $c_2=2$) and at $d=9$ for the ten $\mathrm{GF}(9)$ MUBs. All three results therefore hold in any dimension admitting the bases in question, which is what licenses their use at $d=2$ in \cref{sec:qubit} and at $d=9$ in \cref{sec:twoqutrit}.

\subsection{Closed form of the forced rays}

Each eigenstate of an allowed basis is deterministic for that measurement, so conditions (a)+(b), applied at those eigenstates, force $\tr(G_\lambda s_{m,j})=0$ whenever $j\ne\lambda_m$; this is the one place the preparation set is required to contain the eigenstates themselves. For $k$ pairwise-unbiased bases these are $k(d-1)$ independent conditions in the $(k(d-1)+1)$-dimensional space $V$, pinning each $G_\lambda$ to a ray in closed form.

\begin{lemma}[Closed form of the rays]\label{lem:ray}
For $k$ pairwise-unbiased bases whose eigenstates lie in $\mathcal{P}$, every solution of (0), (a) and (b) of \cref{thm:criterion} has $G_\lambda=x_\lambda N_\lambda$ with $x_\lambda\ge0$ and
\begin{equation}\label{eq:ray}
 N_\lambda=\sum_{m=1}^{k}\Pi^m_{\lambda_m}-\frac{k-1}{d}\,\id,
\end{equation}
so that $\tr(N_\lambda\rho)=s_\lambda(\rho)-\frac{k-1}{d}$, where $s_\lambda$ is the score from \cref{eq:score}. The weights are constrained only by the marginal conditions $\sum_{\lambda:\lambda_m=i}x_\lambda=1$ and $\sum_{\lambda:\lambda_m=i,\lambda_{m'}=j}x_\lambda=\tr(\Pi^m_i\Pi^{m'}_j)=\frac1d$.
\end{lemma}

We prove this in \cref{app:ray}. The criterion therefore becomes combinatorial: a noncontextual model exists iff there is a weighting of outcome tuples with uniform pairwise marginals, supported on tuples whose score is at least $(k-1)/d$ on every allowed preparation. The operators $N_\lambda$ are not new. At $k=d+1$ they are the operators $A_q$ whose non-negativity defines the stabiliser polytope in Ref.~\cite{Howard2014}, and, up to normalisation, the phase point operators $A(u)=N_\lambda/d$ of van Dam and Howard~\cite{vandam2011}; there are $d^{\,d+1}$ of them. At every $k$ they are of the closed-under-inference and noncontextual (CNC) type introduced for qubits by Ref.~\cite{raussendorf2020cnc} and carried to odd prime dimensions by Ref.~\cite{zurelheimendahl2024}: writing $\Omega\subset\mathbb{Z}_d^2$ for the union of the $k$ one-dimensional subgroups labelling the chosen bases, so that $|\Omega|=k(d-1)+1$, and expanding in Weyl operators,
\begin{equation}\label{eq:cnc}
 N_\lambda=\frac1d\sum_{b\in\Omega}\omega^{-\gamma(b)}\mathcal{D}_b,
\end{equation}
with $\gamma:\Omega\to\mathbb{Z}_d$ additive on each subgroup; there are $d$ such restrictions per subgroup, and the choice on the $m$th is the outcome $\lambda_m$. That construction is a simulability criterion, whose non-negative sector strictly contains the Wigner-positive states, and it is not phrased in terms of generalised contextuality. \Cref{lem:ray} adds that these operators are forced by noncontextuality of a fragment rather than by phase space.

\emph{Gross:} For $k=d+1$, the classical identity $\sum_{\ell\ni u}p_\ell=1+dW_u$ (summing the probabilities of the $d{+}1$ lines through a phase-space point $u$) says that on the $d^2$ collinear tuples (those choosing, for each basis, the line through a common point $u$) the ray functional is the Wigner function: $N_{\lambda(u)}=A_u$ and $s_{\lambda(u)}(\rho)-1=dW_\rho(u)$. Uniform weights $x=1/d$ on the collinear tuples reproduce Gross's model exactly.

\emph{Margenau-Hill:} For $k=2$ the weights are completely forced ($x_{ij}=1/d$), and $G_{ij}=\mathcal{P}_V[\tfrac12(F^{\mathrm{KD}}_{ij}+F^{\mathrm{KD}\dagger}_{ij})]$, the $V$-projected Hermitian (Margenau-Hill) part of the Kirkwood-Dirac frame element $F^{\mathrm{KD}}_{ij}=\ket{b_j}\braket{b_j}{a_i}\hskip-1mm\bra{a_i}$. Two-basis fragments therefore have a unique candidate noncontextual model, the Margenau-Hill distribution of the $V$-projected state, whose threshold is its positivity point. That negativity of such an object signals generalised contextuality is known: Ref.~\cite{pusey2014} shows that a projector whose weak value~\cite{aharonov1988} has negative real part, which is a negative Margenau-Hill element, yields a proof of contextuality, and Ref.~\cite{kunjwal2019} makes that result noise-robust, shows its hypotheses are tight, and extends it to the imaginary part. Real-part-of-Kirkwood-Dirac objects therefore witness the dissipation of contextuality for restricted fragments, which is what remains of the comparison abandoned after \cref{thm:pcrit}; what \cref{lem:ray} adds is that for a two-basis fragment this object is the only candidate model, so its positivity point is the threshold and not merely a sufficient condition for contextuality. The $V$-projection is required, as the positivity point of the raw Margenau-Hill distribution is not the threshold, and can err in either direction.

\emph{The eight-state model:} For $d=2$ with bases $\{Z,X,Y\}$, \cref{eq:ray} gives $N_\lambda=\tfrac12(\id+\lambda\cdot\vec\sigma)$, $\lambda\in\{\pm1\}^3$: the rays are the ontic states of the Wallman-Bartlett eight-state model~\cite{wallman2012}, equivalently the eight vertices of the single-qubit $\Lambda$ polytope~\cite{zurel2020,okay2021}.

The two-basis case also bounds the fragment threshold from above by a purely geometric quantity.

\begin{proposition}\label{prop:twobasis}
For any two stabiliser bases and any magic state, the two-basis prepare-and-measure fragment whose preparations are the stabiliser polytope together with $\rho_M(p)$ is noncontextual for every $p\ge\pstab$, so its threshold is at most the stabiliser-polytope entry point.
\end{proposition}

\begin{proof}
For $p\ge\pstab$ the state is a convex mixture of stabiliser states, so every preparation of the scenario lies in the stabiliser polytope. By \cref{lem:ray} with $k=2$, the fragment has a unique candidate noncontextual model, whose weights are forced to be $x_{ij}=1/d$, meaning its epistemic states are
\begin{equation}
 \mu_\rho(i,j)=\tfrac1d\big(s_{ij}(\rho)-\tfrac1d\big),
\end{equation}
with $s_{ij}(\rho)=p(i|1,\rho)+p(j|2,\rho)$ the two-basis score of \cref{eq:score}. Every pure stabiliser state has $s_{ij}\ge1/d$: its outcome probability in each of the two bases is either deterministic or $1/d$ by unbiasedness. Therefore $\mu$ is non-negative on every polytope state, and in particular on $\rho_M(p)$ for all $p\ge\pstab$.
\end{proof}

Since $\pW\le\pstab$, with equality when the collinear clause is the maximal one (\cref{thm:trichotomy}), the bound is $\pW$ for the states with no bound-magic window (\cref{sec:boundmagic}), which is four of the six in \cref{tab:six}.
\section{Value assignments and the trichotomy}\label{sec:assignments}

From here take $k=d+1$ (all stabiliser bases) and write $n_\lambda(\rho)=s_\lambda(\rho)-1$. The Clifford group acts on outcome tuples: a Clifford $U$ permutes the $d{+}1$ bases and their outcomes, hence permutes assignments by some $\pi_U$ with $U N_\lambda U^\dagger=N_{\pi_U\lambda}$, and the induced group is the affine symplectic group $\mathrm{ASL}(2,d)$, of order $d^3(d^2-1)$. Its orbits partition the $d^{\,d+1}$ assignments; there are $2$ orbits for $d=3$ (the $9$ collinear and $72$ non-collinear tuples), $11$ for $d=5$, and $360$ for $d=7$. To each orbit $o$ and state $\rho$ associate the clause value
\begin{equation}
\begin{split}
 \tau_o(\rho)&=\max\Big(0,\,-\min_{\lambda\in o} n_\lambda(\rho)\Big),\\
 p_o(\rho)&=\frac{d\,\tau_o}{1+d\,\tau_o},
 \end{split}
\end{equation}
the noise at which every assignment in the orbit satisfies $s_\lambda\ge1$ along the depolarising trajectory. The collinear orbit's clause is exactly $\pW$.

The trichotomy below is not the three readings of \cref{sec:readings} restated. Its triple is $(\pC,\pW,\pstab)$, and it sorts the orbit clauses of one scenario into a minimum, a designated clause and a maximum; the triple of \cref{sec:readings} is $(\pPM,\pC,\pcrit)$, and it compares three different scenarios. The two share $\pC$, and by \cref{thm:pcrit} the $\pW$ of the one and the $\pcrit$ of the other are the same number, but not $\pstab$ against $\pPM$.

\begin{theorem}[Trichotomy]\label{thm:trichotomy}
For every odd prime $d$ and every state $\rho_m$, with $\pC$ the class-representative threshold of \cref{def:classrep}:
\begin{align}
 \pC(\rho_m)&=\min_{o}\,p_o(\rho_m), \label{eq:pc}\\
 \pW(\rho_m)&=p_{o_{\mathrm{coll}}}(\rho_m), \\
 \pstab(\rho_m)&=\max_{o}\,p_o(\rho_m), \label{eq:pstab}
\end{align}
where the min and max run over the $\mathrm{ASL}(2,d)$ orbits of value assignments and $o_{\mathrm{coll}}$ is the collinear orbit.
\end{theorem}

\Cref{eq:pc} is proved in \cref{app:ray}; the same argument identifies $\min_o p_o$ with the noncontextuality threshold of the prepare-and-measure fragment whose preparations are the stabiliser polytope together with the Clifford orbit of $\rho_M(p)$, a scenario containing no transformation procedures. Only the $(\le)$ direction of \cref{eq:pc} uses \cref{def:classrep}; the $(\ge)$ direction holds for any generalised-noncontextual model, so $\min_o p_o$ is a lower bound on the threshold under either reading of transformation procedures. \Cref{eq:pstab} holds in every odd prime dimension, and is the depolarising specialisation of Cormick \textit{et al.}'s characterisation of the stabiliser polytope by the assignment inequalities~\cite{cormick2006}. We confirm this independently by exhaustive enumeration in exact arithmetic at $d=3$ (the qutrit stabiliser polytope has exactly $81$ facets, the $81$ hyperplanes $\{s_\lambda=1\}$) and numerically at $d=5,7$ along Haar-random central rays, where the two agree to within $2\times10^{-8}$. Every non-trivial polytope-entry point quoted in this paper is also certified exactly: at each claimed crossing we exhibit a stabiliser-mixture decomposition of the state with algebraic weights, verified in exact arithmetic (see \cref{app:numerics}), while just below the crossing the binding assignment facet excludes the state.

The proof turns on a symmetrisation argument and a counting identity: solutions of the criterion can be averaged over the Clifford group, symmetric solutions are constant on orbits, and every orbit's two counting constants satisfy $c_2=c_1/d$, so any single fully admissible orbit already solves the marginal constraints. Class-representative noncontextuality of the Clifford scenario is therefore a disjunction of orbit clauses: some Clifford-invariant family of value assignments must be jointly probable on all preparations. For $d=3$ this is a two-clause formula,
\begin{equation}\label{eq:twoclause}
 \pC=\min(\pW,\,p_{\mathrm{nc}}),\qquad \pstab=\max(\pW,\,p_{\mathrm{nc}}),
\end{equation}
with $p_{\mathrm{nc}}$ the non-collinear clause. The second clause has a Kochen-Specker style state-dependent reading, and with the first it says that no class-representative model exists iff some Wigner point is negative \emph{and} some non-concurrent choice of one line per direction has total probability below $1$. For the strange state the worst such quadruple has probabilities $(\tfrac12,0,0,0)$, and the system stays contextual in that sense until $p=3/5$.

\section{The convex-linear collapse}\label{sec:convexlinear}

\Cref{thm:trichotomy} answers a question about class-representative models. Under Spekkens' convex-linearity assumption the answer is different, and the orbit disjunction collapses to its collinear clause. The argument turns on the sizes of the $\mathrm{ASL}(2,d)$ orbits of value assignments, which the translations alone already constrain.

\begin{lemma}[Free translations]\label{lem:freetrans}
Let $q$ be a prime power and $T\cong(\mathbb{F}_q,+)^2$ the translation subgroup of $\mathrm{ASL}(2,q)$. Then $T$ acts freely on value assignments. Hence every stabiliser meets $T$ trivially and embeds in $\mathrm{SL}(2,q)$, and every orbit has size $q^2k$ with $k$ dividing $|\mathrm{SL}(2,q)|$.
\end{lemma}

\begin{proof}
A translation has trivial linear part, so it preserves each parallel class of lines setwise; an assignment holds one line per class, so a translation fixing the assignment fixes each of its lines. Now $t_v$ fixes a line of direction $L$ only if $v\in L$, and two distinct direction subspaces already meet in $\{0\}$, so $v=0$. Therefore $\mathrm{Stab}(\lambda)\cap T=1$ and $\mathrm{Stab}(\lambda)$ injects into $\mathrm{ASL}(2,q)/T\cong\mathrm{SL}(2,q)$, giving $|o|=q^2k$ with $k=|\mathrm{SL}(2,q)|/|\mathrm{Stab}(\lambda)|$.
\end{proof}

\begin{theorem}[Convex-linear collapse]\label{thm:convexlinear}
Let $d$ be an odd prime, and consider a model of the Clifford scenario satisfying \cref{eq:spekkensNC} in which a transformation procedure implementing a channel as a probabilistic mixture is represented by the corresponding mixture of stochastic maps. Then, on any ontic space:
\begin{enumerate}
\item up to null sets $\Lambda$ is partitioned into blocks labelled by value assignments; each Clifford $U$ is represented by the block permutation $\pi_U$; and the set $B$ of labels of positive measure is a union of $\mathrm{ASL}(2,d)$ orbits;
\item every orbit $o\subseteq B$ has $|o|=d^2$;
\item the collinear orbit has $|o|=d^2$, and the model it gives is Gross's Wigner representation, which is convex-linear and non-negative exactly for $p\ge\pW$.
\end{enumerate}
Hence, by \cref{prop:orbituniq}, in every odd prime dimension,
\begin{equation}\label{eq:pCL}
 \pCL(\rho_m)=\pW(\rho_m)=\pcrit(\rho_m)
\end{equation}
for every magic state $\rho_m$.
\end{theorem}

\begin{proof}
Clauses (1) and (2) are proved in \cref{app:ray}. Clause (1) uses preparation noncontextuality and
the Clifford-orbit average alone to partition the ontic space into blocks labelled by value
assignments on which each Clifford acts by a permutation, with no finiteness assumed; clause (2)
counts the blocks, from the splitting of $\mathrm{Herm}(\mathcal{H}_d)$ into the trivial
representation plus an absolutely irreducible one, together with \cref{lem:freetrans}. We prove (3).

(3) On the collinear orbit $\pi_U$ is the matrix of $U(\cdot)U^\dagger$ in the basis $\{A_u\}$, by $UA_uU^\dagger=A_{\pi_U u}$, so $\pi_U|_{\mathrm{coll}}$ is a linear image of the channel and convex-linearity holds there. The block masses are forced: the map taking a weighting of the $d^2$ collinear assignments to its $d(d{+}1)$ line sums has trivial kernel, so the $d(d{+}1)$ Born probabilities fix them, giving $x_\lambda=1/d$ and $\mu_\rho(u)=W_\rho(u)$. Non-negativity is then Wigner positivity, i.e.\ $p\ge\pW$, and at $p\ge\pW$ Gross's model satisfies every clause of \cref{eq:spekkensNC} together with convex-linearity. \Cref{thm:pcrit} supplies the second equality of \cref{eq:pCL}.
\end{proof}

Clause (2) fixes the size of an admissible orbit but not the orbit itself. \Cref{prop:orbituniq} does.

\begin{proposition}[Uniqueness of the collinear orbit]\label{prop:orbituniq}
For every odd prime power $q$, the collinear orbit is the only $\mathrm{ASL}(2,q)$ orbit of value assignments of size $q^2$.
\end{proposition}

\begin{proof}
Let $|o|=q^2$, so $H:=\mathrm{Stab}(\lambda)$ has order $|\mathrm{ASL}(2,q)|/q^2=q(q^2-1)=|\mathrm{SL}(2,q)|$. By \cref{lem:freetrans} $H\cap T=1$, so $|HT|=|H|\,|T|=|\mathrm{ASL}(2,q)|$ and $H$ is a complement to $T$, the quotient map restricting to an isomorphism $H\to\mathrm{SL}(2,q)$.

For $q$ odd, $-\id$ is the only involution of $\mathrm{SL}(2,q)$: an involution satisfies $A^2=\id$, and $x^2-1$ is separable in odd characteristic, so $A$ is diagonalisable with eigenvalues $\pm1$, while $\det A=1$ excludes $\mathrm{diag}(1,-1)$. Let $z\in H$ be the element mapping to $-\id$. As an affine map $z:x\mapsto-x+v$, whose fixed point $x_0=v/2$ is unique, again because $q$ is odd. The isomorphism takes centre to centre, so $z$ is central in $H$, and $z(hx_0)=h\,z(x_0)=hx_0$ for every $h\in H$; uniqueness of the fixed point gives $hx_0=x_0$. So $H\le\mathrm{Stab}(x_0)$, and the two have equal order, whence $H=\mathrm{Stab}(x_0)$.

Conjugating by a translation, take $x_0=0$ and $H=\mathrm{SL}(2,q)$. This group has two orbits on the $q(q+1)$ affine lines. One is the $q+1$ lines through the origin, on which it is transitive because it is transitive on nonzero vectors. The other is the remaining $q^2-1$: each is $\{f=1\}$ for a unique nonzero functional $f$, equivariantly, and $\mathrm{SL}(2,q)=\mathrm{Sp}(2,q)$ is transitive on nonzero functionals, since $v\mapsto\omega(v,\cdot)$ is an equivariant isomorphism of the natural module with its dual. An $H$-invariant assignment is an $H$-invariant set of $q+1$ lines, hence a union of these orbits of total size $q+1$; as $q^2-1\ne q+1$ for $q\ne2$, it is the pencil through the origin, and $\lambda$ is collinear.
\end{proof}

Oddness is used twice, to make $x^2-1$ separable and to halve $v$, and it is needed: in even characteristic the statement is false, the enumeration returning $q$ orbits of size $q^2$ at $q=2$ and $q=4$ rather than one. \Cref{prop:orbituniq} discharges the hypothesis of \cref{eq:pCL} and of \cref{cor:grossunique} in every odd prime dimension.

The census is a check on \cref{prop:orbituniq}: exhaustive enumeration of the $\mathrm{ASL}(2,d)$ orbits gives sizes $9$ and $72$ at $d=3$; $25$, $600$, six of $1000$ and three of $3000$ at $d=5$; and, among the $360$ orbits at $d=7$, a single one of size $49$; so no non-collinear orbit there even satisfies $|o|\equiv1$ modulo $d^2-1$, let alone has $d^2$ members. At $d=3$ we also confirmed the mechanism directly in exact arithmetic: the span of $\{U(\cdot)U^\dagger\}$ over the $216$ Cliffords has dimension $65=1+8^2$, so the relation space has dimension $151$; the span of $\{\pi_U\}$ on the $9$ collinear assignments has dimension $65$ as well, while on the $72$ non-collinear ones it has dimension $210$, leaving a kernel of dimension $6$. A simpler witness needs no rank computation at all: the uniform mixture of the $9$ Weyl unitaries and the uniform mixture of all $216$ Cliffords are the same channel, complete depolarisation, but the Weyls act freely on the $72$ non-collinear assignments while $\mathrm{ASL}(2,3)$ is transitive on them, so the two procedures would be represented by stochastic maps whose entries differ by as much as $\tfrac19-\tfrac1{72}=\tfrac7{72}$.

\begin{corollary}\label{cor:grossunique}
Let $d$ be an odd prime. The unique generalised-noncontextual model of the single-qudit stabiliser subtheory $\Sd$ with a convex-linear representation of transformation procedures is then Gross's Wigner representation, on any ontic space and up to the block partition of \cref{thm:convexlinear}(1).
\end{corollary}

This is clause 1 of \cref{thm:grossuniq} with diagram preservation replaced by convex-linearity on the transformations, and with no bound on $|\Lambda|$ assumed in advance: the block argument produces the $d^2$ ontic states rather than starting from them, and composability is never used. The conclusion coincides with a uniqueness theorem of Zhu~\cite{zhu2016}: his Theorem~2 shows that in any odd prime power dimension a Clifford-covariant operator basis is a Wigner-Wootters basis, so Gross's is the only Clifford-covariant discrete Wigner function. The hypotheses differ. Zhu takes the $d^2$-element operator basis as given and asks only that every Clifford unitary permute it; \cref{thm:convexlinear} instead derives both the blocks and their count $d^2$, from \cref{eq:spekkensNC}, non-negativity and convex-linearity alone. That a covariant Wigner representation gives a transformation-noncontextual model is Catani's Proposition~1~\cite{catani2020}. Under Zhu's stronger hypotheses his theorem holds unconditionally in every odd prime power dimension; \cref{cor:grossunique} needed an extra orbit-size hypothesis, which \cref{prop:orbituniq} supplies; the gap that remains is that \cref{thm:convexlinear}, and so \cref{cor:grossunique}, is stated only for odd prime $d$. That restriction does not sit in the orbit counting: \cref{lem:freetrans,prop:orbituniq} are proved for every odd prime power $q$, and \cref{sec:twoqutrit} runs \cref{lem:ray,thm:trichotomy} at $q=9$. It sits in what $d=p^n$ with $n>1$ names, since the subtheory of \cref{sec:stabsub} is then two different theories. Read as a single Galois qudit, the Weyl operators are indexed by $\mathbb{F}_q^2$ and the linear part of the symmetry group is $\mathrm{SL}(2,q)$, the case \cref{prop:orbituniq} treats. Read as $n$ qupits, the phase space is $\mathbb{F}_p^{2n}$ and the group is $\mathrm{Sp}(2n,p)$, which for $n\ge2$ contains $\mathrm{SL}(2,q)$ strictly ($51\,840$ against $720$ at $q=9$), so its orbits on value assignments are coarser and the census of \cref{prop:orbituniq} would have to be redone under it. Neither case is settled here.

\begin{corollary}\label{cor:orbitcriterion}
A Clifford orbit $o$ of value assignments admits a convex-linear representation of the Clifford transformations only if $|o|=1+j\,(d^2-1)$ for an integer $j$, and the collinear orbit does.
\end{corollary}

\Cref{cor:orbitcriterion} is not restricted to odd $d$, and it settles the qubit case at once. For $d=2$ with the bases $\{Z,X,Y\}$ the relevant action is that of the Clifford group on sign assignments, not the affine one of \cref{prop:orbituniq}, and it is transitive on all eight, and $8=1+3j$ has no integer solution, so no convex-linear noncontextual model of the single-qubit stabiliser subtheory exists on any ontic space, at any noise level. At zero noise this recovers Ref.~\cite{lillystone2019}'s Theorem~1, that every ontological model of the single-qubit stabiliser subtheory is either preparation or transformation contextual, and supplies its mechanism: the span of $\{U(\cdot)U^\dagger\}$ over the $24$ Cliffords has dimension $10=1+3^2$ and the span of the induced permutations has dimension $20$, and the uniform Pauli mixture and the uniform Clifford mixture, both of which are the completely depolarising channel, are represented by block maps differing by $\tfrac18$. The Wallman-Bartlett eight-state model of \cref{sec:fragments} is class-representative noncontextual and not convex-linear.

Kocia and Love~\cite{kocialove2019} read Ref.~\cite{lillystone2019}'s theorem as limited to models whose states form a single convex set when expressed over disjoint ontic states, and give a Grassmann model over non-disjoint ontic states, preparation and measurement noncontextual and free of transformation contextuality, in which the re-expression over disjoint states does not preserve convex combination. They grant the partition into disjoint blocks, so their objection does not bear on \cref{thm:convexlinear}(1), and shows that the convex-linearity hypothesis of (2) cannot be dropped. In even dimension a different obstruction to Clifford covariance is available: Raussendorf \textit{et al.}~\cite{raussendorf2023cohomology} show it is a non-vanishing class in $H^1$, for every even local dimension and every number of qudits, where \cref{cor:orbitcriterion} uses a congruence on orbit sizes and presumes no operator basis. Their class vanishes for odd $d$, so it is silent on the dimensions \cref{thm:convexlinear} concerns.

\subsection{Relation to the uniqueness theorem}\label{sec:reconcile}

Reading \ref{ScenarioReadingiiiPrepMeas}, and reading \ref{ScenarioReadingiiClifford} taken in the class-representative sense of \cref{def:classrep}, admit noncontextual models at noise levels where the state is still Wigner-negative. This does not conflict with \cref{thm:grossuniq}. Clause 1 there concerns non-negative \emph{diagram-preserving} representations, which are composable objects. That is reading \ref{ScenarioReadingiExactFrame}, and \cref{thm:pcrit} confirms the clause fixes it at $\pW$. Clause 3, that no generalised-noncontextual model exists for the subtheory together with a universality-promoting state, inherits the same composable reading; so does the analogous Kochen-Specker result of Ref.~\cite{Howard2014}, whose argument brings in auxiliary processes and, for a single qudit, two-qudit projectors.

Readings \ref{ScenarioReadingiiClifford} and \ref{ScenarioReadingiiiPrepMeas} ask for less. They admit no auxiliary system, and in the class-representative sense of \cref{def:classrep} they represent each transformation equivalence class once rather than every probabilistic implementation of it. In return they are the questions an experimenter with those settings can answer: $\pC$ is the threshold of a fragment holding no transformation procedures at all, and \cref{app:experiment} measures it with the $d{+}1$ stabiliser settings alone. This is the same sensitivity to the available dual processes that Ref.~\cite{zhang2026reassessingboundaryclassicalnonclassical} identifies as the reason subtheories such as the qutrit stabiliser subtheory and Gaussian quantum mechanics can be classically explainable while exhibiting entanglement, incompatibility and steering; here we compute the dependence along a noise trajectory.

The sample-space size in reading \ref{ScenarioReadingiExactFrame} is not a restriction we impose. Diagram preservation forces it: the structure theorem makes the analysing map invertible, so $|\Lambda|$ equals the dimension $d^2$ of the GPT state space, and a representation built on an over-complete frame therefore fails to be diagram-preserving~\cite{Schmid2024structuretheorem}. Over-completeness is the freedom readings \ref{ScenarioReadingiiClifford} and \ref{ScenarioReadingiiiPrepMeas} gain by asking for less composability. The models of \cref{thm:trichotomy} are non-negative weightings of an over-complete family: for the qutrit strange state the $72$ operators $N_\lambda/24$ of the non-collinear orbit span $\mathcal{L}(\mathcal{H}_3)$, sum to $\id$, and represent every stabiliser state and $\rho_M(p)$ by non-negative numbers for all $p\ge3/5$, where an exact frame needs $p\ge3/4$. Admitting over-complete frames into reading \ref{ScenarioReadingiExactFrame} would therefore not weaken \cref{thm:pcrit}, and would not move $\pcrit$ either.

A frame represents a channel by $\Gamma_{\mathcal{E}}(\lambda'|\lambda)=\tr(F_{\lambda'}\mathcal{E}(D_\lambda))$, a kernel linear in $\mathcal{E}$ however redundant the frame is, and of rank at most $d^2$ for any choice of dual, so \cref{thm:convexlinear} applies unchanged: its part (1) forces each Clifford onto a permutation of the blocks, and already the identity channel caps the blocks of positive weight at $d^2$, which at $d=3$ leaves the collinear orbit and Gross's representation. For the $72$ the failure is explicit: $\tr(\Pi^m_iN_\lambda)$ is the indicator of $\lambda_m=i$, which forces $D_\lambda=N_\lambda$, and the induced kernel represents even the identity channel negatively, with entries $-1/24$, at every noise level. The models of \cref{thm:trichotomy} instead give up the frame on the transformation side, representing each Clifford by the permutation $\pi_U$ of the $72$ assignments, a kernel of rank $72$ that no dual reproduces. Over-completeness on its own therefore does not open the interval, and the transformation-side weakening of \cref{def:classrep} separates $\pC$ from $\pcrit$.

There is a second reason to expect the readings to separate for a single qudit specifically. Delfosse \textit{et al.}~\cite{delfosse2017} prove the equivalence between noncontextuality of value assignments and Wigner positivity for $n\ge2$ qudits; their argument needs a second symplectic plane and does not descend to $n=1$. At $n=1$ there are $d^{\,d+1}$ value assignments rather than the $d^2$ that come from phase-space points, a discrepancy also noted in Ref.~\cite{zurel2024}, and the extra ones supply the models below. \Cref{sec:twoqutrit} asks what a second copy changes, at a measurement set whose contexts are again the lines of a single plane, now over $\mathrm{GF}(9)$.

\section{Assignment inequalities, the stabiliser polytope, and bound magic}\label{sec:polytope}

\subsection{Assignment inequalities and the facets of the stabiliser polytope}\label{sec:facets}

\Cref{eq:pstab} identifies the facets of the stabiliser polytope with the assignment inequalities $s_\lambda(\rho)\ge1$: nine of the qutrit facets are the Wigner conditions $W_u\ge0$ and the other $72$ are non-collinear assignment conditions. Galv\~ao verified that these inequalities characterise the polytope for $2\le d\le5$ and conjectured it for every power-of-prime $d$~\cite{galvao2005}. The same polytope is studied geometrically by Bengtsson~\cite{bengtsson2005mubs} and by Bengtsson and Ericsson~\cite{bengtssonericsson2005}, who analyse its facet structure and its relation to finite affine planes. Cormick \textit{et al.}~\cite{cormick2006} proved it for every power-of-prime $d$, for mixed states as well as pure: their Theorem~1 says the polytope cut out by the assignment inequalities is the convex hull of the MUB projectors, which for odd prime $d$ are the stabiliser states. Refs.~\cite{vandam2011,Howard2014} use it in this form.

That none of these inequalities is redundant is also known, though it is stated in the dual language. Under polar duality, the assignment operator $N_\lambda$ being a vertex of the single-qudit $\Lambda$ polytope, the set of unit-trace Hermitian operators $X$ with $\tr(X\sigma)\ge0$ for every stabiliser state $\sigma$ \cite{zurel2020,okay2021,zurel2024}, is equivalent to $\{s_\lambda=1\}$ being a facet of the stabiliser polytope. Veitch \textit{et al.}~\cite{Veitch2012Resource} prove this for the $d^2$ Wigner phase-point operators and observe that Cormick \textit{et al.}\ imply it in prime dimension. Zurel \textit{et al.}~\cite{zurel2024} prove it for the whole family: their Theorem~3 shows that for any number $n$ of qudits of any odd dimension, the operators $d^{-n}\sum_{b\in\mathbb{Z}_d^{2n}}\omega^{-\gamma(b)}\mathcal{D}_b$ built from a noncontextual value assignment $\gamma$ on the full $n$-qudit phase space are vertices of $\Lambda$. For a single qudit two Weyl operators commute only when they are proportional, so $\gamma$ is constrained separately on each of the $d+1$ subgroups and the family is the full $d^{\,d+1}$; our $N_\lambda$ are these operators, \cref{eq:cnc} at $k=d+1$.

Ref.~\cite{zurelheimendahl2024} takes the converse direction in odd prime dimension, characterising the $\Lambda$ vertices whose Pauli coefficients have unit modulus as CNC-type. For the qutrit, Ref.~\cite{zureldavis2026} states the resulting counts, $81$ facets of the stabiliser polytope and $81$ vertices of $\Lambda$, split into the $9$ linear and $72$ non-linear value assignments, which our exhaustive enumeration reproduces.

We record a polytope-side proof, which uses only unbiasedness and a Gram determinant and identifies the vertices lying on each facet.

\begin{lemma}[Assignment facets]\label{lem:facet}
For every odd prime $d$ and every assignment $\lambda$, the hyperplane $\{s_\lambda=1\}$ is a facet of the stabiliser polytope, containing exactly the $(d{+}1)(d{-}1)=d^2-1$ non-selected stabiliser vertices.
\end{lemma}

\begin{proof}
Unbiasedness gives
\begin{equation}\label{eq:unbiasedscore}
 s_\lambda(s_{m,j})=\delta_{j,\lambda_m}+1,
\end{equation}
so the hyperplane contains a vertex iff it is non-selected. The Gram matrix of the $d^2-1$ non-selected vertices (Hilbert-Schmidt inner products) is $G=\id-\tfrac1d B+\tfrac1d J$, with $B$ the block-diagonal all-ones matrix over the $d+1$ basis blocks of size $d-1$ and $J$ the all-ones matrix. Its spectrum is $d$ (once, on the uniform vector), $1/d$ (on the $d$ remaining block-uniform vectors) and $1$ (on the complement), so $G\succ0$: the non-selected vertices are linearly, hence affinely, independent. A supporting hyperplane through $d^2-1$ affinely independent vertices of a $(d^2-1)$-dimensional polytope is a facet.
\end{proof}

With Cormick \textit{et al.}'s characterisation~\cite{cormick2006} this fixes the facet count of the single-qudit stabiliser polytope at $d^{\,d+1}$, and hence the vertex set of the single-qudit $\Lambda$ polytope, in every odd prime dimension. For qubits the polytope's symmetry group is known~\cite{heinrichgross2019}, and Ref.~\cite{obst2024} settles it for any number of systems at any prime local dimension, the single-qudit group being the wreath product $S_d\wr S_{d+1}$ of permutations within each basis and permutations of the $d+1$ bases.

This correspondence gives the thresholds a second reading. \Cref{thm:trichotomy} says that $\pW$, $\pC$ and $\pstab$ are the noise levels at which $\rho_M(p)$ enters the dual cone of the collinear orbit, of some orbit, and of every orbit of $\Lambda$ vertices, and class-representative noncontextuality of the Clifford scenario asks one Clifford orbit of vertices to support the state, with \cref{thm:convexlinear} restricting that orbit to one of size $d^2$. This also locates our thresholds relative to Refs.~\cite{zurel2020,zurel2024}'s hidden-variable models, which simulate universal quantum computation by decomposing every state over $\Lambda$ vertices, at any noise level. The two are consistent: those models do not require the state-to-distribution map to be convex-linear, and the decomposition over $\Lambda$ vertices is not unique, which is the ground on which Ref.~\cite{zurel2020} identifies its model as preparation contextual and measurement-noncontextual in Spekkens' sense. That is a different weakening from the transformation-side one of \cref{def:classrep}, used in \cref{sec:qubit} and \cref{app:ray}, where the Wallman-Bartlett model~\cite{wallman2012} used in Ref.~\cite{lillystone2019} is convex-linear on states. Below $\pcrit$ a pointwise non-negative decomposition over $\Lambda$ exists while no affine one does.

\Cref{thm:trichotomy} accounts for the coincidences in \cref{tab:six}: states whose collinear clause is maximal have $\pstab=\pW$; states where it is minimal have $\pC=\pW$; the qutrit state $\ket{b}$, whose collinear clause is the smaller of its two, has $\pC=\pW<\pstab$, with its non-collinear clause equal to the polytope entry point (confirmed symbolically).

\subsection{The bound-magic window}\label{sec:boundmagic}

Whenever a non-collinear clause exceeds the collinear one, $\pW<\pstab$ and the trajectory $\rho_M(p)$ spends the interval $(\pW,\pstab)$ Wigner-positive but outside the stabiliser polytope. On this window the state has mana zero, hence is undistillable by stabiliser protocols (mana is a magic monotone~\cite{veitch2014} and any pure magic target state has positive mana), is noncontextual in all three readings, and is classically simulable~\cite{Veitch2012Resource}, yet it is not a mixture of stabiliser states. States of this kind are known. Ref.~\cite{Veitch2012Resource} introduces them as bound universal states, and Ref.~\cite{veitch2014} renames the class ``bound magic states'', keeping the original name only in a footnote. Ref.~\cite{zureldavis2026} uses ``bound magic'' for a unitarily invariant version. A complementary phenomenon sits on the other side of $\pW$: Ref.~\cite{prakashgupta2020} shows that for every $N$ there are states outside the Wigner polytope that are contextual yet bound for all distillation routines based on stabiliser codes of length less than $N$. We add the window as a closed-form segment of the depolarising trajectory, with algebraic endpoints: for the qutrit state $\ket{b}$ it is $p\in(0.3481,\,0.5164)$, and for the ququint $k$-state $p\in(0.6987,\,0.7864)$. The endpoints are certified algebraic numbers of degree at most four, with minimal polynomials in \cref{tab:minpolys}; in closed form the qutrit window opens at $\pW(\ket b)=(294-45\sqrt2+51\sqrt5-87\sqrt{10})/199$. Both closings, the dephasing entry point below, and the strange state's entry at $3/4$ are certified as true polytope entries by exact stabiliser-mixture decompositions at the crossing (see \cref{app:numerics}). Four of the six benchmark states have no window at all ($\pstab=\pW$).

Although bound magic is noncontextual and undistillable, it can be separated operationally from stabiliser mixtures, because \cref{lem:facet} supplies the violated facet as a witness measurable with the $d+1$ stabiliser settings alone. On the whole window the binding assignment inequality fails: $v(p)=1-s_{\lambda^*}(\rho_M(p))>0$, with $v$ linear in $p$ and vanishing at $\pstab$. Operationally, draw one of the $d+1$ bases uniformly and count a success when the outcome $\lambda^*_m$ occurs: every state of the stabiliser polytope succeeds with probability at least $1/(d+1)$, while the window state falls short by $v(p)/(d+1)$. For the qutrit $\ket b$ the binding assignment is $\lambda^*=(1,1,0,0)$ (four assignments are tied), with $s_{\lambda^*}(\rho_b)$ the root $0.64394\ldots$ of $25s^4-150s^3+315s^2-270s+79$; the violation is largest at the opening of the window, where $v(\pW)=0.11604\ldots$, the root of $199v^4-392v^3+186v^2-8v-1$. For the ququint $k$-state the binding assignment is $\lambda^*=(1,1,2,2,2,2)$ and $v(\pW)=0.08217\ldots$, the root of $7495v^4+19220v^3-7680v^2+485v+1$.

The same $d{+}1$ settings also give a noncontextuality inequality for the contextual side of the threshold, and this one belongs to \cref{def:classrep} rather than to the polytope. By the remark after \cref{thm:trichotomy}, $\pC$ is the threshold of the prepare-and-measure fragment whose preparations are the stabiliser polytope together with the Clifford orbit of $\rho_M(p)$. The $d{+}1$ bases are tomographically complete, so the same twelve probabilities reconstruct the Wigner function through the line sums,
\begin{equation}\label{eq:linesum}
 W_\rho(u)=\frac1d\Big(\sum_{\ell\ni u}p(\ell\mid\rho)-1\Big),
\end{equation}
the sum running over the $d{+}1$ lines through $u$, with $p(\ell\mid\rho)$ the Born probability of the outcome whose line is $\ell$. At $d=3$, averaging the non-collinear clause over the eight assignments that put three of their lines through a phase point $u$ gives $\tfrac{15}{8}W_\rho(u)+\tfrac18\ge0$, while the collinear clause gives $W_\rho(u)\ge0$, so every preparation of a noncontextual model of that fragment obeys
\begin{equation}\label{eq:pcwitness}
 \sum_{\ell\ni u}p(\ell\mid\rho)\ \ge\ \tfrac45,
 \qquad\text{equivalently}\qquad W_\rho(u)\ge-\tfrac1{15},
\end{equation}
at every $u$. Draw one of the four bases uniformly and count a success when the outcome's line passes through $u$: noncontextuality then requires a success probability of at least $\tfrac15$, while every state of the stabiliser polytope reaches at least $\tfrac14$. On $\rho_M(p)$ at its negative point the success probability is $p/3$, so \cref{eq:pcwitness} fails exactly for $p<3/5$, with $\rho_M(3/5)$ attaining it, and the bound is tight, being the minimum of $3W(u)+1$ over the region of the non-collinear clause. Both witnesses need only the $d{+}1$ stabiliser settings and the same twelve probabilities; \cref{app:experiment} gives the shot budgets, the systematic requirement, which is the binding one, the device-honest analysis in the sense of Ref.~\cite{mazurek2016}, and what changes at $d=5$.

The window is widest close to, but not at, $\ket{b}$. Maximising $\pstab-\pW$ over pure qutrit states (a $2\times10^5$-state search with local refinement, whose optimum is then certified exactly on the family it lands in) gives, up to Cliffords, the family $\cos\theta\ket0+\tfrac{\sin\theta}{\sqrt2}(\ket1+\ket2)$ at $\cos2\theta=1/\sqrt3$: the maximal window is $\big((15-6\sqrt3)/13,\ (12-3\sqrt3)/13\big)$, of width $3(\sqrt3-1)/13=0.16893\ldots$. The benchmark $\ket b$, at $\cos2\theta=1/\sqrt5$, reaches $0.16833$, within $0.4\%$ of the maximum. All three thresholds coincide across this family wherever the collinear clause is the minimal one, which holds on the whole range containing both states, so the coincidence does not single out $\ket b$; we keep it as the benchmark because its thresholds and window endpoints are algebraic of low degree over $\mathbb{Q}(\sqrt2,\sqrt5)$, which makes the certificates of \cref{app:numerics} exact rather than enclosed.
\section{Qubits}\label{sec:qubit}

Qubits show the readings are independent questions rather than estimates of one quantity, because for them the top rung of the hierarchy is missing entirely, and the convex-linear reading of the Clifford scenario is empty at every noise level. For even $d$ there is no non-negative diagram-preserving representation of the stabiliser subtheory at all (\cref{thm:grossuniq}, Ref.~\cite{Schmid2022StabSub}), so $\pcrit$ does not exist: no noise level admits an exact non-negative frame, and the exact-frame reading declares even the bare qubit stabiliser subtheory nonclassical. The fragment criterion, by contrast, returns finite and meaningful answers, and its rays for the bases $\{Z,X,Y\}$ are the eight-state model \cite{wallman2012} (\cref{sec:fragments}, whose criterion and ray lemma do not assume $d$ odd). Qubit stabiliser contextuality has its own literature: with convex-linear representation of transformation procedures, the single-qubit subtheory admits no noncontextual model, every ontological model of it being either preparation or transformation contextual \cite{lillystone2019}, and in magic-state computation on two-level systems contextuality and Wigner negativity are linked by state-dependent arguments, for rebits with CSS states and CSS-preserving Cliffords in Ref.~\cite{delfosse2015} and for qubit schemes whose free measurements preserve Wigner non-negativity in Ref.~\cite{raussendorf2017}, while Ref.~\cite{bermejovega2017} establishes contextuality as the resource without invoking Wigner functions; these arguments must be reconciled with the state-independent contextuality of multiqubit stabiliser measurements. For odd prime qudits the two notions coincide in the value-assignment sense, for a single qudit by Howard \textit{et al.}~\cite{Howard2014}, and for $n\ge2$ by Delfosse \textit{et al.}~\cite{delfosse2017}, though the latter count Ref.~\cite{Howard2014}'s setting as $n=2$, its observables acting on an ancilla as well, and show that for their own notion the $n=1$ equivalence fails; this does not conflict with \cref{sec:disagree}, where the failed implication is to the existence of a class-representative model of the Clifford scenario. \Cref{cor:orbitcriterion} recovers that verdict from the orbit criterion.

For qubits the Clifford group acts transitively on the eight sign assignments $\lambda\in\{\pm1\}^3$, so there is a single orbit clause, computable in closed form: $n_\lambda(\rho)=\tfrac12(1+\vec r\cdot\lambda)$ for Bloch vector $\vec r$, whence $\pC(\rho)=\max(0,\,1-1/\|\vec r\|_1)$. This is the entry point of the stabiliser octahedron, a standard quantity: for a magic state $\|\vec r\|_1$ is the qubit robustness of magic~\cite{howardcampbell2017}, and the $T$-state value is the $d=2$ entry of Ref.~\cite{vandam2011}'s table. For the $T$ state ($\vec r\propto(1,1,1)$) this is $1-1/\sqrt3\approx0.4226$; for the $H$ state ($\vec r\propto(1,0,1)$), $1-1/\sqrt2\approx0.2929$.

The prepare-and-measure fragment has one free parameter: the marginal constraints leave the triple correlator $c=\av{\lambda_x\lambda_y\lambda_z}$ undetermined, and at $c=\pm1$ the model's support is a four-element GHZ-correlation sector. This one parameter is enough to solve the fragment completely.

\begin{theorem}[qubit criteria in closed form]\label{thm:qubitppm}
A qubit state with Bloch vector $\vec r$ admits a class-representative model (\cref{def:classrep}) of the Clifford scenario iff $\|\vec r\|_1\le1$, and a generalised-noncontextual model of the prepare-and-measure fragment iff
\[ \|\vec r\|_1-2\min_m|r_m|\;\le\;1. \]
Under the convex-linear reading of transformation procedures the Clifford scenario is contextual for every $\vec r$, and at $\vec r=0$ this is Ref.~\cite{lillystone2019}'s result for the noiseless single-qubit stabiliser subtheory, so only the prepare-and-measure clause is reading-independent.
Under depolarising noise, $\vec r\to(1-p)\vec r$, these give
\begin{align*}
 \pC&=\max\Big(0,\ 1-\frac{1}{\|\vec r\|_1}\Big),\\
 \pPM&=\max\Big(0,\ 1-\frac{1}{\|\vec r\|_1-2\min_m|r_m|}\Big),
\end{align*}
and, being per-state criteria, they give the thresholds along any other noise trajectory by evaluating $\vec r(p)$.
\end{theorem}

\begin{proof}
The Clifford criterion is the single orbit clause of the preceding paragraph; for the fragment, the marginal constraints of \cref{lem:ray} force every first and second sign-moment of the weights to vanish, so $x_\lambda=(1+c\,\lambda_x\lambda_y\lambda_z)/4$ with $c\in[-1,1]$ is the complete solution family, and the only supports available are all eight assignments ($|c|<1$) or one parity sector ($c=\pm1$). Stabiliser preparations satisfy every assignment inequality, so the only danger is the noisy state: $\lambda$ must leave the support once $1+(1-p)\,\lambda\cdot\vec r<0$. A model exists iff the dangerous set sits inside a single parity sector. The first dangerous assignment, $\lambda=-\mathrm{sign}(\vec r)$, appears when $(1-p)\|\vec r\|_1>1$; the cheapest dangerous assignment of the opposite parity flips the smallest component, appearing when $(1-p)(\|\vec r\|_1-2\min_m|r_m|)>1$. The former is harmless (drop its sector), the latter fatal.
\end{proof}

The $T$ and $H$ values are the two extremes of this formula. The $T$ direction points at a vertex of the GHZ sector, its opposite-parity gap $2\min|r_m|$ is maximal, and $\|\vec r\|_1-2/\sqrt3=1/\sqrt3<1$ gives $\pPM(T)=0$; the $H$ direction has a vanishing Bloch component, so both parities go dangerous together and $\pPM(H)=\pC(H)=1-1/\sqrt2$. For states outside the octahedron ($\|\vec r\|_1>1$), $\pPM=\pC$ iff some Bloch component vanishes; inside it both are $0$. In general $\pPM=0$ iff $\|\vec r\|_1\le1+2\min_m|r_m|$, a region that includes the whole $T$ orbit but not the $H$ one.

\section{Benchmark states, and where the readings disagree}\label{sec:sixstates}

\begin{table*}[!t]
\caption{\label{tab:six}The six benchmark magic states (as adjoined to the stabiliser subtheory of the corresponding dimension), their mana, minimal Wigner value, and, along the depolarising trajectory $\rho_M(p)$, the three thresholds and the stabiliser-polytope entry point. $\theta=\frac12\arctan 2$; $k=\frac12(\sqrt{6(5+\sqrt5)}-\sqrt5-3)=0.6765\ldots$. The $d{=}5$ $\ket a$ is Ref.~\cite{jain2020}'s maximal-mana ququint Clifford eigenstate (their $\ket{B',-e^{2\pi i/3}}$, conjectured there, on the strength of a numerical search, to maximise mana over all ququint states), with $\mathcal{M}=\ln[(3+\sqrt5+\sqrt{15+6\sqrt5})/5]$; every entry of its row is Clifford-invariant. Its components involve $k$: up to normalisation it is $k\ket0-\tfrac{k^2}{4}(\ket1{+}\ket4)+\ket2{+}\ket3$. All entries exact; decimals truncated. The $\pC$ column is the class-representative threshold of \cref{def:classrep}, equivalently that of the orbit-closed prepare-and-measure fragment; the convex-linear threshold $\pCL$ has no column of its own because by \cref{thm:convexlinear} it equals $\pcrit=\pW$ in every row.}
\begin{ruledtabular}
\begin{tabular}{lllllll}
state & $\mathcal{M}$ & $W_{\min}$ & $\pPM$ & $\pC$ & $\pcrit=\pW$ & $\pstab$ \\
\hline
$d{=}3$: & & & & & & \\
$\ket{a}\propto\ket{1}{-}\ket{2}$ & $0.5108$ & $-\frac13$ & $0$ & $\frac35$ & $\frac34$ & $\frac34$ \\
$\ket{b}=\cos\theta\ket{0}{+}\frac{\sin\theta}{\sqrt2}(\ket{1}{+}\ket{2})$ & $0.3884$ & $\frac{5-\sqrt5-2\sqrt{10}}{60}$ & $0.3481$ & $0.3481$ & $0.3481$ & $0.5164$ \\
$\ket{c}\propto\ket{1}{+}\ket{2}$ & $0.5108$ & $-\frac16$ & $0$ & $\frac35$ & $\frac35$ & $\frac35$ \\
$d{=}5$: & & & & & & \\
$\ket{a}$ ($k$-state) & $0.7482$ & $-\frac{4\sqrt5+k(1+\sqrt5)}{120}$ & $0.2424$ & $0.6987$ & $0.6987$ & $0.7864$ \\
$\ket{b}\propto\ket{2}{-}\ket{3}$ & $0.4990$ & $-\frac15$ & $0$ & $\frac{\sqrt5}{3}$ & $\frac56$ & $\frac56$ \\
$\ket{c}\propto\ket{2}{+}\ket{3}$ & $0.4990$ & $-\frac{1+\sqrt5}{20}$ & $0$ & $\frac{\sqrt5}{3}$ & $\frac{20-5\sqrt5}{11}$ & $\frac{20-5\sqrt5}{11}$ \\
\end{tabular}
\end{ruledtabular}
\end{table*}

\Cref{tab:six} evaluates the thresholds for six benchmark states. Five of the states are Clifford eigenstates from Jain and Prakash \cite{jain2020}'s classification: the qutrit $\ket a,\ket c$ are the ``strange'' and Norrell states, which share the maximal qutrit mana $\ln(5/3)$~\cite{veitch2014} ($\ket c$ is Clifford-equivalent to the Norrell state $\propto-\ket0+2\ket1-\ket2$), and the ququint $\ket a,\ket b,\ket c$ are the maximal-mana state and an equal-mana pair. The remaining state, the qutrit $\ket b$, is not one of their non-degenerate eigenstates, but instead lies in the two-dimensional $+1$ eigenspace of the parity Clifford $A_0$, spanned by $\ket0$ and $(\ket1+\ket2)/\sqrt2$, and is chosen within that family in \cref{sec:boundmagic} for the low algebraic degree of its thresholds, not by an eigenvalue condition and not as an optimum. Between them the six cover maximal mana, equal mana at unequal thresholds, and coincident thresholds.

\subsection{Mana does not order the thresholds}\label{sec:mana}

At $d=3$ the equal-mana pair $\ket a,\ket c$ have $\pW=3/4$ versus $3/5$; at $d=5$ the state with the largest mana has the smallest $\pW$. This is because mana is an $\ell^1$ functional of the Wigner function, whereas \cref{eq:pW} depends only on the single deepest negative value. White-noise robustness tracks how concentrated the negativity is: the peak negativity $\|{\min(W,0)}\|_\infty$, which \cref{eq:pW} depends on but mana does not. That is for the exact-frame threshold; \cref{sec:robust} shows the class-representative threshold is not a function of $|W_{\min}|$ at all. The strange state puts all its negativity on one phase point ($W_{\min}=-1/3$), while the $d{=}5$ Clifford eigenstate of largest mana (mana exactly $\ln[(3+\sqrt5+\sqrt{15+6\sqrt5})/5]$) spreads a larger total negativity over six shallow points.

\subsection{Peak negativity as a resource quantity}\label{sec:peaknu}

The quantity is not new: Veitch \textit{et al.}~\cite{veitch2014} introduce $\mathrm{maxneg}(\rho)=-\min_u W_\rho(u)$ and reject it as a magic monotone, because tensoring with the maximally mixed state suppresses it by $d^2$ despite not changing the marginal. Their objection fixes the quantity's scope. Write $t_u(\rho)=\tr(A_u\rho)=d\,W_\rho(u)$ and
\begin{equation}
\begin{split}
 &\nu(\rho)=\max\big(0,\,-{\min}_u\,t_u(\rho)\big),\\
 &t_{\max}(\rho)={\max}_u\,t_u(\rho),
\end{split}
\end{equation}
with $t_{\max}\in(0,1]$ always: $\sum_u t_u=d$ forces $t_{\max}\ge1/d$, and $A_u^2=\id$ bounds every $t_u$ above by $1$. Then $\nu(\rho)=d\,\|{\min(W_\rho,0)}\|_\infty$ and \cref{eq:pW} becomes $\pW=d\nu/(1+d\nu)$, so $\nu$ orders exact-frame depolarising robustness at fixed dimension. The $d^2$ suppression of $\mathrm{maxneg}$ noted above is a suppression of $\nu$ by $d$.

\begin{proposition}\label{prop:nu}
(i) For fixed dimension, $\nu$ is convex, invariant under Clifford unitaries, non-increasing under mixing with any Wigner-non-negative state, and unchanged by adjoining a \emph{pure} stabiliser ancilla. It is \emph{not} non-increasing under discarding a subsystem: by (ii), $\nu(\rho\otimes\id/d)=\nu(\rho)/d$, while the discard returns $\nu(\rho)$. It is therefore non-increasing under every non-selective single-system channel built from Clifford unitaries, their mixtures and forgotten stabiliser measurements, but not under dilations that discard a mixed ancilla. This is Ref.~\cite{veitch2014}'s objection, which confines $\nu$ to fixed-dimension comparisons, which is the setting in which $\pW$ is defined.

(ii) Writing $\tilde\nu=-\min_u t_u$ (so $\nu=\max(0,\tilde\nu)$), products obey
\[ \nu(\rho\otimes\sigma)=\max\big(0,\ \tilde\nu(\rho)t_{\max}(\sigma),\ \tilde\nu(\sigma)t_{\max}(\rho)\big), \]
and likewise $t_{\max}(\rho\otimes\sigma)=\max(t_{\max}(\rho)t_{\max}(\sigma),\tilde\nu(\rho)\tilde\nu(\sigma))$; in particular a product state is Wigner-negative iff a factor is, and $\nu(\rho\otimes\rho)=\nu(\rho)t_{\max}(\rho)$ whenever $\nu(\rho)>0$.

(iii) $\nu$ is not a monotone under selective stabiliser operations: for $\rho_m\otimes\sigma$ with $\sigma=\tfrac12(\proj0+\proj1)$, part (ii) gives $\nu(\rho_m\otimes\sigma)=\nu(\rho_m)\,t_{\max}(\sigma)=\tfrac12\nu(\rho_m)$, while measuring the ancilla in the computational basis returns $\rho_m$ with certainty, doubling the average peak negativity. On a single qudit every stabiliser measurement is complete and rank-one, its post-measurement states are stabiliser states, and selective monotonicity holds trivially.
\end{proposition}

\Cref{app:ray} gives the proof of (i) and (ii); part (iii) is the computation in its statement, conditioning undoing the tensor suppression $t_{\max}(\sigma)=\tfrac12$ that part (ii) imposes. Mana remains the monotone of choice for selective protocols such as distillation, even though \cref{sec:mana} shows it does not track depolarising robustness. ($\nu$ differs from the thauma family of magic measures~\cite{wang2020thauma}, which are relative-entropy distances to the set of non-positive-mana operators rather than a pointwise minimum.) Part (ii) also recovers \cref{sec:twoqutrit}'s two-copy Wigner observation: for the strange state $\nu(\rho^{\otimes2})=\tfrac12<1=\nu(\rho)$, yet three copies restore $\nu=1$ (three deepest values multiply to $-1$). The depolarising threshold of the common-noise trajectory sits at $3/4$ throughout, since it depends on the sign of the minimum, not its size.

\subsection{The class-representative functional}\label{sec:nuC}

\Cref{thm:trichotomy} makes each of the three thresholds the image of one state functional under the same increasing map. Writing $f_o(\rho)=-\min_{\lambda\in o}n_\lambda(\rho)$ for the signed orbit clause, and using $\min_o\max(0,f_o)=\max(0,\min_o f_o)$,
\begin{equation}\label{eq:nuC}
\begin{split}
 \nu_C&=\max\big(0,\,{\textstyle\min_o}f_o\big),\\
 \nu&=\max\big(0,\,f_{o_{\mathrm{coll}}}\big),\\
 \nu_{\mathrm{stab}}&=\max\big(0,\,{\textstyle\max_o}f_o\big),
\end{split}
\end{equation}
so $\nu_C\le\nu\le\nu_{\mathrm{stab}}$, and $\pC$, $\pW$ and $\pstab$ are $d\nu_X/(1+d\nu_X)$ for $X=C$, $W$ and $\mathrm{stab}$ respectively. Equivalently $d\,\nu_X(\rho)$ is the white-noise robustness of $\rho$ against the matching classical set: the least $s\ge0$ for which $(\rho+s\,\id/d)/(1+s)$ lies in the Wigner polytope, in the union of the orbit polytopes, or in the stabiliser polytope. The outer two gauge convex bodies. The middle one gauges a union, and the next proposition says what follows from that.

\begin{proposition}\label{prop:nuC}
Let $d$ be an odd prime.

(i) $\nu_C$ is invariant under Clifford unitaries and Lipschitz in trace norm with constant $\max_\lambda\|N_\lambda\|_\infty$, which is the golden ratio $\varphi=(1+\sqrt5)/2$ at $d=3$.

(ii) Every $\tau_o$, hence $\nu_C$ and hence $\pC$, is non-increasing under every channel of $\Sd$. Writing such a channel as $\mathcal{E}=\theta\Phi+(1-\theta)\Psi$ with $\Phi$ a mixture of Cliffords and $\Psi$ taking every state into the stabiliser polytope, $\tau_o(\mathcal{E}(\rho))\le\theta\,\tau_o(\rho)$.

(iii) $\nu_C$ is not convex, and neither is its zero set, the set of states whose Clifford scenario admits a class-representative model. Amplitude damping of the strange state is affine, $(1-\gamma)\proj{a}+\gamma\proj0$, with $\pC=\pW=\tfrac12$ (\cref{sec:dephasing}). Take it and depolarising at their own thresholds,
\[\sigma_1=\tfrac12\proj{a}+\tfrac12\proj0,\qquad \sigma_2=\tfrac25\proj{a}+\tfrac15\id,\]
amplitude damping at $\gamma=\tfrac12$ and depolarising at $p=\tfrac35$. Then $\tau_{\mathrm{coll}}(\sigma_1)=0$ and $\tau_{\mathrm{nc}}(\sigma_2)=0$, so $\nu_C(\sigma_1)=\nu_C(\sigma_2)=0$, while
\[\sigma=\tfrac49\sigma_1+\tfrac59\sigma_2=\tfrac13\id-\tfrac29\big(\ket1\bra2+\ket2\bra1\big)\]
has $\tau_{\mathrm{coll}}(\sigma)=\tau_{\mathrm{nc}}(\sigma)=\tfrac19$, hence $\nu_C(\sigma)=\tfrac19$ and $\pC(\sigma)=\tfrac14$.

(iv) $\nu_C$ is not faithful, vanishing on the whole union of the orbit polytopes, and it bounds the robustness of magic $R$~\cite{howardcampbell2017} from below: $R(\rho)\ge1+2\nu_{\mathrm{stab}}(\rho)\ge1+2\nu(\rho)\ge1+2\nu_C(\rho)$.
\end{proposition}

\begin{proof}
(i) $U^\dagger N_\lambda U=N_{\pi_U^{-1}\lambda}$ and $\pi_U^{-1}$ permutes each orbit, so every $\min_{\lambda\in o}n_\lambda$ is invariant. For the constant, $|f_o(\rho)-f_o(\sigma)|\le\max_\lambda\|N_\lambda\|_\infty\|\rho-\sigma\|_1$, and both $\min_o$ and $\max(0,\cdot)$ are $1$-Lipschitz; at $d=3$ the non-collinear rays have spectrum $(\varphi,0,1-\varphi)$ (\cref{sec:robust}) and the collinear $N_\lambda=A_u$ has norm $1$.

(ii) For a mixture of Cliffords, $n_\lambda(\Phi(\rho))=\sum_kq_k\,n_{\pi_{U_k}^{-1}\lambda}(\rho)\ge\min_{\lambda'\in o}n_{\lambda'}(\rho)$, an average of values taken on the same orbit, so $f_o(\Phi(\rho))\le f_o(\rho)$. Convex mixing distributes over sequential composition, so every channel of $\Sd$ is a mixture of deterministic words in the generators. A word containing a stabiliser measurement or a mid-circuit preparation has its output in the stabiliser polytope, since at prime $d$ every stabiliser measurement is complete and rank one and the remaining generators preserve the polytope; a word containing neither is a single Clifford. That gives the stated form, and it also covers adaptive composition, where $\Pi_i\rho\Pi_i=\tr(\Pi_i\rho)\Pi_i$ makes the whole channel polytope-valued. By \cref{eq:unbiasedscore}, every $n_\lambda$ is non-negative on the polytope, so $n_\lambda(\mathcal{E}(\rho))\ge\theta\,n_\lambda(\Phi(\rho))\ge-\theta f_o(\rho)$, giving $f_o(\mathcal{E}(\rho))\le\theta f_o(\rho)$ and then $\tau_o(\mathcal{E}(\rho))\le\theta\,\tau_o(\rho)$.

(iii) Each $f_o$ is a maximum of affine functionals, hence convex, and so is each $\tau_o$; $\nu_C$ is their pointwise minimum, which need not be convex. For the witness, $W_{\sigma_1}\ge0$ gives $\tau_{\mathrm{coll}}(\sigma_1)=0$, while $\sigma_2$ is the strange state at its own class-representative threshold, so its worst non-collinear score is $1$. The mixture has $W_\sigma=-\tfrac1{27}$ at one phase point, giving $\tau_{\mathrm{coll}}=\tfrac19$, while the non-collinear assignments binding at $\sigma_1$ score $1-\tfrac19$ on it, which is also the minimum over that orbit (\cref{app:numerics}). Both clauses fail, so \cref{thm:trichotomy} gives $\pC(\sigma)=3\cdot\tfrac19/(1+3\cdot\tfrac19)=\tfrac14$.

(iv) Non-faithfulness is \cref{sec:disagree}. For the bound, write $\rho=\sum_ic_i\sigma_i$ over pure stabiliser states with $\sum_i|c_i|=R(\rho)$ and $\sum_ic_i=1$. \Cref{eq:unbiasedscore} gives $n_\lambda(\sigma_i)\in\{0,1\}$, so $n_\lambda(\rho)\ge-\sum_{i:c_i<0}|c_i|=-\tfrac12(R(\rho)-1)$ for every $\lambda$; maximise over $\lambda$.
\end{proof}

\Cref{prop:nuC} is the sense in which $\pC$ differs from $\pW$. \Cref{prop:nu}(i) makes $\nu$ convex, being a single maximum of affine functionals, whereas $\nu_C$ is a minimum of such maxima, taken over the orbit disjunction of \cref{thm:trichotomy}. Under the convex-linear reading the disjunction collapses to its collinear clause (\cref{thm:convexlinear}) and the non-convexity goes with it. At $d=3$ the two pieces of the union are the two phenomena already in hand: $\sigma_1$ is in the Wigner polytope and not the non-collinear one, at the edge of the bound-magic window (\cref{sec:boundmagic}), and $\sigma_2$ is in the non-collinear polytope and not the Wigner one, at the edge of the reading gap (\cref{sec:disagree}). Neither piece contains the other. Two noise channels that each leave the strange state class-representative noncontextual therefore mix to one that does not, and adjoining $\sigma_1$ and $\sigma_2$ to $\Sd$ together already gives a contextual scenario, since one orbit has to be admissible on every preparation at once. A resource theory of class-representative robustness would have to give up convexity, just as \cref{prop:nu}(i) and (iii) make one of peak negativity give up dimension-changing operations.

\subsection{The strange state, and where the readings part company}\label{sec:disagree}

The strange state is where the readings separate furthest. Its bare fragment is noncontextual at $p=0$, with two explicit models certified by exact rational linear programs: a $17$-ontic-state one and the uniform $18$-state one of \cref{tab:ontic}, both in \cref{app:numerics}. With the Cliffords included the class-representative threshold rises to $3/5$, and the convex-linear and exact-frame thresholds to $3/4$. For all $p\in(3/5,3/4)$ the single-qutrit Clifford scenario therefore has a class-representative noncontextual model, supported on the non-collinear orbit, while the state's Wigner function is still negative.

The interval is where \cref{def:classrep} and the exact frame part company, and it is empty under the standard reading of transformation procedures. It does not contradict Refs.~\cite{Howard2014,delfosse2017}, whose noncontextual models fix the ontic space to the phase points, whereas \cref{thm:criterion} admits every value assignment and the model here is supported on the non-collinear orbit. Their noncontextuality is measurement noncontextuality with deterministic responses, where ours is Spekkens', in which preparation equivalences also constrain the model. Ref.~\cite{delfosse2017}'s equivalence is in any case proved for $n\ge2$. Nor does it contradict Ref.~\cite{Schmid2022StabSub}: its clause on universality-promoting states (clause 3 of \cref{thm:grossuniq}) concerns representations that are composable and convex-linear on transformation procedures, and the model here is neither. Two weakenings are in play and they are not interchangeable. Over-completeness of the frame is what composability forbids; \cref{def:classrep} is what convex-linearity forbids, and by \cref{thm:convexlinear} it is the second that decides the outcome. The interval $(3/5,3/4)$, three twentieths of the trajectory, therefore measures the width of the reading gap rather than a property of the state: within a single qudit and a fixed operational scenario with Clifford transformations, Wigner negativity does not imply the absence of a class-representative model, while it does imply the absence of a convex-linear one, at every noise level, in every dimension in which the collinear orbit is the only Clifford orbit of size $d^2$.

\emph{Restricted measurement subsets:} If only $k<d+1$ stabiliser bases are allowed, the thresholds are non-monotonic in the measurement set, for the reason given in the introduction: neither the loss of Born-rule constraints nor the gain of preparation equivalences dominates. For the strange state and $k=1,\dots,4$ bases the fragment threshold runs $0\to\frac12\to\frac13\to0$, the same value being attained by every subset of a given size (the $k=1$ entry is the trivial one: a lone measurement supports no operational equivalences between preparations, so every fragment with $k=1$ is noncontextual). Through \cref{lem:ray}, for the $k=2$ value, the two-basis weights are completely forced, every stabiliser basis assigns the strange state a zero-probability outcome, and the doubly-excluded tuple becomes admissible when $2p/3\ge1/3$, i.e., at $p=\tfrac12$, inside the bound of \cref{prop:twobasis}. Some subsets make particular states classical at $p=0$ (e.g., two of the four bases suffice for $\ket b$ and $\ket c$), but no proper subset with $k\ge2$ makes the strange state classical, and at $d=5$ no proper subset with $k\ge2$ makes any of the three benchmark states classical at $p=0$. \Cref{app:survey} places the benchmark states against $20\,000$ Haar-random qutrit states.

\section{Other noise models, the noise plane, and rates}\label{sec:dephasing}

Nothing in \cref{thm:criterion,lem:ray,thm:trichotomy} is special to depolarising noise: every score $s_\lambda(\rho(p))$ is linear in the state, so along any noise trajectory that is affine in $p$ and ends in a stabiliser mixture, admissibility of each assignment is monotone and every threshold is again a finite minimum or maximum of per-orbit crossing points. The crossings stop being the one-parameter family $d\tau/(1+d\tau)$.

\emph{Dephasing.} Take qutrit dephasing, $\rho(q)=(1-q)\rho+q\,\Delta(\rho)$ with $\Delta$ the completely dephasing channel (the $Z$-twirl), writing $q$ for the dephasing strength throughout this section and keeping $p$ for the depolarising one. For the strange state, dephasing never restores classicality: $\pPM$ remains $0$ (its fragment is noncontextual with no noise, and stays so), while $\pC=\pW=\pstab=1$, so the Clifford scenario is contextual, and the state Wigner-negative and outside the polytope, at every $q<1$. We see this as $A_0$ is the parity operator, so $\tr(A_0\rho)=\rho_{00}+\rho_{12}+\rho_{21}$; the strange state has $\rho_{00}=0$ and the $Z$-twirl scales its coherences by $(1-q)$, giving $\tr(A_0\rho(q))=-(1-q)<0$ for all $q<1$; the trajectory approaches the boundary point $\Delta(\rho_a)=\mathrm{diag}(0,\tfrac12,\tfrac12)$ from outside. The contrast with depolarising noise is not peculiar to this state: Ref.~\cite{rossi2023} gives a scenario related to state discrimination whose contextuality withstands arbitrary partial dephasing, though sufficiently strong depolarising destroys it.

The state $\ket b$ behaves as it does under depolarising noise: the three thresholds coincide, $\pPM=\pC=\pW=0.56298\ldots$, while the polytope entry point is $0.72039\ldots$; both are algebraic of degree four (\cref{tab:minpolys}). The interval between them is a bound-magic window of the dephasing trajectory.

\emph{Amplitude damping} is also affine on the strange state: the state's coherences live entirely in the $\{1,2\}$ block, so the damped state is the mixture $(1-\gamma)\rho_a+\gamma\,\proj0$, and the same approach applies with $\proj0$ as endpoint. The result inverts the depolarising ordering of the clauses: the collinear clause is now minimal, $\pC=\pW=\tfrac12$ exactly (the fragment stays noncontextual throughout, as under dephasing), while the trajectory stays outside the stabiliser polytope until $\gamma=1$ ($\proj0$ is again a boundary point). The whole second half of the damping trajectory, $\gamma\in(\tfrac12,1)$, is a bound-magic window. 

\begin{table}[tbp]
\caption{\label{tab:noises}Thresholds for the qutrit strange state $\ket a$ under three standard noise models. Entries are the value of the channel's strength parameter at which the corresponding scenario becomes noncontextual (or, for $\pstab$, at which the state enters the stabiliser polytope): the strength $p$ for depolarising, the strength $q$ for dephasing, and the damping parameter $\gamma$ in the amplitude-damping row. $\pC$ is the class-representative threshold of \cref{def:classrep}; by \cref{thm:convexlinear} the convex-linear one is the $\pW$ column in every row, so of the three noise models here only depolarising separates the two Clifford readings; under composed depolarising and dephasing they separate at every $q<1$, the band of \cref{fig:noiseplane}.}
\begin{ruledtabular}
\begin{tabular}{lcccc}
noise on $\ket a$ & $\pPM$ & $\pC$ & $\pW$ & $\pstab$\\
\hline
depolarising & $0$ & $3/5$ & $3/4$ & $3/4$\\
dephasing & $0$ & $1$ & $1$ & $1$\\
amplitude damping & $0$ & $1/2$ & $1/2$ & $1$\\
\end{tabular}
\end{ruledtabular}
\end{table}

\begin{figure}[tbp]
\includegraphics[width=\columnwidth]{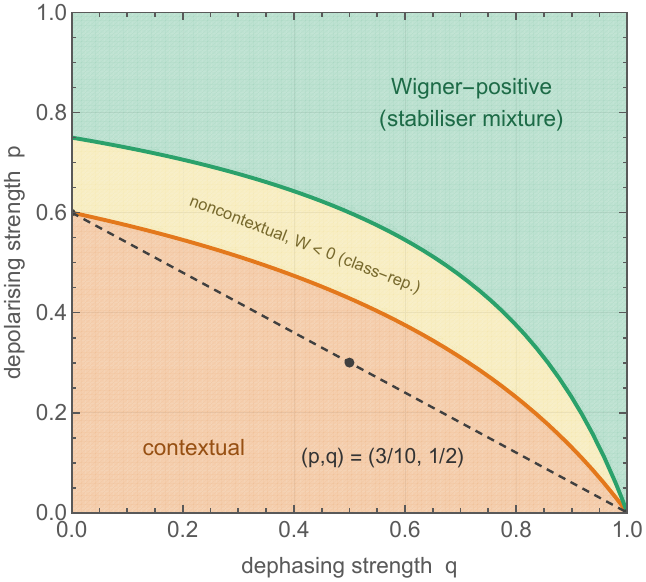}
\caption{\label{fig:noiseplane}The depolarising-dephasing plane for the strange state, with dephasing strength $q$ on the horizontal axis and depolarising strength $p$ on the vertical. Below the lower arc, \cref{eq:noiseplane}, the Clifford scenario has no class-representative model. Above the upper arc, $p=3(1-q)/(4-3q)$, the state is Wigner-positive; Wigner positivity and stabiliser-polytope membership coincide along this arc everywhere in the plane, so no bound-magic window opens. In the band between the arcs the Clifford scenario has a class-representative model while the Wigner function is still negative; by \cref{thm:convexlinear} the convex-linear reading places its own boundary on the upper arc, so for that reading the band is part of the contextual region. The dashed chord joins $(p,q)=(\tfrac35,0)$ and $(0,1)$, both noncontextual in the class-representative reading, and passes through contextual settings, including the marked midpoint $(p,q)=(\tfrac3{10},\tfrac12)$: read that way, the classical region is not convex.}
\end{figure}

\subsection{The noise plane}

Composing the two commuting channels, depolarising strength $p$ and dephasing strength $q$, gives a two-parameter phase diagram, shown in \cref{fig:noiseplane}. Every score of $\rho(p,q)$ is fixed by the pair $(a,b)=(s_\lambda(\rho_a),s_\lambda(\Delta\rho_a))$ with $a\in\{0,\tfrac12,1,\tfrac32,2\}$ and $b\in\{1,\tfrac32\}$, and the binding classes are $(0,1)$ (collinear) and $(\tfrac12,1)$ (non-collinear). The Clifford scenario has a class-representative model iff
\begin{equation}\label{eq:noiseplane}
 p\;\ge\;\frac{3(1-q)}{5-3q}.
\end{equation}
Wigner positivity and polytope entry coincide along $p=3(1-q)/(4-3q)$, so no bound-magic window opens anywhere in the plane, and the fragment is noncontextual at every $(p,q)$. The arc \eqref{eq:noiseplane} interpolates between the class-representative depolarising threshold $3/5$ and the dephasing endpoint $q=1$, and it is concave, so the set of classical noise settings is not convex: read class-representatively, $(p,q)=(\tfrac35,0)$ and $(0,1)$ are both noncontextual, yet their midpoint $(\tfrac3{10},\tfrac12)$ lies below the boundary value $\tfrac37$ and is contextual.

Separating states from noise parameters explains the non-convexity. Each orbit clause defines the region $\{\rho:\ s_\lambda(\rho)\ge1\ \text{for all}\ \lambda\in o\}$, an intersection of half-spaces and hence a convex polytope, and every such region contains the maximally mixed state, whose scores all equal $(d+1)/d$. The set of states whose Clifford scenario has a class-representative model is therefore the union of \cref{eq:nuC}'s orbit polytopes, and it is star-shaped about $\id/d$: along $(1-x)\sigma+x\,\id/d$ each score is $(1-x)s_\lambda(\sigma)+x\tfrac{d+1}d$, so an admissible orbit stays admissible. Consequently depolarising noise never restores contextuality, from any starting state, and each threshold is a single crossing. For the same reason, mixing (rather than composing) two channels that each make a state noncontextual cannot restore contextuality when both do so through one clause, though \cref{prop:nuC}(iii) shows it can when they do so through different ones: the equal mixture of depolarising at the class-representative threshold $3/5$ and complete dephasing leaves the strange state's worst non-collinear score at exactly $1$, on the boundary of its clause polytope, hence noncontextual in that reading. Composition, by contrast, is bilinear in the two strengths, and the classical region it defines in the $(p,q)$ plane is not convex, even though composing the two settings at full strength, $(p,q)=(\tfrac35,1)$, is itself classical. At $d=3$ the union has two pieces, the Wigner polytope (collinear clause) and the non-collinear polytope, intersecting in the stabiliser polytope; the bound-magic windows of \cref{sec:boundmagic} lie in the first and not the second. \Cref{app:rates} restates the plane in continuous time, as onset times for commuting depolarising and dephasing semigroups.

\section{The most noise-robust magic states}\label{sec:robust}

\begin{table}[tbp]
\caption{\label{tab:maxpC}The most $\pC$-robust magic state in each dimension under depolarising noise, with $\pC$ the class-representative threshold of \cref{def:classrep}, and the fraction of the exact-frame maximum $d/(d+1)$ it retains. By \cref{thm:convexlinear} the convex-linear reading has the same maximiser and the same maximum $d/(d+1)$ as the exact-frame one. $\varphi=(1+\sqrt5)/2$. Minimal polynomials for $\tau^*$ and $\max_\rho\pC$, and the root meant in each case, are in \cref{tab:minpolys}.}
\begin{ruledtabular}
\begin{tabular}{cccc}
$d$ & $\tau^*$ & $\max_\rho\pC$ & ratio to $\frac{d}{d+1}$\\
\hline
$3$ & $1/\varphi$ & $\frac{21-3\sqrt5}{22}=0.64963$ & $0.8662$\\
$5$ & $0.7028115$ & $0.7784698$ & $0.9342$\\
$7$ & $0.5910090$ & $0.8053363$ & $0.9204$\\
\end{tabular}
\end{ruledtabular}
\end{table}

Which magic state tolerates the most depolarising noise before its scenario becomes noncontextual? The answer depends on the reading. The Wigner-point maximisation follows from Ref.~\cite{vandam2011}'s Theorem~3, and by \cref{thm:pcrit} that settles the exact-frame case, and with \cref{thm:convexlinear,prop:orbituniq} the convex-linear one as well in every odd prime dimension; we restate it for comparison with the other two, which are new.

\begin{proposition}[after van Dam and Howard~\cite{vandam2011}]\label{thm:pwmax}
$\max_\rho \pW(\rho)=\frac{d}{d+1}$, attained exactly on states supported in the $(-1)$-eigenspace of a phase-point operator; these eigenspaces have dimension $(d-1)/2$.
\end{proposition}

\begin{proof}
$W_\rho(u)=\tr(A_u\rho)/d\ge\eigmin(A_u)/d=-1/d$, with equality iff $\rho$ is supported in the $(-1)$-eigenspace; insert into \cref{eq:pW}. Ref.~\cite{vandam2011} proves the value $d/(d+1)$ and exhibits an attaining family of $(d-1)/2$ orthogonal states per phase-point operator; the ``exactly'' clause, and the reading in terms of $\pW$ rather than $\pstab$, are the converse above.
\end{proof}

At $d=3$ the eigenspace is one-dimensional and the maximiser is unique up to Cliffords: the strange state ($A_0$ is the parity operator, with $-1$-eigenvector $\propto\ket1-\ket2$). At $d\ge5$ the pure maximisers form a $\mathbb{CP}^{(d-3)/2}$ per phase point: at $d=5$ a Bloch sphere of states, all with $\pW=5/6$, comprising the benchmark $\ket b$ but also a continuum of states of different mana ($0.4990$ to ${\approx}0.624$), a family equally robust and inequivalently magic, an extreme case of the decoupling in \cref{sec:mana}. The degeneracy does not extend to the class-representative reading: scanning the sphere with the orbit clauses, $\pC$ varies across the family, from $\sqrt5/3=0.74535\ldots$ (attained at $\ket b$) up to about $0.7539$, located numerically by a seeded search and a $60\times60$ grid; that it stays below $\pW=5/6$ across the family needs no scan, following from the proved bound $\max_\rho\pC=0.7784698\ldots$ of \cref{tab:maxpC}. Equal exact-frame robustness thus hides a spread of class-representative robustness, and no Wigner-optimal state is $\pC$-optimal.

For the Clifford scenario in the class-representative reading, \cref{thm:trichotomy} reduces the maximisation to a bottleneck problem over orbits: since $n_\lambda(\rho)\ge\eigmin(N_\lambda)$ for every state,
\begin{equation}
 \max_\rho \pC(\rho)\le\frac{d\,\tau^*}{1+d\,\tau^*},\quad
 \tau^*=\min_o\max_{\lambda\in o}\big({-}\eigmin(N_\lambda)\big),
\end{equation}
the max-min being at most the min-max. The bound is attained. For each of $d=3,5,7$ the projector onto a bottom eigenvector of the critical orbit's ray saturates it: no other assignment in that orbit scores lower on the state, and no other orbit's clause binds first, both certified in exact arithmetic over the field generated by $\omega$ and $\tau^*$ (\cref{app:numerics}). \Cref{tab:maxpC}'s entries are therefore exact. The critical orbit is the non-collinear one at $d=3$; at $d=5$ it is one of size $1000$ among the eleven, and at $d=7$ one of size $5488$ among the $360$, where the bottom eigenspace is two-dimensional and three distinct members of it give the same verdict.

At $d=3$ the critical (non-collinear) ray has spectrum $(\varphi,\,0,\,1-\varphi)$ (characteristic polynomial $\lambda(\lambda^2-\lambda-1)$); the optimal state is $(\varphi\ket1-\ket2)/\sqrt{1+\varphi^2}$ up to Cliffords. At $d=7$ the critical clause value is the middle of the three real roots of its minimal polynomial $7\tau^3-7\tau^2+1$ (the roots are $-0.3280$, $0.5910$ and $0.7370$); at $d=5$ it is a degree-eight algebraic number, the seventh in increasing order of the eight real roots of its minimal polynomial (\cref{tab:minpolys}). The maxima themselves are algebraic of the same degrees, $\max_\rho\pC=(21-3\sqrt5)/22$ at $d=3$ and the other two in \cref{tab:minpolys}. The optimal states are not \cref{thm:pwmax}'s Wigner-optimal states: the $\pC$-optimal states have smaller peak negativity but a larger critical non-collinear clause.

The fraction of exact-frame robustness retained in the class-representative reading, $\max\pC/\max\pW$, is not monotonic in $d$ ($0.866,\,0.934,\,0.920$ at $d=3,5,7$), so the exact-frame overshoot neither closes nor grows steadily with $d$ over this range. For the bare prepare-and-measure fragment the largest value we find at $d=3$ is $\frac12$, attained by $\frac13(2\omega^2\ket0-\ket1+2\ket2)$, on which all three notions coincide. At $d=5$ a seesaw search with a constrained polish gives $\pPM\ge0.5378$, attained on a state supported on four of the five levels with near-equal weights; the witness is numerical and no closed form for it was found, but the infeasibility just below the crossing has a rational Farkas dual checked in exact arithmetic over the $0/1$ constraint matrix, so the bound does not rest on solver tolerances, even though the crossing itself is located numerically (\cref{app:numerics}), and three of five independent seeds exceed $0.47$, so it rests on no single trajectory. The $d=5$ fragment maximum therefore exceeds the largest $d=3$ value we find, $\tfrac12$, and is to be compared with the proved upper bound $0.7785$. The $d{=}5$ benchmark value $\pPM(\ket a)=0.24241$ of \cref{tab:six} is itself an exact algebraic number of degree four (\cref{tab:minpolys}).

\section{Two qutrits}\label{sec:twoqutrit}

Composite systems introduce a new mechanism, and can close the gap between the readings that \cref{def:classrep} opens. We illustrate with two qutrits ($\mathcal{H}_9$), whose stabiliser subtheory has $360$ pure states and $40$ stabiliser bases; among these, the ten Galois-field MUBs of $\mathrm{GF}(9)$~\cite{gibbons2004} are pairwise unbiased and include six entangled bases, which we verify in exact arithmetic.

\begin{theorem}\label{thm:twomeas}
The scenario whose preparations are all $360$ two-qutrit stabiliser states (and mixtures) and whose measurements are just two suitably chosen stabiliser MUBs is Spekkens contextual, with no magic state and no noise.
\end{theorem}

\begin{proof}
Let $B_1$ be the computational basis and $B_2$ a suitable entangled $\mathrm{GF}(9)$ MUB, and let $s=\ket0\otimes\ket{+}$. One checks directly that $s$ assigns probability zero to some outcome $i$ of $B_1$ and simultaneously to some outcome $j$ of $B_2$, while $B_1,B_2$ are mutually unbiased. By \cref{lem:ray} (with $k=2$, $d=9$) any noncontextual model has forced weights $x_{ij}=\tr(\Pi^1_i\Pi^2_j)=\tfrac19$ and $\mu_s(i,j)=x_{ij}\,\tr(N_{ij}s)=\tfrac19\big(0+0-\tfrac19\big)<0$, a contradiction.
\end{proof}

That few measurements can suffice for generalised contextuality is known: a single measurement can be enough in a suitable prepare-and-measure scenario~\cite{selbyincompat2023}, though the single measurement there is a flag-convexified, and necessarily non-projective, POVM. Ref.~\cite{rossitypicality2026} quantifies how common contextuality is over randomly sampled scenarios, finding that no set of at most $d^2$ randomly sampled states is contextual whatever the number of measurements, while for a qubit with more than four pure states and four or more projective measurements the probability of contextuality rises quickly towards one (above $99\%$ by seven states and eight measurements).

What is at issue here is the opposite direction, a fixed structured set of preparations for which dropping measurements creates the contextuality. The negative weight is not an anomalous weak value in the sense of Ref.~\cite{pusey2014}: the doubly-excluded pair makes the raw Kirkwood-Dirac element vanish, and it is the $V$-projection of \cref{lem:ray} that turns it negative. The choice of bases is not delicate: in exact arithmetic, every one of the $45$ pairs of the ten $\mathrm{GF}(9)$ MUBs admits a stabiliser state with a zero-probability outcome in each basis, so any two of the ten measurements witness the contextuality. For $33$ pairs a product or MUB state suffices; the remaining $12$ pairs are witnessed only by entangled stabiliser states. The mechanism is impossible for a single qudit of prime dimension, where every stabiliser state has overlap exactly $1/d$ with every vector of a different stabiliser basis: it is the composite overlap spectrum $\{0,\tfrac13,\tfrac19\}$ that supplies the exclusions. Nor does it contradict the noncontextuality of the full subtheory: with all measurements present, the two-qutrit Gross model applies (we verified in exact arithmetic that all $81$ phase points give deterministic outcomes on all ten $\mathrm{GF}(9)$ MUBs, with the correct support combinatorics). The two-measurement fragment is contextual because it forgets measurements: their absence creates operational equivalences among preparations that the remaining statistics cannot respect. \Cref{thm:twomeas} is the extreme case of the non-monotonicity in \cref{sec:disagree}.

The exact frame, by contrast, is not sensitive to the choice between field and product structure. For any $\mathbb F_3$-linear identification of $\mathrm{GF}(9)$ with $\mathbb F_3^2$, the $81$ phase-point operators of the $\mathrm{GF}(9)$ Gross frame coincide, as a set, with the $81$ products $A_u\otimes A_v$ of qutrit phase-point operators: the parity operator factorises because $-(c_0,c_1)=(-c_0,-c_1)$, every field Weyl operator is a product Weyl operator (shifts act componentwise, and the character $x\mapsto\omega^{\operatorname{tr}(bx)}$ factorises through the $\mathbb F_3$-bilinear trace form, which is nondegenerate), and the induced relabelling of the $81$ phase points is a bijection, which we also confirmed numerically. The field qudit and the qutrit pair therefore share their exact-frame reading, and only its group and measurement set distinguish the $\mathrm{GF}(9)$ scenario: six of the ten MUBs are entangled, and the field Clifford group contains entangling gates. 

\emph{Copies of a magic state:} Does a second copy of a magic state strengthen contextuality? In the exact-frame reading the answer is trivially no change: the product Wigner function is negative iff a factor's is, so $\pcrit$ for common-noise copies equals the single-copy $\pW$. At the fragment level the answer is scenario-dependent. We work inside the $\mathrm{GF}(9)$-qudit subtheory: the $90$ MUB states and the ten MUBs, where \cref{lem:ray} applies verbatim and \cref{thm:trichotomy} extends because $\mathrm{ASL}(2,9)$ is still flag-transitive (transitive on the $90$ lines and on ordered pairs of non-parallel lines), which is what \cref{app:ray}'s counting constants require. Here ``the stabiliser polytope'' is the hull of the $90$ MUB states, not of all $360$ two-qutrit stabiliser states. At equal measurement sets we compare the pair $\rho_M(p)^{\otimes2}$ of noisy strange states against the control $\rho_M(p)\otimes\sigma_M(p)$, one noisy strange state and one equally depolarised stabiliser state. Below the full set the answer depends on which bases are kept and not only on how many, and for the control on which stabiliser state it uses: the pair state is not invariant under the Cliffords that permute the ten directions, and for most subsets the setwise stabiliser is not transitive on the chosen bases, so \cref{thm:trichotomy}'s orbit disjunction is unavailable and a threshold has to be quoted for a named fragment. Take the bare prepare-and-measure fragment, whose preparations are the $90$ MUB states together with the state. On every $4$-subset holding two product and two entangled MUBs the pair crosses at $3/5$, against a control that crosses at $\tfrac12$ or at $0$ according to the subset and the choice of $\sigma$, so there the second copy strengthens contextuality; on the all-product $4$-subset the pair crosses at $2\sqrt3-3$ and the control at $3-\sqrt6$, and the ordering reverses. On the $5$-subsets holding three product and two entangled MUBs the pair's classical region is not even an interval: a model exists from $\sqrt5-2=0.236$, fails again above $(4-\sqrt2)/7=0.369$, and exists for every $p\ge(1+2\sqrt2)/7=0.547$, against a control that crosses at $(3-\sqrt5)/2=0.382$ for nine of the twelve stabiliser states $\sigma$ and at $(5-\sqrt{17})/2=0.438$ for the remaining basis. At the full set of ten MUBs the question is settled by exhaustive enumeration.

\begin{theorem}\label{thm:twocopy}
Within the $\mathrm{GF}(9)$ subtheory with all ten MUBs, the pair of noisy strange states saturates the exact-frame threshold: along the common local depolarising trajectory $\rho_M(p)^{\otimes2}$, whose Wigner point is the single-copy $\pW=3/4$ by the product rule of \cref{prop:nu}(ii), the class-representative threshold is $\pC\big(\rho_m^{\otimes2}\big)=3/4$, in contrast to the single copy's $3/5$.
\end{theorem}

\begin{proof}
Write $c_\lambda(u)$ for the number of the ten chosen lines through the phase point $u$, so that $\tr(N_\lambda\rho)=\sum_u(c_\lambda(u)-1)W_\rho(u)$. An assignment is collinear exactly when $c_\lambda=(10,1,\dots,1)$, and then $N_\lambda=A_u$ and its clause is Wigner positivity. The qutrit strange state has $W=-\tfrac13$ at a single phase point $u_0$ and $\tfrac16$ at the other eight, so along the trajectory $W_{\rho_M(p)}$ equals $-w_-$ at $u_0$ and $w_+$ elsewhere, with $w_-=\frac{3-4p}9$ and $w_+=\frac{3-p}{18}$. The pair's Wigner function then takes the three values $w_-^2$, $-w_-w_+$ and $w_+^2$: on the centre $\hat u_0=(u_0,u_0)$, on the rest of the cross $C=\{(u,v):u=u_0\ \text{or}\ v=u_0\}$, and off it. Hence, writing $c_0=c_\lambda(\hat u_0)$, $\Sigma_\lambda=\sum_{u\in C}c_\lambda(u)$ and $x=w_-/w_+=\frac{2(3-4p)}{3-p}$,
\begin{equation}\label{eq:twocopyphi}
 \tr\big(N_\lambda\rho_M(p)^{\otimes2}\big)=w_+^2\,\Phi_\lambda(x),
\end{equation}
with $\Phi_\lambda(x)=(c_0-1)x^2+(c_0+16-\Sigma_\lambda)x+(26-\Sigma_\lambda)$, using $|C|=17$ and $\sum_u c_\lambda(u)=90$. The score therefore depends on $\lambda$ only through the pair $(c_0,\Sigma_\lambda)$, of which the incidence geometry allows exactly $69$ values; these are the $69$ score classes. For a collinear $\lambda$ the pair is $(10,26)$, $(1,26)$ or $(1,17)$ according as its concurrency point is the centre of the cross, elsewhere on the cross, or off it, giving $\Phi_\lambda=9x^2$, $-9x$ and $9$. The collinear region is thus exactly $x\le0$, that is $p\ge\tfrac34=\pW$, which gives $\pC\le\tfrac34$ with no appeal to \cref{prop:nu}. For the converse, $\Phi_\lambda$ depends on $\lambda$ only through its class, so the admissible set changes only where one of the $69$ quadratics vanishes: their roots give eight breakpoints in $x\in(0,2]$, the image of $p\in[0,\tfrac34)$, one of them the endpoint $x=2$, and between consecutive breakpoints the admissible set is constant. Testing one rational $x$ in each of the eight resulting intervals \emph{and} each of the eight breakpoints (at a breakpoint the admissible set can be larger than on either side, so the interval tests need not cover it; at six of the eight a single class vanishes and the set equals one neighbour's, at $x=\tfrac12$, i.e., $p=\tfrac35$, five classes vanish and the set is contained in no tested set, and at the remaining breakpoint, the endpoint $x=2$, i.e., $p=0$, seven classes vanish), a depth-first search over the ten directions against the $1215$ images of $C$ under $\mathrm{ASL}(2,9)$ finds no orbit admissible; at $x=0$ the same search returns exactly the $81$ collinear assignments. So no orbit opens below $3/4$, and only the collinear one opens there. The search is exhaustive rather than heuristic; its pruning rule and the argument that the rule loses no admissible assignment are in \cref{app:numerics}, \cref{eq:prune}.

\end{proof}

Composition therefore closes the gap that \cref{def:classrep} opens, in the regime the composable theorems \cite{Howard2014,Schmid2022StabSub} address: for the pair, at the full field measurement set, the class-representative and exact-frame readings agree, which by \cref{eq:order} pins the convex-linear reading between them, and Wigner negativity again implies contextuality of the Clifford scenario however transformation procedures are represented. Ref.~\cite{delfosse2017}'s character lemma is not what does it. That lemma makes any assignment additive on all commuting pairs a character, of which there are $81$, but it needs additivity on many more pairs than the field measurement set supplies: of the $2000$ ordered pairs of distinct nonzero commuting labels in $\mathbb{F}_3^4$, only the $560$ lying within a single $\mathrm{GF}(9)$ line are measured together here, and additivity on those alone leaves the $9^{10}$ assignments that \cref{thm:twocopy} enumerates. The closure is created by the entangled contexts alone, since the frame is the product frame; and it is clause-unique: every non-collinear orbit clears strictly later.

The same phase-space form says something about both extensions. For $k$ copies the criterion at $p=3/4$ is combinatorial: writing $\overline C_k$ for the $8^k$ phase points at which no factor sits at $u_0$, an assignment satisfies $s_\lambda\ge1$ iff its $d{+}1$ chosen lines meet $\overline C_k$ in at least $8^k$ points in total, while a collinear assignment meets it in exactly $8^k$ or $8^k+3^k$. Summing the per-direction minima gives $8$, $58$ and $460$ against thresholds $8$, $64$ and $512$ for $k=1,2,3$: at one copy the two coincide, and the bound is attained by the collinear assignment at $u_0$ alone, so every other assignment has strict slack at $\pW$ and $\pC$ falls strictly below it, whereas from two copies there is slack in the bound itself, and it widens.

At $k=3$ the search is out of reach, but the worst one-line perturbation of a collinear assignment violates the criterion by $10$ and the worst two-line perturbation by $20$, against $2$ and $4$ at $k=2$. Both are class-dependent: the $217$ tight collinear base points at $k=3$ fall into five classes, with worst one-line and two-line deficits $(2,4)$, $(7,14)$, $(10,14)$, $(10,17)$ and $(10,20)$ on $1$, $24$, $12$, $120$ and $60$ base points, so the deficit $20$ is reached on $60$ of the $217$; at $k=2$ the $17$ tight base points fall into two classes, $(1,2)$ at the centre of the cross and $(2,4)$ at the other sixteen. The other two qutrit benchmark states, and random qutrit magic states, behave as the strange state does when paired: $\pC=\pW$, with the collinear orbit alone opening there. Proofs at $k\ge3$, and in other dimensions, remain open.
\section{Conclusions}
\label{sec:conclusions}

For an odd-dimensional stabiliser subtheory with an injected magic state, where a decohering system crosses the levels of the classicality hierarchy of Ref.~\cite{bozzetto2026warringcontextualitiesprovably} depends on which operational scenario is granted, and by a wide margin: for the qutrit strange state under depolarising noise the three readings of \cref{sec:readings} give $0$, $3/5$ and $3/4$. The middle number belongs to the class-representative reading of transformation procedures (\cref{def:classrep}), equivalently to the prepare-and-measure fragment whose preparations are closed under the Clifford group. On Spekkens' reading, with mixtures of transformations represented by mixtures of stochastic maps, the Clifford scenario gives $3/4$ as well: by \cref{thm:convexlinear} a Clifford orbit admits a convex-linear representation only if it has $d^2$ members, which in every odd prime dimension leaves only Gross's representation, for every magic state. 

The top of the hierarchy is rigid. \Cref{thm:pcrit}, a corollary of Ref.~\cite{Schmid2022StabSub}'s uniqueness clause, puts the exact-frame threshold at the Wigner-positivity point, so the search over frames reduces to Gross's representation; the reduction adds scope, since it covers any mixed state adjoined to the subtheory, not only states on a depolarising trajectory. That settles the exact-frame analogue of the question Ref.~\cite{Schmid2022StabSub} closes on, for any mixed state; in the class-representative Clifford-scenario reading the analogue fails (\cref{sec:disagree}), and the universality analogue they ask about, and generic channels, stay open. \Cref{thm:convexlinear} adds a uniqueness statement on the same side of the hierarchy: in every odd prime dimension the only convex-linear generalised-noncontextual model of the single-qudit stabiliser subtheory is Gross's (\cref{cor:grossunique}), which is Ref.~\cite{Schmid2022StabSub}'s clause 1 with diagram preservation replaced by convex-linearity on transformations and with no bound on the size of the ontic space assumed in advance. The same argument shows no Kirkwood-Dirac frame represents the subtheory non-negatively at any noise level, which Ref.~\cite{Schmid2024KDSpekkensC} had already established for the noiseless subtheory, so there is no subtheory-wide Kirkwood-Dirac threshold, and the comparison between representations that motivated the earlier version of this work does not exist in that form. Kirkwood-Dirac objects do govern restricted fragments: for a two-basis prepare-and-measure fragment the unique candidate noncontextual model is the $V$-projected Margenau-Hill part of the Kirkwood-Dirac distribution (\cref{lem:ray}), and the fragment threshold is its positivity point, which can lie well below $\pW$.

Below the top rung the answers are softer, and one combinatorial object supplies all of them: the Born-rule scores of value assignments over the stabiliser bases. The same family gives the fragment criterion and its forced rays, which interpolate between Gross's model and Margenau-Hill and reproduce the eight-state model at $d=2$, and the trichotomy fixing $\pC$, $\pW$ and $\pstab$ as a minimum, a designated clause and a maximum over Clifford orbits. It also gives the facets of the stabiliser polytope, equivalently the vertices of the single-qudit $\Lambda$ polytope, for which we give a self-contained proof of a result of Refs.~\cite{Veitch2012Resource,zurel2024}, with the count fixed by Ref.~\cite{cormick2006}'s characterisation, and, in the bound-magic window, an inequality violated by the state and measurable with the $d+1$ stabiliser settings alone (\cref{lem:facet}), whose device-honest version in the sense of Ref.~\cite{mazurek2016} remains to be constructed.

The type of noise matters as much as the choice of state: under dephasing the strange state stays contextual at every dephasing strength $q<1$, under amplitude damping half its trajectory is bound magic, and the noncontextual settings of composed depolarising and dephasing form a non-convex region even though, state by state, the classical set is a finite union of convex polytopes, itself non-convex (\cref{sec:nuC}). Separately, composition can close the gap that \cref{def:classrep} opens: for two copies of the qutrit strange state at the full $\mathrm{GF}(9)$ measurement set the class-representative and exact-frame readings agree (\cref{thm:twocopy}), and the other benchmark states and random qutrit magic states behave the same way numerically. This is not Ref.~\cite{delfosse2017}'s $n\ge2$ equivalence being inherited: the ten field MUBs are the lines of a single plane over $\mathrm{GF}(9)$, where that paper's character lemma does not apply, and the agreement comes from the enumeration. Composition can also open a gap of the opposite sign, since the bare stabiliser subtheory restricted to two of the ten field MUBs is already contextual (\cref{thm:twomeas}), even though it is noncontextual with all ten, so composition supplies entangled contexts, not a uniform verdict. How much the entanglement of the contexts is doing depends on the local dimension: Ref.~\cite{wrightkunjwal2023} shows that unentangled measurements can never give a state-independent Kochen-Specker proof on a multiqubit system, while for any system with a factor of dimension three or more they build Kochen-Specker sets from product rays alone. Whether the coincidence persists at $k\ge3$ and in other dimensions is open.

Several directions remain open. The bound-magic window now has an operational witness measurable with stabiliser settings alone, but no task in which its states are a resource. Similarly, peak negativity, which orders exact-frame noise robustness where mana does not, still lacks a fixed-dimension resource theory; \cref{prop:nu}(i) and (iii) show what such a theory would have to give up. More broadly, the scenario-relativity found here is not special to stabiliser subtheories, and the same three questions can be put to any operational fragment for which the accessible statistics are known, including stabiliser codes, whose phase-space representations and magic lifetimes under displacement noise we study elsewhere~\cite{bozzettohance2026phasespace}, and the ones a decohering device supplies~\cite{aloy2026}.

\emph{Note added.} Earlier versions of this work (arXiv:2605.09558v1 and v2) proposed the Kirkwood-Dirac representation as a better witness of the transition than the Wigner function, with a threshold $p_{\mathrm{KD}}\le\pW$. That claim is withdrawn. \Cref{thm:pcrit} and the remark following it show that no Kirkwood-Dirac frame represents the subtheory non-negatively at any noise level, so there is no subtheory-wide Kirkwood-Dirac threshold to set against $\pW$. The proof of the withdrawn bound took the Wigner-positive states to be the stabiliser polytope, which \cref{sec:boundmagic} shows they are not, and took every stabiliser state to be non-negatively represented in a Kirkwood-Dirac frame built from two stabiliser bases, which holds only for the eigenstates of those bases. Ref.~\cite{Schmid2024KDSpekkensC} establishes the corrected statement, that no Kirkwood-Dirac frame represents the subtheory non-negatively, for the noiseless subtheory.

\textit{Acknowledgments---}
We thank Martin Pl\'avala for useful discussions. JRH acknowledges support from a Royal Society Research Grant (RG/R1/251590), from an EPSRC Mathematical Sciences Small Grant (UKRI3647), and from their EPSRC Quantum Technologies Career Acceleration Fellowship (UKRI1217).

\textit{Data and Code Availability---} A Wolfram Mathematica notebook reproducing all numerical results in this paper is openly available on Zenodo~\cite{BozzettoHance2026Notebook}.

\appendix
\begin{appendices}

\section{Quasiprobability representations}
\label{app:qprep}

Many equivalent formulations of quantum mechanics use quasiprobabilities in phase space~\cite{Wigner1932,Kirkwood1933KDDistr, Dirac1945KDDistr, Harriman1993HusiniQDistr}. A quasiprobability representation is a linear map from the operator space $\mathcal{L}(\mathcal{H}_d)$ to functions on an index set, real-valued in the familiar cases and complex-valued in general (Ref.~\cite{wagner2025structuretheoremcomplexvaluedquasiprobability}, Def.~II.11), normalised in the cases we use; Ref.~\cite{Brif1999PhaseSpace} gives a general construction from a Hilbert space and its symmetry group. Here we collect the Wigner~\cite{Wigner1932} and Kirkwood-Dirac~\cite{Kirkwood1933KDDistr, Dirac1945KDDistr} representations used in the main text.

\subsection{Wigner function}
The most familiar example is the Wigner phase-space representation~\cite{Wigner1932}, defined for infinite-dimensional systems by
\begin{equation}
\label{eq:wigner}
 W_\rho(q,p)=\frac{1}{(2\pi)^2}\int\int_{\mathbb{R}^2}\tr(e^{i(\xi Q+\eta P)}\rho)e^{-i(\xi q+\eta p)}d\xi d\eta
\end{equation}
with $Q$ the position operator and $P$ the momentum one. It is real, its overlap integral reproduces $\tr(\rho_1\rho_2)$, and its marginals along lines in phase space are the probability densities of the corresponding quadratures. Introducing a function $f(\xi,\eta)$ in \cref{eq:wigner} gives the more general $s$-ordered family~\cite{cahillglauber1969,Brif1999PhaseSpace}. Hudson's theorem~\cite{hudson1974} states that a pure state has a non-negative Wigner function if and only if it is Gaussian, extended to multi-mode states by Soto and Claverie~\cite{sotoclaverie1983}; \cref{sec:setting} uses its discrete analogue for odd prime dimensions, Gross's Wigner function~\cite{gross2006}, whose theorem characterises the pure states of non-negative Wigner function as exactly the stabiliser states.

\subsection{Kirkwood-Dirac distribution}

One of the most general quasiprobability representations of quantum theory is the Kirkwood-Dirac (KD) distribution~\cite{Kirkwood1933KDDistr,Dirac1945KDDistr}. The KD distribution $\varrho(\rho)$ of a state $\rho\in \mathcal{H}$ with respect to $\ket{a}\in\text{Spec}(\hat{A})$, $\ket{b}\in\text{Spec}(\hat{B})$, with $\hat{A}$, $\hat{B}$ Hermitian and Spec the eigenspectrum, is
\begin{equation}
\varrho(a,b)=\tr(\rho\ket{b}\bra{b}\ket{a}\bra{a})
\end{equation}
In the eigenbases $\{\ket{a_i}\}$, $\{\ket{b_j}\}$ of $\hat{A}$, $\hat{B}$ the distribution becomes a $d\times d$ matrix with entries
\begin{equation}
\label{eq:KDfirst}
 \varrho_{i,j}(\rho)\equiv \langle b_j|a_i\rangle\langle a_i|\rho|b_j\rangle
\end{equation}
The KD distribution is normalised and has the correct marginals:
\begin{equation}
\begin{split}
 &\sum_{i,j}\varrho_{i,j}(\rho)=1,\\
 &\sum_j\varrho_{i,j}(\rho)=\langle a_i|\rho|a_i\rangle,\\
 &\sum_i\varrho_{i,j}(\rho)=\langle b_j|\rho|b_j\rangle
\end{split}
\end{equation}
It is not in general a Kolmogorov probability distribution~\cite{Kolmogorov1950foundations}, since it can take negative or non-real values.

The distribution extends to more than two measurement operators, and to positive operator-valued measures rather than projective measurements~\cite{ArvidssonShukur2024KDDistr}; the fragments of \cref{sec:fragments} use only the two-basis form above.

Ref.~\cite{Schmid2024KDSpekkensC} extends the Kirkwood-Dirac distribution from states to arbitrary quantum processes, channels and measurements included, using the frame representation, which unifies many quasiprobability representations of quantum mechanics~\cite{Ferrie2011Quasiprobabilities}.

A frame generalises a basis by allowing linearly dependent, redundant vectors:
\begin{definition}
 Let $\{\phi_k\}_{k\in I}\subset \mathcal{H}$ be a family indexed by a countable set $I$, and let $a,b> 0$. Then $\{\phi_k\}$ is a frame for $\mathcal{H}$ if
 \begin{equation}
 a||v||^2\le \sum_{k\in I}|\langle v,\phi_k\rangle|^2\le b||v||^2,\quad \forall v\in \mathcal{H}
 \end{equation}
\end{definition}

For a finite-dimensional Hilbert space $\mathcal{H}_A$ of a system $A$ and two bases $\{\ket{a_i}\}$, $\{\ket{a'_{i'}}\}$ with $\langle a'_{i'}|a_i\rangle\neq 0$ for all $i,i'$, define a basis $F=\{F_{i,i'}\}$ for $\mathcal{L}(\mathcal{H}_A)$, regarded as a Hilbert space under the Hilbert-Schmidt inner product $\langle M,N\rangle=\tr(M^\dagger N)$, as
\begin{equation}
 F_{i,i'}=|{a'_{i'}}\rangle\langle{a'_{i'}}|\ |{a_{i}}\rangle\langle{a_{i}}|
\end{equation}
Such a basis constitutes a frame. We can define then the dual frame $D=\{D_{i,i'}\}$ as the unique basis of operators satisfying
\begin{equation}
 \tr(F_{j,j'}D_{k,k'})=\delta_{j,k}\delta_{j',k'}
\end{equation}
equivalently, the element $i,i'$ is
\begin{equation}
 D_{i,i'}=\frac{|a_i\rangle\langle a'_{i'}|}{\langle a'_{i'}|a_i\rangle}
\end{equation}
Operators of this type are well known: $\tr(E D_{i,i'})$ is the weak value of $E$ with preselection $\ket{a_i}$ and postselection $\ket{a'_{i'}}$, in the sense of Aharonov, Albert and Vaidman~\cite{aharonov1988}, as used for polarisation states in Ref.~\cite{Hance2023CCQCC}.
The frame and the dual frame satisfy $\tr(D_{i,i'})=1$ and $\tr(F_{i,i'})=|\langle a'_{i'}|a_i\rangle|^2$.
A state $\rho$ is represented as
\begin{equation}
 \mu(i,i'|\rho)=\tr(F_{i,i'}\rho)=\langle a_i|\rho|a'_{i'}\rangle \langle a'_{i'}|a_i\rangle
\end{equation}
A POVM element $E$ is represented as $\xi(E|i,i')=\tr(ED_{i,i'})$,
and a channel $\mathcal{E}$ between a quantum system $A$ and a quantum system $B$ is represented as
\begin{equation}
 \Gamma(j,j'|i,i',\mathcal{E})=\tr[F_{j,j'}\mathcal{E}(D_{i,i'})]
\end{equation}
Such representations take complex values and are not confined to the $[0,1]$ interval, but no information is lost: distinct processes have distinct $\mu$, $\xi$ and $\Gamma$, each of which determines the process it represents~\cite{Schmid2024KDSpekkensC}.

Refs.~\cite{Schmid2021, Schmid2024structuretheorem} establish that a scenario admits a generalised-noncontextual model exactly when the GPT it defines admits an ontological model, and Ref.~\cite{Schmid2024KDSpekkensC} draws the consequence for Kirkwood-Dirac representations: an everywhere-real, non-negative one is itself such a model, so exhibiting one certifies noncontextuality, whereas negativity or imaginarity of any particular one settles nothing. Ruling out a Spekkens noncontextual model therefore requires a statement about the whole class of linear, functorial and empirically adequate quasistochastic representations, taken as complex-valued functions rather than one representation at a time.

\section{Proof of the fragment criterion (\texorpdfstring{\cref{thm:criterion}}{the fragment criterion})}\label{app:criterion}

($\Leftarrow$) Given $\{G_\lambda\}$, set $\Lambda=[d]^k$, $\mu_\rho(\lambda)=\tr(G_\lambda\rho)$, $\xi_m(i|\lambda)=\delta_{i,\lambda_m}$. Non-negativity is (b); normalisation follows from summing (a) over $i$ for any fixed $m$, $\sum_\lambda G_\lambda=\id$; the Born rule is (a). Preparation noncontextuality: operational equivalence in the fragment means equal statistics on the allowed measurements, i.e.\ equal $\mathcal{P}_V\rho$; by (0), $\mu_\rho$ depends on $\rho$ only through $\mathcal{P}_V\rho$. Measurement noncontextuality: since the preparations span $\mathcal{L}(\mathcal{H}_d)$ (a ball about $\id/d$ already does), two effects are operationally equivalent iff equal as operators; by hypothesis no two of the $M_m$ share a rank-one projector, so distinct $(m,i)$ give distinct effects, and the response functions depend only on the effect.

($\Rightarrow$) Given any noncontextual model, define $\tr(G_\lambda\rho):=\int\mu_\rho(d\kappa)\prod_m\xi_m(\lambda_m|\kappa)$. The map $\rho\mapsto\mu_\rho$ is convex-linear on the preparation convex set, which is full-dimensional, since $\mathcal{P}$ contains a ball about $\id/d$, so it extends to a linear functional and $G_\lambda$ is a well-defined Hermitian operator; non-negativity of the integrand gives (b). Marginalising over all coordinates but $m$ and using $\sum_i\xi_m(i|\kappa)=1$ gives $\sum_{\lambda:\lambda_m=i}\tr(G_\lambda\rho)=p(i|m,\rho)$ for a spanning set of $\rho$, i.e.\ (a). For (0): preparation noncontextuality makes $\mu$ (hence $G$-statistics) constant on fibres of $\mathcal{P}_V$ restricted to the preparation set; since for any traceless $Y\perp V$ both $\id/d\pm\epsilon Y$ lie in $\mathcal{P}$ for small $\epsilon$, the fibre differences span $V^\perp$, so $\tr(G_\lambda Y)=0$ for all $Y\in V^\perp$, i.e.\ $G_\lambda\in V$. \hfill$\square$

\section{Proofs of \texorpdfstring{\cref{lem:ray}}{the ray lemma}, the orbit-clause theorem, and \texorpdfstring{\cref{thm:convexlinear}}{the convex-linear collapse}(1) and (2)}\label{app:ray}

\emph{\Cref{lem:ray}.} Fix $\lambda$ and let $N_\lambda$ be as in \cref{eq:ray}. For any allowed basis $m'$ and outcome $j\ne\lambda_{m'}$, unbiasedness gives
$\tr(N_\lambda\,\Pi^{m'}_{j})=0+\frac{k-1}{d}-\frac{k-1}{d}=0$,
so $N_\lambda$ satisfies all $k(d-1)$ forced-zero conditions; it lies in $V$ by construction; and $\dim V=k(d-1)+1$ while the conditions are independent (their Gram matrix is invertible by unbiasedness), so the solution ray is unique and spanned by $N_\lambda$. Normalisation: $\tr(N_\lambda s_{m,\lambda_m})=1$, so writing $G_\lambda=x_\lambda N_\lambda$, non-negativity on $s_{m,\lambda_m}$ forces $x_\lambda\ge0$. Contracting (a) with the eigenstates of the allowed bases, which span $V$, yields the marginal constraints quoted. \hfill$\square$

\emph{\Cref{thm:trichotomy}, \cref{eq:pc}.} $(\le)$ Suppose orbit $o$ is fully admissible at noise $p$, i.e.\ $n_\lambda(U\rho_M(p)U^\dagger)\ge0$ for all $\lambda\in o$ and Cliffords $U$ (equivalently, by covariance $U^\dagger N_\lambda U=N_{\pi_U^{-1}\lambda}$, the orbit-minimum clause holds). Let $c_1=\#\{\lambda\in o:\lambda_m=i\}$ and $c_2=\#\{\lambda\in o:\lambda_m=i,\lambda_{m'}=j\}$. Identifying (basis, outcome) pairs with the lines of $\mathrm{AG}(2,d)$, transitivity of $\mathrm{ASL}(2,d)$ on lines makes $c_1$ independent of $(m,i)$. For $c_2$ that is not enough, and one needs transitivity on \emph{ordered pairs of non-parallel} lines. $\mathrm{ASL}(2,d)$ has it: two non-parallel lines meet in a unique point, the translations are transitive on points, and the stabiliser of a point is $\mathrm{SL}(2,d)$, which acts on the $d+1$ parallel classes as $\mathrm{PSL}(2,d)$ does on $\mathrm{PG}(1,d)$, $2$-transitively. So $c_2$ is independent of $(m\ne m',i,j)$ too, and summing the pairwise count over $j$ gives $c_2\,d=c_1$. Hence the constant weighting $x_\lambda=1/c_1$ on $o$ satisfies both marginal families, and \cref{thm:criterion} provides a model; representing each Clifford by its induced permutation $\pi_U$ of $\Lambda$ (after symmetrising $G$ over the group, which preserves (0),(a),(b) since the preparation and measurement sets are Clifford-invariant) extends it to the transformations. This assigns one $\Gamma$ per transformation channel, which is the reading of \cref{def:classrep}. Mixtures of Cliffords equal as channels act identically on all reachable epistemic states, which makes the assignment consistent with the operational statistics; it does not make the two mixtures' convex-linear representations equal as kernels. Reachability constrains $\Gamma$ only on the $d^2$-dimensional image of $\rho\mapsto\mu_\rho$, where the Weyl and Clifford averages agree, whereas convex-linearity constrains it on the whole support, where they do not. So the thresholds of \cref{sec:qubit}, and the non-collinear models of \cref{sec:disagree}, are read in the class-representative sense, which by \cref{thm:convexlinear} is the only sense available to them. Under Spekkens' own convex-linearity assumption~\cite{Spekkens2005GenContext} the single-qubit stabiliser subtheory admits no noncontextual model at all \cite{lillystone2019}, and \cref{cor:orbitcriterion} recovers that from the same criterion. Only this direction of \cref{eq:pc} uses \cref{def:classrep}. $(\ge)$ Conversely any model of the Clifford scenario restricts to the fragment with the orbit of $\rho_M(p)$ as extra preparations; symmetrising its $G$'s over the group gives a solution constant on orbits, whose support is a union of fully admissible orbits; at least one orbit is therefore fully admissible. This half uses no transformation procedures, only the enlarged preparation set, so the two halves together identify $\min_o p_o$ with the noncontextuality threshold of the orbit-closed prepare-and-measure fragment as well as with $\pC$. \hfill$\square$

\emph{\Cref{eq:pstab}.} $(\ge)$: each $s_{m,i}$ has $n_\lambda(s_{m,i})=\delta_{i,\lambda_m}\ge0$, so every state of the polytope satisfies all assignment inequalities, whence along the trajectory $\pstab\ge\max_o p_o$. $(\le)$: by Cormick \textit{et al.}'s Theorem~1~\cite{cormick2006} the assignment inequalities cut out exactly the stabiliser polytope in every power-of-prime dimension, so $\rho_M(p)$ lies in it as soon as every clause is satisfied, whence $\pstab\le\max_o p_o$. We confirm this independently at $d=3$ by exhaustive enumeration over vertex subsets, which shows the polytope has exactly $81$ facets, each the hyperplane of an assignment (the facet through the eight ``non-selected'' vertices of $\lambda$), and at $d=5,7$ numerically, along $24$ and $12$ Haar-random central rays respectively, where the assignment maximum and the polytope entry point agree to within $2\times10^{-8}$. \hfill$\square$

\emph{\Cref{prop:nu}, parts (i) and (ii).} $\nu$ is a maximum of affine functionals; Cliffords permute the $A_u$; $t$-values of a product are products of $t$-values, and a pure stabiliser state has $t_v\in\{0,1\}$ with $t_{\max}=1$; and a forgotten stabiliser measurement is the uniform twirl over the $d$ Weyl operators generating that basis's stabiliser group, hence a mixture of Cliffords (equivalently, its output is a mixture of stabiliser states, so $\nu$ of the output vanishes); $\nu(\sigma)=0$ whenever $W_\sigma\ge0$, so convexity also gives the mixing clause. The discard clause fails because $\tr(A_u\,\tr_2\rho)=\tfrac1d\sum_v\tr(A_u{\otimes}A_v\,\rho)$ is $d$ times an average of composite $t$-values, not an average, so the marginal's peak can exceed the composite's. \hfill$\square$

\emph{Proof of \cref{thm:convexlinear}, clauses (1) and (2).}

(1) The Clifford-orbit average of any state is $\id/d$, since the traceless part of $\mathrm{Herm}(\mathcal{H}_d)$ is an irreducible representation of the Clifford group, as (2) shows; $\id/d$ is also the uniform mixture of the eigenstates of each stabiliser basis, so preparation noncontextuality assigns it a single epistemic state $\mu_{\id/d}$, and ``prepare $\id/d$'' is operationally equivalent to ``prepare $\id/d$, then apply $U$'', so $\Gamma_U\mu_{\id/d}=\mu_{\id/d}$. Every preparation of the scenario is dominated by $\mu_{\id/d}$. For the stabiliser states and their mixtures this is immediate, since $\id/d$ is the uniform mixture of the eigenstates of each basis. It holds for $\rho_M(p)$ below $\pstab$, where $\rho_M(p)$ is not such a mixture, because ``prepare $\rho_M(p)$, then apply a Clifford drawn uniformly at random'' is a procedure of the scenario operationally equivalent to preparing $\id/d$, giving $\mu_{\id/d}=|\mathcal{C}_d|^{-1}\sum_U\Gamma_U\mu_{\rho_M(p)}$ with every term non-negative and the $U=\mathrm{id}$ term equal to $\mu_{\rho_M(p)}$. Write $\Lambda_0$ for the support of $\mu_{\id/d}$. Next, $\tr(\Pi^m_i s_{m,i})=1$ with $\xi_{m,i}\le1$ forces $\xi_{m,i}=1$ almost everywhere on the support of $\mu_{s_{m,i}}$, and $\tr(\Pi^m_i s_{m,j})=0$ forces $\xi_{m,i}=0$ almost everywhere on the support of $\mu_{s_{m,j}}$; since $\mu_{\id/d}$ is the average of the $\mu_{s_{m,i}}$ over $i$, the sets $\Delta_{m,i}$ on which $\xi_{m,i}=1$ are disjoint in $i$ and cover $\Lambda_0$, and $\xi_{m,i}=\mathbf 1_{\Delta_{m,i}}$ almost everywhere. Their common refinement over the $d{+}1$ bases partitions $\Lambda_0$ into blocks indexed by value assignments. Each Clifford $U$ maps $\Pi^m_i$ to another stabiliser effect $\Pi^{\alpha(m,i)}=U^\dagger\Pi^m_iU$, so empirical adequacy gives $\int\Gamma_U(\lambda',\lambda)\xi_{m,i}(\lambda')\,d\lambda'=\xi_{\alpha(m,i)}(\lambda)$, and taking the preparations to be the deterministic stabiliser states pins this pointwise in $\lambda$ rather than in integral; non-negativity of the kernel $\Gamma_U(\lambda,\cdot)$ then leaves it no mass outside the single block $\pi_U(\lambda)$. Blocks of $\mu_{\id/d}$-measure zero are null for every preparation, and $B$ is $\pi_U$-invariant because $\Gamma_U$ preserves $\mu_{\id/d}$. None of this uses finiteness of $\Lambda$.

(2) Convex-linearity requires $\sum_U v_U\Gamma_U=0$ whenever $\sum_U v_U\,U(\cdot)U^\dagger=0$, and evaluating on blocks gives $\sum_U v_U\,\pi_U|_B=0$. The algebra map from the group algebra of the Clifford group to $\mathrm{End}(\mathbb{R}^B)$ therefore factors through the image $\mathcal{A}$ of that algebra in $\mathrm{End}(\mathrm{Herm}(\mathcal{H}_d))$. Now $\mathrm{Herm}(\mathcal{H}_d)$ is the trivial representation plus the adjoint one, and the adjoint is absolutely irreducible: the $d^2-1$ Weyl operators $\mathcal{D}_b$ with $b\ne0$ are a simultaneous eigenbasis for the translations, with pairwise distinct characters, and $\mathrm{SL}(2,d)$ permutes those $b$ transitively. So $\mathcal{A}=\mathbb{R}\oplus M_{d^2-1}(\mathbb{R})$, and every irreducible constituent of the permutation representation on $B$ is either trivial or adjoint. An orbit is transitive, hence contains the trivial representation exactly once, so $|o|\equiv1$ modulo $d^2-1$. By \cref{lem:freetrans}, $|o|=d^2k$ with $k$ dividing $|\mathrm{SL}(2,d)|=d(d^2-1)$. Since $d^2\equiv1$, the congruence reads $k\equiv1$ modulo $d^2-1$; so $\gcd(k,d^2-1)=1$, and $k$ divides $d$. As $1\le k\le d<d^2-1$, the only $k$ congruent to $1$ is $k=1$, and $|o|=d^2$.

\section{Random-state survey}\label{app:survey}

\begin{figure*}[!t]
\includegraphics[width=\textwidth]{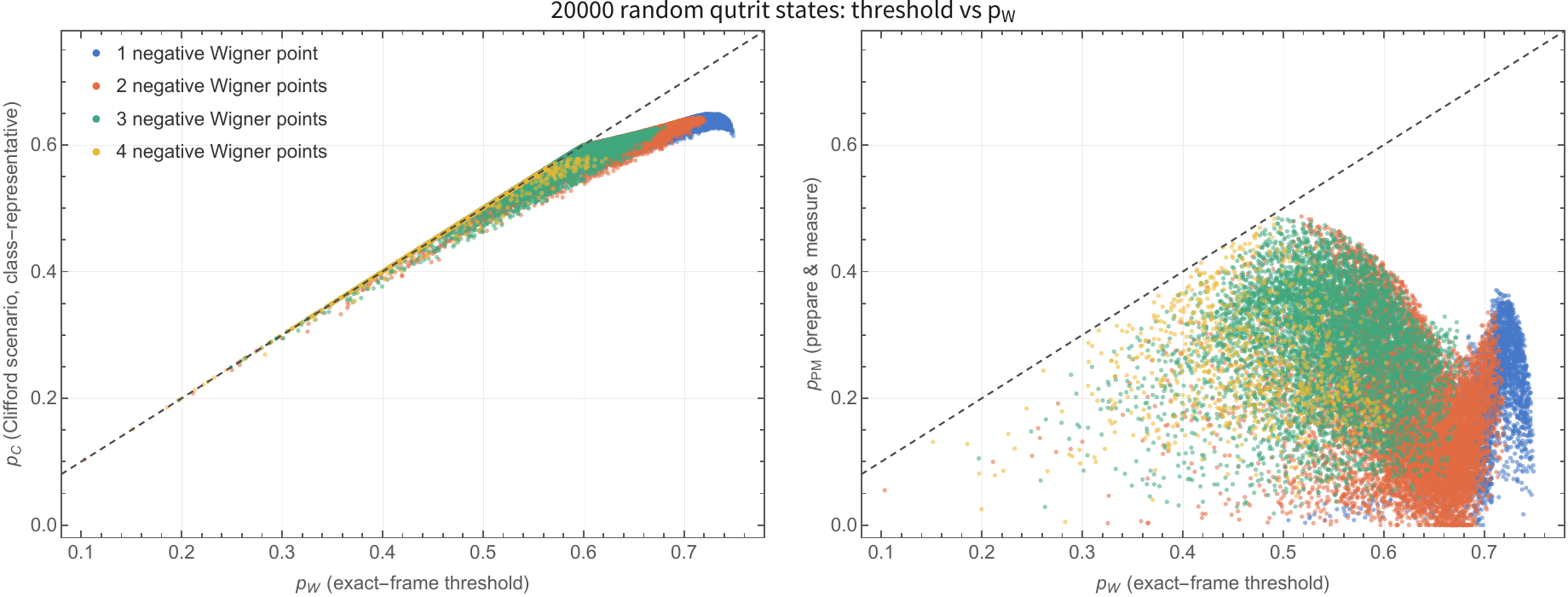}
\caption{\label{fig:survey}Fragment thresholds versus the exact-frame threshold $\pW$ for $20\,000$ Haar-random pure qutrit states, coloured by the number of negative Wigner points. In both panels the horizontal axis is $\pW$ and the dashed line is the diagonal, on which the two thresholds coincide. Left: the vertical axis is the class-representative Clifford-scenario threshold $\pC$, which lies on or just below the diagonal ($\pC=\pW$ for $25.0\%$ of states) until the two-clause formula \eqref{eq:twoclause} caps it at $(21-3\sqrt5)/22$, the horizontal ceiling in the point cloud at the right of the panel. Under the convex-linear reading the left panel would be the diagonal identically, by \cref{thm:convexlinear}. Right: the vertical axis is the bare prepare-and-measure threshold $\pPM$, far below the diagonal for typical states (the mean of $\pPM/\pW$ is $0.375$).}
\end{figure*}

\Cref{fig:survey} places the benchmark states in context. Across $20\,000$ Haar-random pure qutrit states, $\pPM=\pW$ occurred for none (and $\pPM=0$ for only $0.4\%$), so the benchmark state $\ket b$, for which all three thresholds coincide, sits on a measure-zero locus. By contrast, $\pC=\pW$ holds on a positive-measure region ($25.0\%$ at $d=3$; $34\%$ in a $4000$-state $d=5$ sample), with probability increasing strongly with the number of negative Wigner points ($0\%$, $11\%$, $53\%$, $81\%$ for $1,2,3,4$ points). The two-clause formula explains the pattern: $\pC=\pW$ iff the collinear clause is the minimum, which favours states with dispersed negativity; concentrated negativity gives a large $\pW$ that the non-collinear clause then undercuts. For a typical qutrit state the bare prepare-and-measure threshold is about $0.37$ of the exact-frame one. The ratio the other way round, $\pW/\pPM$, has no finite mean, since $\pPM$ vanishes on a set of positive measure, so the comparison is quoted in this direction only.

\section{Experimental requirements for the two witnesses}\label{app:experiment}

The bound-magic witness of \cref{sec:boundmagic} needs no settings beyond the subtheory's own. The $d+1$ stabiliser bases are tomographically complete, so the same settings that evaluate $s_{\lambda^*}$ reconstruct the entire Wigner function through the line sums of \cref{eq:linesum}, and one experiment certifies both halves of bound magic, $W\ge0$ and $s_{\lambda^*}<1$. For the qutrit $\ket b$ at mid-window, $p=0.43$, the two margins are $W_{\min}=0.014$ and $v=0.060$, so, at worst-case per-setting variance, some $2\times10^4$ shots per setting certify both at five standard deviations, and the score alone needs only ${\sim}2\times10^3$ shots per setting at the window opening. Those are statistical requirements only, and the systematic requirement is the binding one: a per-outcome probability error below $v/(d+1)\approx1.5\%$. We know of no published qutrit figure that meets it end to end. Ref.~\cite{ringbauer2022}'s qutrit Clifford error of $2(2)\times10^{-3}$ and detection error of $3\times10^{-3}$ are both quoted with state preparation and measurement decoupled, while its single-qutrit process tomography without SPAM correction runs from $0.6\%$ to $5.5\%$ infidelity, the largest being the qutrit Fourier gate, which is the gate class a rotation into another stabiliser basis needs. Ref.~\cite{Giordani2023Experimental} reports, in its $d=4$ set, a value $1.391\pm0.011$ above that quantity's own quantum maximum of $1.375$, and attributes the excess to circuit-setting systematics of $2.5$ to $10$ times the counting uncertainty. The witness is therefore calibration-limited rather than shot-limited, and a five-standard-deviation statistical statement means little without a systematic budget of the same order. The depolarising trajectory itself can be implemented by Weyl-twirling the pure preparation, with the randomised-compiling overhead that implies. Two qualifications apply. First, in the window the state is noncontextual, so the experiment certifies magic without distillability rather than contextuality. Second, the polytope bound $s_\lambda\ge1$ presumes ideal stabiliser settings, so a device-honest analysis should follow Ref.~\cite{mazurek2016}, building the operational equivalences into secondary procedures or replacing the facet by its noise-robust linear-programming dual; the secondary procedures multiply the shot budget by roughly their number, which for the $d{+}1$ stabiliser settings is a factor of several.

The same twelve probabilities certify the contextual side of the threshold. By \cref{thm:trichotomy} no class-representative model of the Clifford scenario exists iff every orbit contains a violated assignment, so for $d=3$ the certificate is a pair of inequalities: some line sum below $1$ (a negative Wigner point) and some non-collinear score below $1$. For the noisy strange state the two margins are $-3W_{\min}(p)=1-\tfrac43p$ and $v_C(p)=\tfrac12-\tfrac56p$, the second being the binding one at every $p<1$, closing at $3/5$. The demands are milder than for the bound-magic witness: at $p=3/10$, halfway to the threshold, $v_C=\tfrac14$, so ${\sim}4\times10^2$ shots per setting, again at worst-case per-setting variance, violate both clauses at five standard deviations, with a systematic per-probability tolerance of $v_C/4\approx6\%$. Here the qualification is different: the strange state's bare fragment is never contextual ($\pPM=0$), so the inference leans on the scenario containing the ideal Clifford transformations, meaning a device-honest test should certify those too by secondary procedures~\cite{mazurek2016}. Under the convex-linear reading the second inequality is not needed at all: by \cref{thm:convexlinear} a single line sum below $1$ certifies contextuality, out to $\pW=3/4$.

\Cref{eq:pcwitness} of \cref{sec:boundmagic} is the inequality to aim a device-honest test at, and the fragment behind it is small: for the strange state the Clifford orbit has nine members, one per phase point, its stabiliser is the $24$ Cliffords fixing the origin, and it averages to $\id/3$; solving the fragment in exact rational arithmetic puts its threshold at $3/5$, as the remark after \cref{thm:trichotomy} says. The single inequality $s_{\lambda^*}\ge1$ locates the same threshold along this trajectory but is not a noncontextuality inequality for the fragment: a Clifford image of $\ket b$ at $p=0.43$, chosen so that the same $\lambda^*$ binds, is Wigner-positive, hence noncontextual, and still has $s_{\lambda^*}=0.9404$. \Cref{eq:pcwitness} is $\tfrac85$ times the average of those eight inequalities, which is what makes it valid under either clause. Its statistics are the better of the two: at the crossing each of the four frequencies is $\tfrac15$, giving the left side a standard error of $\tfrac45n^{-1/2}$ against a slope of $\tfrac43$ in $p$, so $n=3.6\times10^3$ shots per basis resolve $p$ to $10^{-2}$, against $1.0\times10^4$ for $s_{\lambda^*}$.

Its device-honest version needs two families of operational equivalences rather than one. Beyond the four basis mixtures, each orbit member needs the two-design identity $\tfrac14\rho_u+\tfrac34\,\id/3\simeq\sum_\ell\tfrac14p(\ell\mid\rho_u)\Pi_\ell$, which holds because the $d+1$ stabiliser bases form a two-design, and which is what forces the ray form. The orbit average will not do on its own: given only the basis equivalences and $\tfrac19\sum_u\rho_u\simeq\id/3$, the fragment is noncontextual at every $p$. With all thirteen equivalences the threshold is $3/5$ again. Secondary preparations then come from a linear program in the tomographic estimates of the twenty-one primaries, with non-negative mixing weights, the equivalences imposed as equalities of the estimated statistics, and the weight on each intended primary maximised, following Ref.~\cite{mazurek2016}. If each equivalence holds only to within total variation $\delta$, the data stay contextual while $\delta<\delta_{\min}(p)$, itself an exact linear program: $\tfrac18$ at $p=0$, $0.029$ at $p=\tfrac3{10}$, $0.009$ at $p=\tfrac12$ and $9\times10^{-4}$ at $p=0.59$, vanishing at $\pC$. At $p=\tfrac3{10}$ that is $2.9\%$, the same order as the per-probability tolerance above, so both witnesses of this appendix are calibration-limited in the same way.

The construction is not special to the qutrit, but its reach is. At $d=5$ the remark after \cref{thm:trichotomy} again holds, on the hardest state in \cref{tab:maxpC}: its transformation-free fragment, now with $1030$ preparations against $21$, crosses at exactly $\max_\rho\pC$. Averaging an orbit's clause over the $\mathrm{Stab}(u)$-orbit of an assignment with $c_0$ of its lines through $u$ gives $W_\rho(u)\ge(c_0-d-1)/\big(d(c_0d-d-1)\big)$, and the worst orbit's bound is the witness; at $d=5$ that is again $-\tfrac1{15}$, which is a coincidence of these two dimensions, the same construction giving $-\tfrac5{91}$ at $d=7$. The line-sum form becomes $\sum_{\ell\ni u}p(\ell\mid\rho)\ge\tfrac23$, with noncontextuality requiring a success probability of $\tfrac19$ against the polytope's $\tfrac16$. What does not carry over is its force. \Cref{eq:pcwitness} depends on the state only through $W_\rho(u)$, so it sees peak negativity and nothing else; at $d=3$ that loses nothing, because the $\pC$- and $\pW$-maximising states coincide, while at $d=5$ they do not, and the witness closes at $0.4480$ against $\pC=0.7785$. This is a ceiling for the whole symmetry class rather than a defect of the construction: $\mathrm{Stab}(u)=\mathrm{SL}(2,d)$ has two orbits on phase points, $\{u\}$ and the rest, so every $\mathrm{Stab}(u)$-invariant linear functional of the Born data is affine in $W_\rho(u)$. A witness that reaches $\pC$ beyond the qutrit must average over a smaller group, and will not be a line sum through one point.

\section{Continuous-time rates}\label{app:rates}

In continuous time the plane becomes a statement about rates, which is the form in which the hierarchy is easiest to compare with an experiment. The two noise families are commuting semigroups under composition, so running them at rates $\gamma_{\mathrm{dep}}$ and $\gamma_{\mathrm{deph}}$ gives $p(t)=1-e^{-\gamma_{\mathrm{dep}}t}$ and $q(t)=1-e^{-\gamma_{\mathrm{deph}}t}$, and \cref{eq:noiseplane} turns into closed-form onset times for noncontextuality: the Clifford scenario of the strange state has a class-representative model once $e^{-\gamma_{\mathrm{dep}}t}\big(2+3e^{-\gamma_{\mathrm{deph}}t}\big)\le2$, and Wigner-positive (equivalently, inside the stabiliser polytope) once $e^{-\gamma_{\mathrm{dep}}t}\big(1+3e^{-\gamma_{\mathrm{deph}}t}\big)\le1$. Pure depolarising gives $t_C=\ln(5/2)/\gamma_{\mathrm{dep}}$; equal rates $\gamma$ give $t_C=\ln\!\big(3/(\sqrt7-1)\big)/\gamma$ and $t_W=\ln\!\big(6/(\sqrt{13}-1)\big)/\gamma$; pure dephasing never reaches noncontextuality, matching \cref{tab:noises}. The onset time is monotone in each rate, and in the dephasing fraction at fixed total rate, so there is no rate-mixing analogue of the restoration effect. The divergence at the dephasing end is only logarithmic: for $\gamma_{\mathrm{dep}}\ll\gamma_{\mathrm{deph}}$, $t_C\simeq\gamma_{\mathrm{deph}}^{-1}\ln(\gamma_{\mathrm{deph}}/\gamma_{\mathrm{dep}})$, so a trace of white noise collapses the infinite dephasing time to a few dephasing times. So a device with comparable $T_1$ and $T_2$ sits far from the pure-dephasing corner.

\section{Symbolic and numerical methods}\label{app:numerics}

All results of \cref{sec:setting,sec:readings,sec:frames,sec:fragments,sec:assignments,sec:sixstates,sec:dephasing,sec:robust,sec:qubit,sec:twoqutrit} come from a single companion pipeline in the Wolfram Language (Mathematica 15.0). The computational core is the criterion of \cref{thm:criterion} in its ray-reduced form (\cref{lem:ray}). This is a linear program, and because admissibility is monotone in the noise, each threshold is attained at one of finitely many algebraic values (the points at which individual tuples become admissible) and is located exactly by bisection over that finite candidate set rather than by numerical search.

\begin{table}[tbp]
\caption{\label{tab:minpolys}Minimal polynomials of the algebraic quantities quoted in the text, and the root meant in each case. Depolarising strengths are written $p$, dephasing strengths $q$, and critical clause values $\tau$. $^{\rm a}$Smallest real root. $^{\rm b}$Seventh of eight real roots, increasing order. $^{\rm c}$Middle of three real roots. $^{\rm d}$Fourth of eight real roots, increasing order. $^{\rm e}$Second of four real roots, increasing order; the larger of the two in $(0,1)$.}
\begin{ruledtabular}
\footnotesize
\begin{tabular}{>{\raggedright\arraybackslash}p{0.20\columnwidth}c>{\raggedright\arraybackslash}p{0.40\columnwidth}l}
quantity & $d$ & minimal polynomial & root \\
\colrule
$\pstab(\ket b)$, depol.\ window close & $3$ & $199p^4-606p^3$ $+459p^2+54p-81$ & $0.516485$ \\
$\pPM=\pC=\pW$ for $\ket b$, dephasing & $3$ & $4q^4-16q^3$ $+18q^2-4q-1$ & $0.562984$ \\
$\pstab(\ket b)$, dephasing entry & $3$ & $41q^4-120q^3$ $+96q^2-16$ & $0.720396$ \\
$\max_\rho\pC$ & $3$ & $11p^2-21p+9$ & $(21-3\sqrt5)/22$ \\
$\pPM(\ket a)$ of \cref{tab:six} & $5$ & $1019p^4-9740p^3$ $+25050p^2$ $-21500p+3875$ & $0.242407^{\rm a}$ \\
$\pstab(k)$, depol.\ window close & $5$ & $841p^4-3900p^3$ $+6450p^2$ $-4500p+1125$ & $0.786481^{\rm e}$ \\
critical clause value $\tau^*$ & $5$ & $125\tau^8+125\tau^7$ $-375\tau^6-100\tau^5$ $+325\tau^4-50\tau^3$ $-55\tau^2+15\tau-1$ & $0.702812^{\rm b}$ \\
$\max_\rho\pC$ & $5$ & $16479p^8-118980p^7$ $+365025p^6-621525p^5$ $+642125p^4-411875p^3$ $+160000p^2-34375p$ $+3125$ & $0.778470^{\rm d}$ \\
critical clause value $\tau^*$ & $7$ & $7\tau^3-7\tau^2+1$ & $0.591009^{\rm c}$ \\
$\max_\rho\pC$ & $7$ & $41p^3-140p^2+147p-49$ & $0.805336^{\rm a}$ \\
\end{tabular}
\end{ruledtabular}
\end{table}
\begin{table}[tbp]
\caption{\label{tab:ontic}An explicit generalised-noncontextual model of the strange state's four-basis prepare-and-measure fragment at $p=0$, in the form of \cref{lem:ray}: the $18$ outcome tuples $\lambda=\lambda_0\lambda_1\lambda_2\lambda_3$ of non-zero weight, each with $x_\lambda=\tfrac16$, all other weights zero. The coordinates $\lambda_0\lambda_1\lambda_2\lambda_3$ are the outcomes in the eigenbases of $Z$, $X$, $XZ$ and $XZ^2$ respectively. Outcome $i$ of the computational basis is $\ket i$; outcome $i$ of the $m$th of the other three is the vector with components $\omega^{2(m-1)x^2+ix}/\sqrt3$, $x\in\mathbb{Z}_3$ (the ordering used in the companion notebook). The weights sum to $d=3$ and satisfy the marginal conditions of \cref{lem:ray} exactly, and every listed tuple has $s_\lambda(\rho_a)\in\{1,2\}$, so $\mu_{\rho_a}\ge0$.}
\begin{ruledtabular}
\begin{tabular}{cccccc}
$0012$ & $0021$ & $0102$ & $0110$ & $0201$ & $0220$ \\
$1001$ & $1020$ & $1111$ & $1122$ & $1200$ & $1212$ \\
$2002$ & $2010$ & $2100$ & $2121$ & $2211$ & $2222$ \\
\end{tabular}
\end{ruledtabular}
\end{table}

\emph{Exact results.} The following are certified symbolically over the relevant cyclotomic and radical fields (\texttt{RootReduce}, \texttt{MinimalPolynomial}); machine or extended precision appears only in the auxiliary role of locating a candidate minimiser, which is then verified exactly.

\begin{itemize}[leftmargin=*,itemsep=1pt,topsep=2pt]
\item \emph{Polytope and facets.} The Gram spectra behind \cref{lem:facet}, both from the abstract overlap structure (checked at every odd prime $d\le13$) and from the explicit MUB projectors at $d=3,5$; and the $81$-facet enumeration of the qutrit stabiliser polytope, which is exhaustive.
\item \emph{Thresholds and states.} All entries of \cref{tab:six}, with the minimal polynomials collected in \cref{tab:minpolys}. The $\pPM(\ket b)$ entry at $d=3$ is settled by combining two exact steps: matched feasibility and infeasibility certificates give $\pW-10^{-3}<\pPM\le\pW$, and an exhaustive search over the tuple-admissibility candidates returns none strictly inside that interval, so $\pPM(\ket b)=\pW(\ket b)$ exactly. At $d=5$ no bracket is needed for the \cref{tab:six} entry $\pPM(\ket a)$: the fragment LP is feasible at that candidate and infeasible at every tuple-admissibility candidate strictly below it. Also every closed form and minimal polynomial quoted in the text, including $\pC(\ket b)=\pW(\ket b)$, the equality of $\pstab(\ket b)$ with the non-collinear clause, and the golden-ratio spectrum $(\varphi,0,1-\varphi)$ of the critical qutrit ray.
\item \emph{The convex-linear collapse.} The finite ingredients of \cref{thm:convexlinear}, in exact integer arithmetic at $d=3$: the rank computations of \cref{sec:convexlinear}; the identity $\pi_U|_{\mathrm{coll}}=$ the matrix of $U(\cdot)U^\dagger$ in the basis $\{A_u\}$, checked for all $216$ Cliffords; the Weyl-versus-Clifford witness, whose two stochastic maps differ by $1/9-1/72$; and the forcing of the collinear weights to $x=1/3$. The orbit-size census of \cref{sec:convexlinear} is exhaustive at $d=3,5,7$, and agrees with \cref{prop:orbituniq}; the same enumeration at $q=2$ and $q=4$ returns $q$ orbits of size $q^2$ for the affine action on assignments, confirming that the odd-characteristic hypothesis there is used. The sign-assignment action of the next sentence is a different one. The same pipeline at $d=2$ returns one orbit of $8$, spans of dimension $10$ and $20$, and a witness differing by $1/8$, which is \cref{cor:orbitcriterion}'s qubit instance. Ref.~\cite{schmid2024ptm}'s linear program (their Formulation F1) was also implemented, with a certificate for every verdict, checked in exact arithmetic. It reproduces their published qubit example, a $260\times256$ constraint matrix, infeasible with a Farkas certificate and feasible once their transformation constraint is dropped; that Farkas vector lifts to the $1968\times1536$ program of the full single-qubit stabiliser subtheory and is verified there, which certifies that subtheory transformation contextual. At $d=3$ the program was built with no \cref{thm:convexlinear} reduction, over the $12$ stabiliser projectors together with $\rho_M(p)$, the $12$ stabiliser effects and all $216$ Cliffords as channel matrices in $\{A_u\}$, with $\ker T$ of dimension $151$. Its source vertices are the vertices of the polytope of those thirteen states, so its size varies with $p$: $1\,417\,176$ variables at $p\ge3/4$, $2\,099\,520$ at $p=0.55$. It was decided by streaming the constraint matrix block by block, the streaming first checked against the dense construction on the qubit example. With the transformation constraint the program is infeasible at $p=0$, $1/2$, $0.55$, $0.65$ and $0.74$, each with an exact Farkas certificate, and feasible at $3/4$ and $0.80$, so its threshold is $3/4$. Since F1 returns necessary conditions, the infeasible side is the informative one: it establishes $\pCL\ge3/4$ at $d=3$ without appeal to \cref{thm:convexlinear}, the matching model above the crossing being Gross's. Without that constraint it is feasible at every $p$, \cref{tab:ontic}'s eighteen tuples being admissible throughout, since dropping it also drops any requirement that the support be closed under the Clifford orbit; $\pC=3/5$ needs that closure, which \cref{def:classrep} keeps and Formulation F1 does not impose. We also ran the orbit description that \cref{thm:convexlinear,thm:trichotomy} license, in which each Clifford acts by its permutation of the $81$ value assignments: for the strange state the collinear orbit is fully admissible from $p=3/4$ and the non-collinear orbit from $p=3/5$. That second run follows from \cref{thm:convexlinear} and is not an independent test of it; the F1 run above is.
\item \emph{Attainment of the $\pC$ maximum.} The critical orbits, their representatives and rays, and the two sign tests of \cref{sec:robust}, all over the ring generated by $\omega$ and $\tau^*$; and the minimal polynomials of $\max_\rho\pC$ quoted there. 
\item \emph{Bound magic.} The binding assignments of \cref{sec:boundmagic}, their degeneracies, and the minimal polynomials of the binding scores and violations; the maximal-window family, its endpoints $(15-6\sqrt3)/13$ and $(12-3\sqrt3)/13$, and the width $3(\sqrt3-1)/13$.
\item \emph{Other noise models.} The dephasing and amplitude-damping thresholds of \cref{sec:dephasing}, including the identity $\tr(A_0\rho(q))=-(1-q)$, the affinity of the damped strange-state trajectory, and the coincidence of the dephasing $\pPM$, $\pC$ and $\pW$ for $\ket b$; and the noise-plane boundaries of \cref{eq:noiseplane}, verified on a rational grid together with the score-class decomposition, the plane-wide fragment feasibility, and the non-convexity witness.
\item \emph{Peak negativity and qubits.} The tensor-rule and ancilla identities of \cref{prop:nu} on the pair (strange, $\ket b$), including the discard counterexample $\nu(\rho\otimes\id/d)=\nu(\rho)/d$; and the $T$ and $H$ instances of \cref{thm:qubitppm}.

\item \emph{The class-representative functional.} The non-convexity witness of \cref{prop:nuC}(iii) in exact rational arithmetic: $\tau(\sigma_1)=(0,\tfrac14)$, $\tau(\sigma_2)=(\tfrac15,0)$, the identity $\tfrac49\sigma_1+\tfrac59\sigma_2=\tfrac13\id-\tfrac29(\ket1\bra2+\ket2\bra1)$, and $\tau(\sigma)=(\tfrac19,\tfrac19)$ with $\tfrac19$ the minimum over each orbit; and $\max_\lambda\|N_\lambda\|_\infty=\varphi$.

\item \emph{Experimental witnesses.} For the enlarged fragment of \cref{app:experiment}: the Clifford orbit of $\rho_M(p)$ has nine members with stabiliser of order $24$ and average $\id/d$; the fragment linear program crosses at exactly $3/5$, with a rational Farkas certificate of infeasibility below it, verified on the full program, and the uniform weight $1/24$ on the $72$ non-collinear assignments as a model at $3/5$. \Cref{eq:pcwitness}'s bound $\tfrac45$ is the minimum of $3W(u)+1$ over the non-collinear clause region, by exact linear programming. With only the $d{+}1$ basis equivalences and the orbit average the fragment is noncontextual at every $p$; with the nine two-design identities added the threshold is $3/5$ again. The residual tolerances $\delta_{\min}(p)$ are an exact linear program, $\tfrac18$ at $p=0$ down to $0$ at $\pC$. At $d=5$: the eleven orbit clauses of the $\pC$-maximiser, its Clifford orbit of $1000$ states with stabiliser of order $3$ and average $\id/5$, the crossing at $\max_\rho\pC$ certified by one within-orbit and ten across-orbit sign tests rather than by a linear program, and each orbit's averaging bound checked against an exact linear program over its clause region, the two agreeing for all eleven.
\item \emph{Two qutrits.} The two-copy analysis of \cref{thm:twocopy}: the ten $\mathrm{GF}(9)$ MUBs are verified mutually unbiased in exact arithmetic, the reduction of the score to the pair $(c_0,\Sigma_\lambda)$ and the identification of the $69$ classes are symbolic, the eight breakpoints of $\Phi_\lambda$ are exact algebraic numbers, and the searches at the eight rational interval test values, at each of the eight breakpoints, and at $x=0$ are over exact integer tables, the last returning exactly the $81$ collinear assignments; the pruning rule is \cref{eq:prune} below. At the six quadratic-irrational breakpoints the admissible set of the $69$ classes is decided symbolically and then replaced by a rational half-plane $\Sigma_\lambda\le\alpha c_0+\beta$ that agrees with it on all $69$ classes, so the search tables at those points are again exact integers. The $k$-copy criterion and its per-direction minima are likewise exact, as are the perturbation margins at $k=3$, which are maxima over all $3\,582\,306$ images of the cross; that count is $|\mathrm{ASL}(2,27)|/4$, the setwise stabiliser of the cross being the cyclic group of order four generated by the coordinate swap, and the margins follow from the incidence identity $\sum_m|\ell_m(P)\cap C_k|=|C_k|+3^k[P\in C_k]$ without visiting an image. Of the pair searches for the other magic states, only the Norrell state's runs on exact integer tables. The parity state $\ket b$ has Wigner values in $\mathbb{Q}(\sqrt2,\sqrt5)$, and the random states, though rational, scale to integers past machine range, so both use a certified enclosure: the set it returns provably contains every solution, and the collinear orbit's admissibility follows from $N_\lambda=A_P$ rather than numerically. The $45$-pair census was carried out twice: exhaustively over all $360$ stabiliser projectors at machine precision (described below), and over the $198$ distinct product and MUB candidate witnesses in cyclotomic arithmetic with no tolerances (the $144$ product states and the $90$ $\mathrm{GF}(9)$-MUB states share $36$, the four $\mathrm{GF}(9)$ MUBs that are product bases). The latter run also covers the phase-point determinism check of \cref{sec:twoqutrit} and certifies the $12$ entangled census witnesses. The two-qutrit certificate of \cref{thm:twomeas} gives $\mu_s=-1/81$.
\item \emph{Polytope-entry certificates.} At each quoted entry point (the two window closings, the dephasing entry for $\ket b$, and the strange state at $3/4$) an explicit stabiliser-mixture decomposition is verified in exact arithmetic, with supports of $6$, $18$, $6$ and $8$ vertices and weights that are algebraic numbers of degree at most $4$, residual and normalisation exactly zero, and positivity certified by rigorous numerical enclosures.
\item \emph{Linear programs.} The $d=5$ fragment bound $\pPM\ge0.5378$ is bracketed by a feasible model at the crossing (support $219$) and a rational Farkas dual of infeasibility just below it. The dual is rationalised and then checked in exact arithmetic, the constraint matrix being $0/1$; the witness state, however, is a machine-precision seesaw output with no closed form, the crossing is located numerically, and the column set the dual rules out is selected by floating-point clause values at $t^*-10^{-7}$. The bound is therefore certified only relative to numerically chosen inputs, not end to end. At $d=3$: feasibility of the fragment of the optimal state $\tfrac13(2\omega^2\ket0-\ket1+2\ket2)$ at exactly $p=\tfrac12$ with infeasibility at $\tfrac12-10^{-3}$, whose $W_{\min}=-1/9$ is likewise exact; and an explicit $17$-ontic-state model for the strange-state fragment at $p=0$ (sixteen ontic states of weight $\tfrac16$ and one of weight $\tfrac13$), whose marginals hold as exact rationals.
\end{itemize}

\emph{The depth-first search of \cref{thm:twocopy}.} The search is exhaustive rather than heuristic; one pruning rule makes $9^{10}\approx3.5\times10^9$ assignments tractable. Scale the $k$th image's pair Wigner values to integers $V_k(u)$ and write $f_k(m,j)=\sum_{u\in\ell_{m,j}}V_k(u)$ for the outcome probabilities they induce, so that $\sum_u V_k(u)$ is the scaled unit. Admissibility of $\lambda$ against image $k$, which is $s_\lambda\ge1$ here since $(k{-}1)/d=1$ at ten directions on nine levels, is then
\begin{equation}\label{eq:prune}
\sum_{m=1}^{10}\Big[\max_j f_k(m,j)-f_k(m,\lambda_m)\Big]\ \le\ \mathrm{cap}_k,
\end{equation}
with $\mathrm{cap}_k=\sum_m\max_j f_k(m,j)-\sum_u V_k(u)$. Every bracket is a non-negative integer, so a partial sum over the directions already fixed can only grow as the branch is extended. A node whose partial sum exceeds $\mathrm{cap}_k$ for even one $k$ therefore has no admissible completion, and cutting its subtree loses nothing: the search is a sound and complete enumeration, not a sample. Translations act transitively on the outcomes of any one direction, so $\lambda_1$ is held fixed and the solution count multiplied by $9$; the $x=0$ run visits $19\,710\,874$ nodes in that reduced form and $177\,397\,866$ without it, exactly nine times as many. For the states whose pair Wigner values are not rational the same search runs on $30$-digit reals and prunes only when the excess exceeds a tolerance, which over-approximates: the set it returns provably contains every admissible assignment, so the negative conclusions it supports are rigorous.

\emph{Machine-precision results.} The remaining results are computed in machine or extended precision, each checked against an exact reference where one exists. \Cref{thm:qubitppm}'s formula is checked against the direct fragment LP on $60$ random Bloch vectors (agreement to $10^{-12}$); a seeded search over $400$ random single-qutrit states confirms the trivial single-qudit monotonicity of \cref{prop:nu}(iii); and the scan of the $d=5$ Wigner-maximiser sphere in \cref{sec:robust} runs the orbit clauses over the family, its minimum matching the exact $\sqrt5/3$ of the member $\ket b$ to $10^{-8}$. The exhaustive half of the $45$-pair census of \cref{sec:twoqutrit} checks the double exclusions at machine precision, with margins of order $10^{-16}$. The prepare-and-measure column of \cref{tab:six} is recomputed by the general fragment LP and agrees with the exact $d=3$ and $d=5$ values to $10^{-9}$; the same LP gives the restricted-measurement results of \cref{sec:disagree}, where every subset of each size is enumerated. The $d=5$ orbit decomposition is rebuilt independently ($11$ orbits, matching sizes) and every per-orbit minimum certified against its exact value at $35$-digit precision. The $d=7$ enumeration processes all $7^{8}\approx5.8\times10^{6}$ assignments, returning $360$ orbits with the size histogram and critical representative that reproduce the exact cubic of \cref{sec:robust}. The four- and five-MUB two-qutrit fragment thresholds of \cref{sec:twoqutrit} are computed for every $4$- and $5$-subset of the ten field MUBs. They are not constant on the $\mathrm{ASL}(2,9)$ orbits of subsets, since the pair state is not invariant under the Cliffords realising those permutations, so each quoted value names its own fragment; the quoted values are verified on both sides by solver-independent certificates: an explicit feasible support just above each bracket, and a Farkas dual just below, whose entries are integer multiples of $\tfrac16$ and for which a total-mass bound makes the infeasibility insensitive to rounding. 

\emph{Quantities that remain purely numerical} are stated as such in the text, and all of them are seeded so that they are reproducible: the Monte Carlo survey of \cref{app:survey} and \cref{fig:survey}; the extra-basis feasibility of \cref{sec:allmeas}; the central-ray comparison in the proof of \cref{eq:pstab} below; the $\pC$ and mana ranges over the $d=5$ sphere in \cref{sec:robust}; the $2\times10^5$-state search of \cref{sec:boundmagic}, whose optimum is then certified exactly on the family it lands in; and the search behind \cref{prop:nu}'s single-qudit monotonicity. A Wolfram Language notebook available at \cite{BozzettoHance2026Notebook} reproduces the benchmark tables and the algebraic quantities of the text: every entry of \cref{tab:six} in both dimensions, including the $\pPM$ column, with matched exact rational feasibility and Farkas certificates at its two non-zero crossings; the orbit censuses at $d=3,5$ and, in even characteristic, at $q=2,4$; the rank computations, the forcing of the collinear weights to $x=\tfrac13$ and the Weyl-versus-Clifford witness behind \cref{thm:convexlinear} at $d=3$, together with its qubit instance; the Gram spectra of \cref{lem:facet} at every odd prime $d\le13$ and from the explicit MUB projectors at $d=3,5$; the bound-magic binding assignments, their degeneracies, and the minimal polynomials of the binding scores and violations; the thresholds of \cref{tab:noises} and the $\ket b$ dephasing entries; the peak-negativity identities of \cref{prop:nu}; the $\mathbb{F}_3$-symplectic combinatorics behind the remark on Ref.~\cite{delfosse2017}; the $\pC$ witness of \cref{app:experiment} at $d=3$; and a seeded generator for the survey of \cref{app:survey}, with the sample size a parameter whose default is small. The six heavy enumerations are in it too, each as a section holding its batch script verbatim and running it: the two-copy enumeration of \cref{thm:twocopy}, the $d=5$ Farkas certificate behind $\pPM\ge0.5378$, the polytope-entry stabiliser-mixture decompositions, the $d=7$ enumeration of all $7^{8}$ assignments, the Formulation F1 programs, and the four-MUB thresholds of \cref{sec:twoqutrit} over all $90$ two-product subsets and all twelve stabiliser states $\sigma$. Those sections, and five shorter ones, are guarded by a switch set to \texttt{False} on delivery, so that evaluating the notebook top to bottom takes a couple of minutes and the long runs are opt-in; the notebook lists them with their measured runtimes, of which the longest are the two-copy enumeration at about half an hour and the four-MUB scan at about three and a half hours on ten subkernels.

\end{appendices}

\end{document}